\documentclass[onefignum,onetabnum]{siamart190516}

\usepackage{lipsum}
\usepackage{amsfonts}

\usepackage{graphicx,epstopdf} 
\usepackage[caption=false]{subfig} 

\usepackage{algorithmic}
\usepackage{booktabs}
\usepackage{enumitem}
\usepackage{mathrsfs}  
\usepackage{amssymb}
\usepackage{wrapfig}
\ifpdf
  \DeclareGraphicsExtensions{.eps,.pdf,.png,.jpg}
\else
  \DeclareGraphicsExtensions{.eps}
\fi

\newsiamremark{remark}{Remark}
\newsiamremark{hypothesis}{Hypothesis}
\crefname{hypothesis}{Hypothesis}{Hypotheses}
\newsiamthm{claim}{Claim}
\newsiamthm{conjecture}{Conjecture}
\newsiamthm{observation}{Observation}

\headers{Hypergraph Cuts with General Splitting Functions}{Nate Veldt, Austin R. Benson, and Jon Kleinberg}

\title{Hypergraph Cuts with General Splitting Functions
\thanks{This research was supported by NSF Award DMS-1830274, ARO Award W911NF19-1-0057, a Simons Investigator grant, a Vannevar Bush Faculty Fellowship, and ARO MURI.}}

\author{Nate Veldt\thanks{Center for Applied Mathematics, Cornell University
  (\email{nveldt@cornell.edu}).}
\and Austin R. Benson\thanks{Department of Computer Science, Cornell University
  (\email{arb@cs.cornell.edu}).}
\and Jon Kleinberg\thanks{Department of Computer Science, Cornell University (\email{kleinberg@cornell.edu})}}

\usepackage{mathtools}

\DeclarePairedDelimiter\floor{\lfloor}{\rfloor}

\newcommand{\cut}{\textbf{cut}}

\providecommand{\mat}[1]{\boldsymbol{\mathrm{#1}}}%
\renewcommand{\vec}[1]{\boldsymbol{\mathrm{#1}}}

\newcommand{\mcut}{\textsc{Maxcut}}

\newcommand{\hstgen}{\textsc{Gen-Hyper-st-Cut}}
\newcommand{\nesst}{\textsc{NES-HyperCut}}
\newcommand{\dhstgen}{\textsc{Degenerate-HyperCut}}

\newcommand{\nmc}{\textsc{Node-Weighted-MC}}	
\newcommand{\hc}{\textsc{Hyper-st-Cut}}

\newcommand{\nnhc}{Needy-Node \hc{}}
\newcommand{\ghmc}{\textsc{Gen-HyperMC}}
\newcommand{\hmc}{\textsc{HyperMC}}

\newcommand{\hmp}{\textsc{HyperMP}}
\newcommand{\smp}{\textsc{Sub-MP}}
\newcommand{\turh}{\textsc{Rainbow-HyperMC}}

\DeclareMathOperator*{\minimize}{minimize}

\DeclareMathOperator{\argmax}{argmax}
\DeclareMathOperator{\argmin}{argmin}

\providecommand{\mA}{\ensuremath{\mat{A}}}

\providecommand{\vc}{\ensuremath{\vec{c}}}

\providecommand{\vm}{\ensuremath{\vec{m}}}

\providecommand{\vp}{\ensuremath{\vec{p}}}

\providecommand{\vw}{\ensuremath{\vec{w}}}

\providecommand{\vy}{\ensuremath{\vec{y}}}
\providecommand{\vz}{\ensuremath{\vec{z}}}

\newcommand{\hw}{\hat{w}}
\newcommand{\hp}{\hat{p}}
\newcommand{\hm}{\hat{m}}
\newcommand{\hy}{\hat{y}}

\usepackage{xparse}

\ExplSyntaxOn
\NewDocumentCommand{\rvect}{m}
{
	\seq_set_split:Nnn \l_tmpa_seq { , } { #1 }
	\begin{bmatrix}
		\seq_use:Nn \l_tmpa_seq { & }
	\end{bmatrix}
}
\ExplSyntaxOff

\usepackage{stackengine}

\DeclareMathOperator*{\minimum}{minimum}

\newcommand{\hide}[1]{}

\begin{document}

\maketitle
\begin{abstract}
The minimum $s$-$t$ cut problem in graphs is one of the most fundamental problems in combinatorial optimization, and graph cuts underlie algorithms throughout discrete mathematics, theoretical computer science, operations research, and data science.
While graphs are a standard model for pairwise relationships, hypergraphs provide the flexibility to model multi-way relationships, and are
now a standard model for complex data and systems.
However, when generalizing from graphs to hypergraphs, the notion of a ``cut hyperedge'' is less clear, as a hyperedge's nodes can be \emph{split} in several ways.
Here, we develop a framework for hypergraph cuts by considering the problem of separating two terminal nodes in a hypergraph in a way that minimizes a sum of penalties at split hyperedges.
In our setup, different ways of splitting the same hyperedge have different penalties, and the penalty is encoded by what we call a \emph{splitting function}.

Our framework opens a rich space on the foundations of hypergraph cuts.
We first identify a natural class of \emph{cardinality-based} hyperedge splitting functions that depend only on the number of nodes on each side of the split. In this case, we show that the general hypergraph $s$-$t$ cut problem can be reduced to a tractable graph $s$-$t$ cut problem if and only if the splitting functions are submodular.
We also identify a wide regime of non-submodular splitting functions for which the problem is NP-hard.
We also analyze extensions to multiway cuts with at least three terminal nodes
and identify a natural class of splitting functions for which the problem can be reduced in an approximation-preserving way to the node-weighted multiway cut problem in graphs, again subject to a submodularity property.
Finally, we outline several open questions on general hypergraph cut problems.
\end{abstract}



\tableofcontents
\section{Introduction}
Graphs have long been a popular model for analyzing systems of interconnected objects~\cite{albert_barabasi,bertozzi2016hd,Easley:2010:NCM:1805895,newman2003sirev}. A standard primitive in graph analysis is the concept of a \emph{cut} edge, i.e., an edge whose two endpoints are separated, in any task that involves arranging the nodes of a graph into different subsets or {clusters}. One of the most fundamental problems in graph theory and combinatorial optimization is the minimum $s$-$t$ cut problem, which seeks a minimum weight set of edges to cut in order to separate two distinguished nodes ($s$ and $t$) from each other in a graph. The dual of finding a minimum $s$-$t$ cut is the well-known maximum $s$-$t$ flow problem, which seeks to route a maximum amount of flow along edges from $s$ to $t$, subject to edge capacity and conservation of flow constraints.
Classical methods for finding minimum cuts and maximum flows~\cite{dinic1970algorithm,edmonds1972theoretical,Ford_fulkerson,Goldberg:1986:NAM:12130.12144} are included in numerous textbooks in mathematics, operations research, and computer science, and are considered standard material for nearly any course on algorithms. There is a long history in algorithm design that continues into modern times, which includes both exact algorithms~\cite{King:1992:FDM:139404.139438,Orlin:2013:MFO:2488608.2488705} and various fast solver approaches for computing approximate solutions~\cite{kelner_soda,peng_approxmaxflow,sherman2013nearlymax}. In data science, maximum flows and minimum cuts are extensively used as subroutines in a variety of machine learning and clustering algorithms~\cite{Andersen:2008:AIG:1347082.1347154,Blum:2001:LLU:645530.757779,boykov2004mincut,heuer_et_al:LIPIcs:2018:8936,LangRao2004,Orecchia:2014:FAL:2634074.2634168,veldt16simple}.


While graphs provide a useful way to model pairwise relationships, many complex systems and datasets are characterized by higher-order relationships that are better modeled by hypergraphs~\cite{BensonE11221,Estrada2006hyper,Lambiotte2019higher,porter2019nonlinearity,Selvakkumaran-2006-multiobjective}.
For example, in scientific computing, nodes may represent rows in a sparse matrix and hyperedges encode the nonzero patterns of each column~\cite{Ballard:2016:HPS:3012407.3015144}. In machine learning classification tasks, a hyperedge can represent evidence that a set of objects in a dataset belongs to the same cluster or should all be associated with the same label~\cite{Zhou2006learning,yadati2019hypergcn}. And in VLSI layout and circuit design, nodes and hyperedges model transistors and signals in digital circuits~\cite{alpert1995survey,lawler1973}. Higher-order relationships are also inherent in nature and society, and hypergraphs model
multi-way relationships between organisms in food webs~\cite{BensonGleichLeskovec2016,panli2017inhomogeneous},
human dynamics and behavior~\cite{Benson2019three,neuhauser2019multi},
and various joint biological interactions~\cite{Bassett-2014-cross,michoel2012molecular,Tian:2009:HLA:1666837.1666865}.

Because hyperedges may contain more than two nodes, the definition of a \emph{cut} hyperedge is much more nuanced than that of a cut edge.
In this manuscript, we explore broad notions of what it means to cut or \emph{split} a hyperedge, along with \emph{when} and \emph{how} hypergraph cut problems can be minimized efficiently in practice.
Along the way, we develop our computational framework to
characterize the computational complexity of the hypergraph $s$-$t$ cut problem under these generalized definitions,
unify a number of seemingly disparate techniques for solving hypergraph cut problems via graph reductions,
and derive a number of open questions.
Given the importance of the minimum $s$-$t$ cut problem, we hope that our framework will bring substantial value for research and applications using higher-order models and data.
We next provide additional background and preview our ideas in the context of previous research.

\subsection{Hypergraph Cut Problems}
Due to broad modeling capability, hypergraph generalizations of numerous graph cut objectives are applied in practice~\cite{Agarwal2005beyond,BensonGleichLeskovec2016,Chandrasekara_hyperkcut,Chekuri:2016:SFR:2884435.2884492,lawler1973,Zhou2006learning}. Hypergraph cut problems are based on minimizing the sum of nonnegative penalties at hyperedges, where the penalty at each hyperedge is determined by how the hyperedge's constituent nodes are split. The standard penalty function is \emph{all-or-nothing}, which assigns a penalty of zero if the nodes are all together, but assigns a penalty equal to the weight of the hyperedge if the nodes are split or arranged in any other way. This all-or-nothing penalty function is one natural generalization of the cut function in graphs, since in graphs there is only one possible nonzero cut penalty for an edge, based on whether two nodes of an edge are separated or placed together.

Specifically for the minimum $s$-$t$ cut problem, Lawler proved nearly half a century ago that hypergraph problems under the all-or-nothing penalty function can be solved in polynomial time via reduction to a minimum $s$-$t$ cut problem in a directed graph~\cite{lawler1973}. Lawler's results, along with many other advances in algorithms for hypergraph cut problems, have been widely applied to VLSI layout and circuit design~\cite{hadley1995,hadley1992efficient,karypis1999multilevel,vannelli1990Gomoryhu,yangwong1996}, scientific computing applications such as sparse matrix partitioning~\cite{akbudak2014simultaneous,akbudak2013cachelocality,Ballard:2016:HPS:3012407.3015144,catalyurek2011hypergraph,catalyurek2010ontwo,Cevahir2010,Kayaaslan:2012:PHS:2340033.2340060,ucar2007revisiting}, and computer vision problems such as image or video segmentation~\cite{Agarwal2005beyond,govindu2005,huang2009videoobject,kim2011highcc,ochs2012motion,purkait2017largehyperedges}. Hypergraph cut problems also arise in a variety of other applications, including semi-supervised learning~\cite{Hein2013,Zhou2006learning}, consensus clustering~\cite{yaros2013imbalanced}, document clustering~\cite{Hu:2008:HPD:1390334.1390548}, and bioinformatics~\cite{michoel2012molecular,tian2009gene}.

\subsection{Hyperedge Splitting Functions}
The number of ways to split an $r$-node hyperedge is exponential in $r$, even in the case of only two clusters.
In our work, we formally define the notion of a hyperedge splitting function, which assigns a penalty to each configuration of a hyperedge's nodes. Despite there being an exponential number of ways to split a hyperedge, nearly all previous work on hypergraph cut problems focuses on the {all-or-nothing} splitting function. In practical applications, however, we would expect there to be a significant difference between distinct ways of splitting a hyperedge. Consider, for example, a hypergraph-based classification or clustering task in which hyperedges represent \emph{evidence} that a certain set of nodes should be clustered together or classified in the same way. If a large hyperedge is split in such a way that all nodes but one are placed in the same cluster, this assignment would \emph{mostly} agree with the implicit evidence that all nodes should be clustered together. However, if an all-or-nothing splitting function is applied, such a clustering would be penalized just as strongly as splitting the hyperedge into two equal-sized groups, or even placing each node in a distinct cluster. As another example, consider a scientific computing application in which nodes represent data objects (e.g., nonzero row entries in a column of a sparse matrix), and hyperedges indicate computational tasks that rely on subsets of data blocks (e.g., operations for sparse matrix or sparse vector multiplication). If clustering the hypergraph corresponds to partitioning the data and the computational workload among computer processors, a cut hyperedge represents a need for communication between processors. Separating a hyperedge across multiple processors would naturally lead to a higher communication cost than if all but a few nodes in the hyperedge were contained in the same cluster. For this reason, minimizing an all-or-nothing splitting function does not map well to the ultimate goal of minimizing communication in parallel computations.

Despite these and other motivating examples, the existing literature on hypergraph cut problems with splitting functions that are \emph{not} all-or-nothing is small and fragmented.
Within hypergraph partitioning, there are a few other penalty functions designed specifically for multiway splittings, i.e., penalties that are applied when a hyperedge is split into two or more clusters. These include the absorption~\cite{sun1995absorption}, the sum of external degrees~\cite{alpert1995survey,Karypis:1999:MKW:309847.309954,yaros2013imbalanced}, the $K-1$ cut~\cite{gong_km1}, and the discount cut~\cite{yaros2013imbalanced}, all of which can be viewed as multiway splitting functions that assign penalties based on the \emph{number} of clusters spanned by a hyperedge. However, aside from the discount cut, these splitting functions do not consider \emph{how many} nodes of a hyperedge are in each cluster, and they reduce to the all-or-nothing function when only two clusters are formed. Unrelated to these, Li and Milenkovic introduced the very general notion of inhomogenous hypergraph clustering~\cite{panli2017inhomogeneous} and later considered the special case of submodular hypergraph clustering~\cite{panli_submodular}. For both problems, a different penalty can be associated with \emph{each} possible bipartition of a hypergedge. In the most general setting, this means the number of penalties for splitting a hyperedge can be exponential in the hyperedge size. Their analysis is restricted to the case where only two clusters are formed, and their primary focus is to develop approximation algorithms for NP-hard ratio-cut objectives, by approximating the hypergraph with a graph. Meanwhile, for the minimum $s$-$t$ cut problem, although a polynomial time algorithm for the all-or-nothing splitting function has been known since the work of Lawler~\cite{lawler1973}, there are no known complexity results for solving the problem under other splitting functions.

\subsection{Graph Reduction Techniques}
A common approach for solving hypergraph cut problems is to first \emph{reduce} the hypergraph to a graph problem by \emph{expanding} each hyperedge into a small graph. After concatenating expanded hyperedges into a larger graph, a standard graph algorithm can be applied. By far the most popular expansion techniques are the clique expansion~\cite{Agarwal2006holearning,BensonGleichLeskovec2016,hadley1995,hadley1992efficient,Zhou2006learning,panli2017inhomogeneous,vannelli1990Gomoryhu,zien1999} and the star expansion~\cite{Agarwal2006holearning,heuer_et_al:LIPIcs:2018:8936,hu1985multiterminal, zien1999}. As their names imply, these operate by replacing each hyperedge with a small (possibly weighted) clique or star graph. These two expansions are frequently used in hypergraph learning and hypergraph partitioning~\cite{Agarwal2006holearning,Agarwal2005beyond,Hein2013,Huang2015,Zhou2006learning}. However, in nearly all cases, the graphs obtained by applying existing hyperedge expansion techniques only approximately model the original hypergraph in terms of cut properties. In fact, Ihler et al.\ showed that a class of expansion techniques that includes the star and the clique expansion cannot exactly model the all-or-nothing splitting function, even when additional auxiliary vertices are used~\cite{ihler1993modeling}.

These results may tempt one to conclude that representing hypergraph cut problems as graph cut problems is bound to fail. However, Lawler's expansion technique for exactly solving the all-or-nothing hypergraph $s$-$t$ cut problem proves that this is not the case~\cite{lawler1973}. The reason is that Lawler's approach converts hyperedges into a \emph{directed} graph on a larger vertex set, whereas Ihler et al.\ restrict to hyperedge expansions involving only \emph{undirected} edges and auxiliary nodes. Lawler's polynomial time solution for the all-or-nothing hypergraph $s$-$t$ cut problem leads to several fundamental open questions. Is there a broader class of hyperedge expansion techniques that can \emph{exactly} model hypergraph cut problems via graph reduction? Given the ``easy'' solution for the all-or-nothing hypergraph $s$-$t$ cut problem, are there polynomial time solutions when we consider more general splitting functions? Is the problem ever in fact hard to solve, or impossible to represent as a graph? Finally, what can we say about \emph{multiterminal} generalizations of the hypergraph $s$-$t$ cut objective?

\subsection{Our contributions}
This paper undertakes a systematic study of the hypergraph minimum $s$-$t$ cut problem for general hyperedge splitting functions. Our work includes broad contributions at the intersection of graph theory, optimization, scientific computing, and other subdisciplines in applied mathematics. We identify new polynomial time algorithms for certain variants of the hypergraph $s$-$t$ cut problem and NP-hardness results for others.  We also provide a new set of tools for modeling and analyzing higher-order interactions, which we expect will be broadly useful for machine learning and data mining applications. More specifically, we present a unified framework for \emph{exactly} modeling hypergraph cut problems via standard graph cut problems. Our results on multiway generalizations of the hypergraph $s$-$t$ cut problem are also closely related to previous results from the theory and discrete algorithms community on multiway cut and hypergraph labeling problems.
Finally, our framework also leads to several new clear-cut complexity questions in discrete algorithms and theoretical computer science.
We summarize our main contributions and outline the remainder of our paper below.

\subsubsection*{Cardinality-Based Splitting Functions}
We analyze a natural class of hyperedge splitting functions that we call \emph{cardinality-based}. These assign penalties based only on the \emph{number} of nodes of a hyperedge that are placed on either side of a two-way split. Cardinality-based functions are particularly relevant given that in most applications, a node's name or identity is not expected to affect the overall quality of a cut. Nearly all splitting functions used in practice are either cardinality-based, or can be viewed as multiway generalizations of a cardinality-based function~\cite{alpert1995survey,gong_km1,Karypis:1999:MKW:309847.309954,panli_submodular,sun1995absorption,yaros2013imbalanced}.

Figure~\ref{fig:FourUCard} shows how different cardinality-based splitting functions lead to different optimal solutions for the hypergraph minimum $s$-$t$ cut problem on a toy 4-uniform hypergraph. In Section~\ref{sec:experiments}, we conduct related experiments on hypergraphs constructed from real data, where nodes represent tags used in an online mathematics forum, and hyperedges represent sets of tags used in the same post. Figure~\ref{fig:jaccard1} previews these results and shows the Jaccard similarity between the \emph{all-or-nothing} $s$-$t$ cut solution and the solution to $s$-$t$ cut problems with different cardinality-based splitting functions. The changes in similarity show that different splitting functions can lead to a range of different types of solutions for the hypergraph $s$-$t$ cut problem in practice.

\begin{figure}[t]
	\centering
	\subfloat[$w_2 = 0.5$ solution]{\label{fig:4uc1}\includegraphics[width=.3\linewidth]{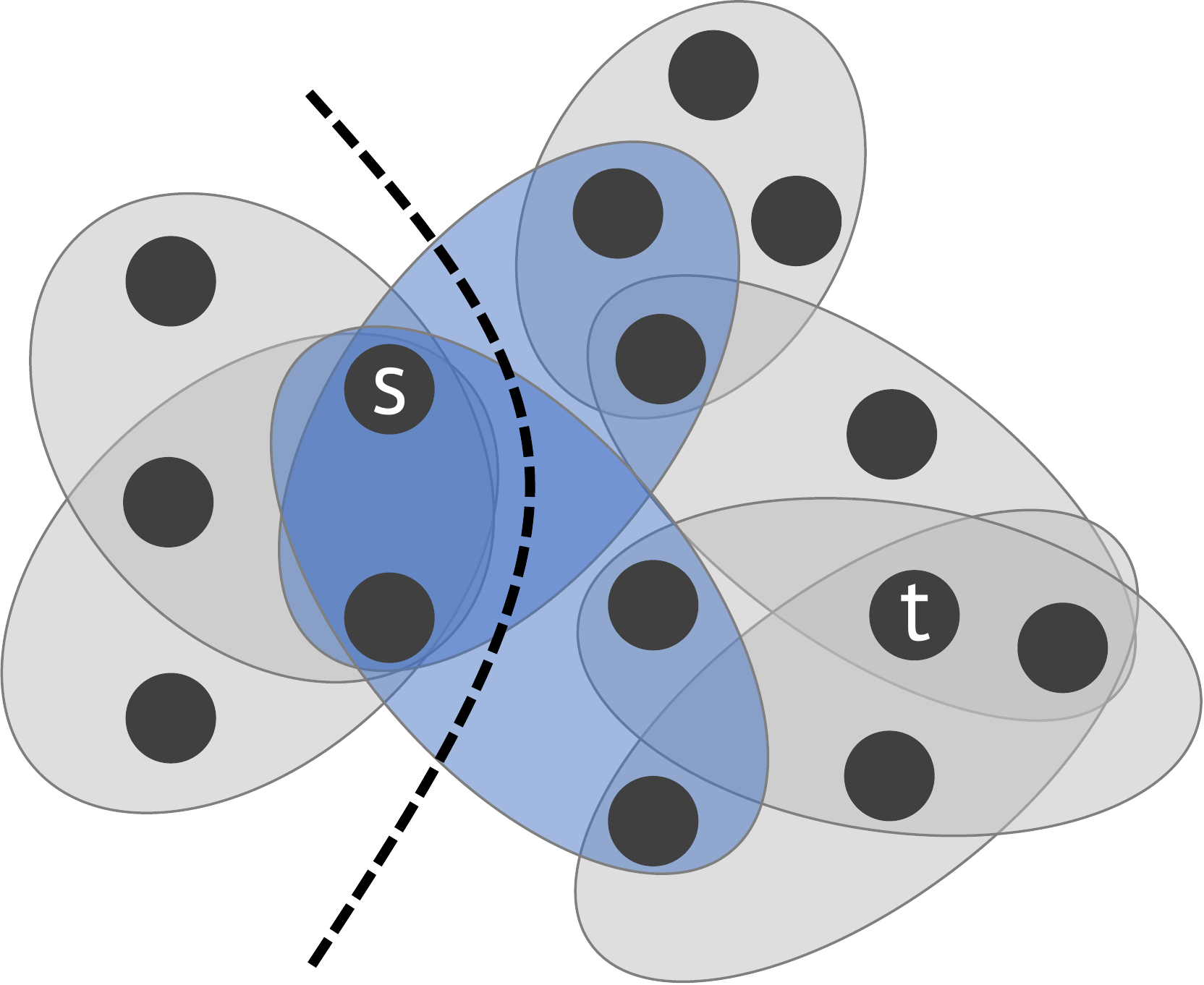}}\hfill
	\subfloat[$w_2 = 1.5$ solution]
	{\label{fig:4uc2} \includegraphics[width=.3\linewidth]{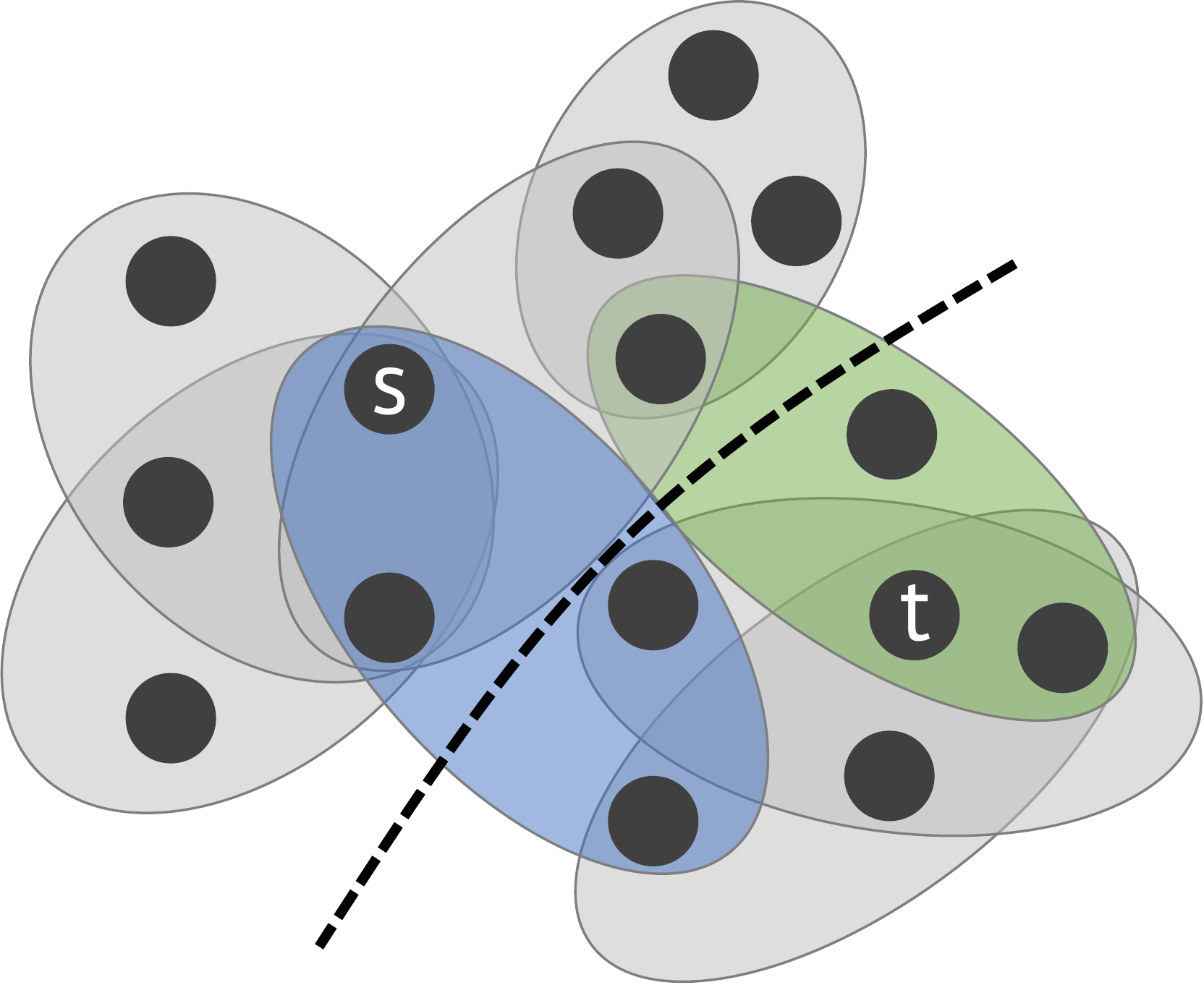}}\hfill
	\subfloat[$w_2 = 2.5$ solution]
	{\label{fig:4uc3}\includegraphics[width=.3\linewidth]{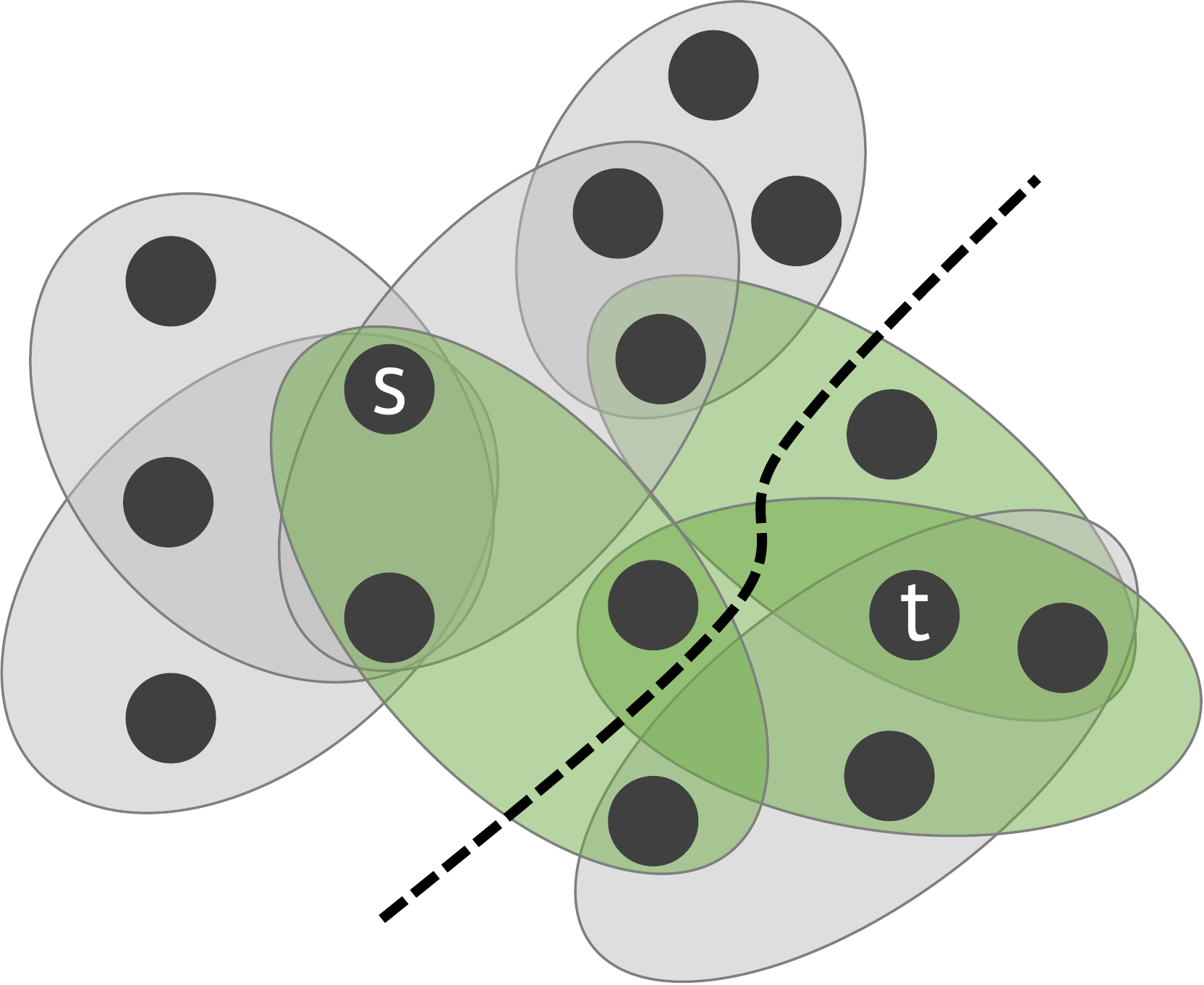}}
	\caption{For four-node hyperedges, cardinality-based splitting functions assign a penalty of $w_1$ if one node is separated from the other three, and a penalty of $w_2$ if the hyperedge has a 2-2 split. Here we show the minimum $s$-$t$ cut solution on a small hypergraph when $w_1 = 1$ is fixed and $w_2$ takes on three different values. Grey indicates uncut hyperedges, blue indicates 2-2 splits, and green indicates 1-3 splits. Solutions for $w_2 \in \{0.5, 1.5\}$ are unique; for $w_2 = 2.5$ we illustrate a solution with a minimum number of source-side nodes.
	Among other results, we prove that for 4-uniform hypergraphs, this problem is NP-hard for case (a) ($w_2 < w_1$, Theorem~\ref{thm:4maxcut}), is tractable for case (b) ($w_2 \in [w_1, 2w_1]$, Theorem~\ref{thm:iffsub}), and has unknown complexity in case (c) ($w_2 > 2w_1$, Figure~\ref{fig:tractable2}). 
}
	\label{fig:FourUCard}
\end{figure}
\begin{figure}[tbhp]
	\centering
		\centering
	\subfloat[Hyperedges containing ``hypergraphs'' tag\label{fig:math1}]
	{\includegraphics[width=.475\linewidth]{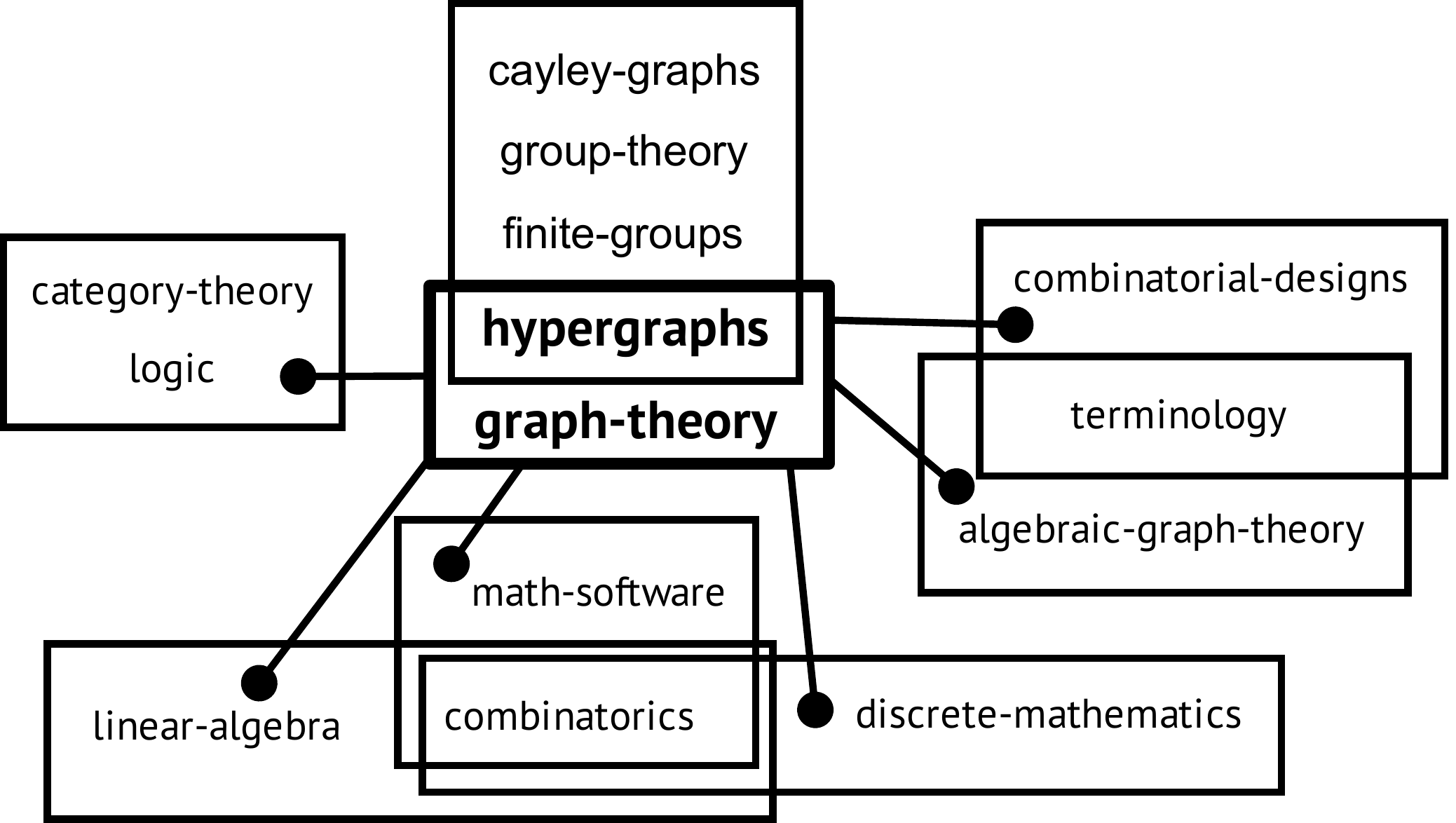}}\hfill
	\subfloat[Jaccard similarity scores\label{fig:math2}]
	{\includegraphics[width=.4\linewidth]{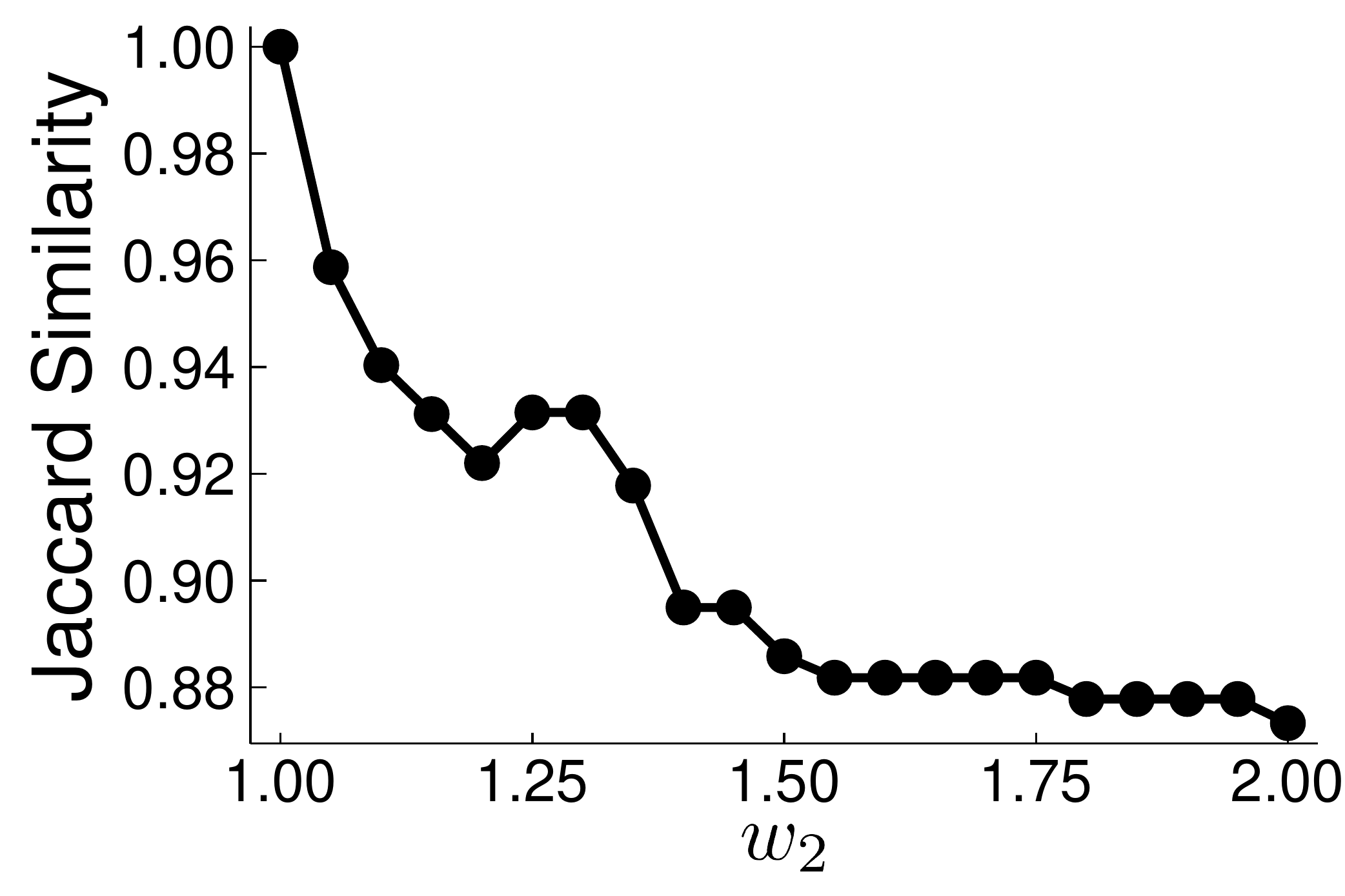}}
	\caption{Our new techniques enable us to find minimum $s$-$t$ cuts of a real-world hypergraph under generalized notions of hypergraph cuts. The hypergraph here is constructed from questions posted in an online math forum~(\url{https://math.stackexchange.com/}). Each node is a tag for a math term (e.g., ``logic'', ``combinatorics''), and hyperedges represent 2 to 5 tags from a single post.
          (a)~All hyperedges in the dataset that have four nodes and contain the tag ``hypergraphs.'' Many overlap and all but one also contain the ``graph-theory'' tag.
          (b)~For {cardinality-based} splitting functions, separating one node by itself costs $w_1 = 1$ and splits with two nodes together have penalty $w_2$.
          There are substantial changes in the Jaccard similarity between the solution for the problem under the standard all-or-nothing splitting function ($w_2 = w_1 = 1$)
          and the solution for the problem when $w_2 > 1$. Section~\ref{sec:experiments} has more details on these experiments.}
	\label{fig:jaccard1}
\end{figure}

\subsubsection*{Positive results via graph reduction}
One of our primary contributions is a generalized framework for hyperedge expansion techniques, which enables us to completely characterize which cardinality-based hypergraph $s$-$t$ cut problems can be solved via reduction to a graph $s$-$t$ cut problem. Specifically, we show that for cardinality-based splitting functions, the hypergraph $s$-$t$ cut problem is reducible to a graph $s$-$t$ cut problem if and only if the splitting functions are submodular. 

\subsubsection*{NP-hard regimes for hypergraph $s$-$t$ cuts}
Next we identify a large class of cardinality-based~$s$-$t$ cut problems \emph{outside} the submodular region for which the $s$-$t$ cut problem is NP-hard to solve. This is somewhat surprising, given that the hypergraph $s$-$t$ cut problem has been viewed as an ``easy'' problem since the work of Lawler~\cite{lawler1973}, but again, Lawler only considered the all-or-nothing splitting function.
At the same time, we give a simple example of a cardinality-based splitting function outside the submodular region for which the problem is still tractable, even if not graph reducible. This rules out the possibility that tractability is exactly determined by submodularity of splitting functions and leads to a number of interesting open questions on the limits between tractable and intractable instances of the hypergraph $s$-$t$ cut problem. We also show that there are variants of the problem that are NP-hard to approximate to within any multiplicative factor.

\subsubsection*{Generalizations to multiway hypergraph cuts}
Finally, we consider extensions to the multiterminal case, where the goal is to form a $k$-clustering of a hypergraph so that $k$ distinguished terminal nodes are all separated from each other. This problem is NP-hard to optimize for $k >2$ even in the graph setting~\cite{Dahlhaus94thecomplexity}. However, by generalizing our hyperedge expansion techniques, we prove that for a multiway generalization of submodular cardinality-based splitting functions, the hypergraph multiway cut problem can be modeled by a node-weighted multiway graph cut problem. Thus, we can apply existing approximation algorithms for the latter problem~\cite{garg1994multiway}. We also show a special case of hypergraph multiway cut that is NP-hard to approximate.

\section{Preliminaries and Additional Related Work}
\label{sec:prelims}
We present several technical preliminaries and related work on graph and hypergraph cut problems. Our work also shares many connections with minimizing submodular functions; we give an overview of relevant related work in Section~\ref{sec:submodular}.

\subsection{Graph Cut Problems}
Let ${G} = (V,E)$ be a weighted and possibly directed graph. A \emph{partition} of a graph, which we will also call a \emph{clustering}, is a separation of its nodes into \emph{clusters} $\{V_1, V_2, \hdots, V_k\}$ such that $V_i \cap V_j = \emptyset$ if $i \neq j$, and $\bigcup_{i = 1}^k V_i = V$. A \emph{bipartition} is simply a partition with two clusters, corresponding to a set $S \subset V$ and its complement $\bar{S}  = V \backslash S$. 
For an edge $(i,j) \in E$, let $w_{ij}$ denote the (nonnegative) weight of the edge. Given any two node sets $S \subset V$ and $T \subset V$, we define the cut between $S$ and $T$ to be
\[ \cut(S,T) = \sum_{u \in S} \sum_{v \in T} w_{uv} \,.\]
Note that if the graph is undirected, $\cut(S,T) = \cut(T,S)$, though this need not be the case for directed graphs. A single edge is said to be \emph{cut} by $S$ if only one of its endpoints is in $S$. This is associated with a penalty equal to the edge's weight. In this way, every set of nodes $S \subset V$ is associated with a cut score, obtained by summing the weight of edges crossing between $S$ and its complement:
\[ \cut(S) = \cut(S, \bar{S}) = \sum_{i\in S, j \in \bar{S}} w_{ij}.\]

\subsubsection*{Minimum cuts and maximum flows}
If we identify two special \emph{terminal} nodes $s$ and $t$ in $V$ (called the \emph{source} and \emph{sink} respectively), then the minimum $s$-$t$ cut problem is the task of removing a minimum-weight set of edges in $E$ so that there are no paths from $s$ to $t$. Formally, this can be written:
\begin{equation}
\label{eq:minstcut}
\begin{array}{ll}
\minimize_{S \subset V} & \cut(S, \bar{S}) \\
\text{subject to } & s \in S, t \in \bar{S}.
\end{array}
\end{equation}
There is a subtle but important difference between minimizing this objective over directed or undirected graphs. In undirected graphs, if set $S$ contains the source $s$, the cut penalty includes any edges with one endpoint in $S$ and the other in $\bar{S}$. The directed minimum $s$-$t$ cut problem is equivalent to removing a minimum weight set of edges so that there is no \emph{directed} path from $s$ to $t$. Thus, we do not incur penalties for edges crossing from $\bar{S}$ to $S$, regardless of their weight. This distinction will be important for our results on reducing hypergraph cut problems to graph cut problems. 

The dual of finding a minimum $s$-$t$ cut in a graph is the well-known maximum $s$-$t$ flow problem. Classical algorithms for computing maximum $s$-$t$ flows are among the most well-known combinatorial graph algorithms~\cite{dinic1970algorithm,edmonds1972theoretical,Ford_fulkerson,Goldberg:1986:NAM:12130.12144}, and there is an extensive body of research dedicated to finding improved flow algorithms~\cite{kelner_soda,King:1992:FDM:139404.139438,Orlin:2013:MFO:2488608.2488705,peng_approxmaxflow,sherman2013nearlymax}. Although minimum $s$-$t$ cuts and maximum $s$-$t$ flows are intimately related, all of our results can be understood entirely in terms of cuts, so we focus on a cut-based view in the remainder of the paper.

\subsubsection*{The multiway cut problem}
The multiway cut problem is a generalization of the \emph{undirected} minimum $s$-$t$ cut problem in graphs. An input of the problem is a graph $G = (V,E)$ along with a set of $k > 2$ terminal nodes $\{t_1, t_2, \hdots , t_k\} \subset V$. The goal is to find a $k$-clustering $\{V_1, V_2, \hdots, V_k\}$ of $V$, with $t_i \in V_i$ for $i \in\{ 1, 2, \hdots, k\}$, such that the total weight of edges crossing between different clusters is minimized. The problem is NP-hard when $k \geq 3$, but can be approximated to within a factor $2(1-\frac{1}{k})$ by combining results from solving $k$ minimum $s$-$t$ cut problems~\cite{Dahlhaus94thecomplexity}. There are several generalizations of the problem, with many relating to hypergraphs~\cite{chekuri2011,Chekuri:2016:SFR:2884435.2884492,Ene:2013:LDS:2627817.2627840,garg1994multiway,Okumoto2012,Sharma:2014:MCP:2591796.2591866,Zhao2005}, and we consider these in greater depth in Section~\ref{sec:multiway}.

\subsection{Hypergraph Cut Problems}
\label{sec:hypergraph2cut}
Let $\mathcal{H} = (V,E)$ denote a hypergraph, where  $e \in E$ is a hyperedge, i.e., a set of (possibly more than two) vertices from $V$, that may be associated with a weight $w_e$. Although there are directed notions of hypergraphs, we restrict our attention to undirected hypergraphs. For a bipartition $\{S, \bar{S}\}$ of $V$, a hyperedge $e$ is \emph{cut} if it has at least one node in both $S$ and $\bar{S}$. We represent the set of cut hyperedges by
\begin{equation}
\label{eq:cutset}
\partial S = \{ e \in E : e \cap S \neq \emptyset \text{ and } e \cap \bar{S} \neq \emptyset \}.
\end{equation}
A simple and widely-studied hypergraph generalization of the graph cut is the sum of weights of cut hyperedges, which we refer to as the \emph{all-or-nothing} cut function:
\begin{equation}
\label{eq:allornothing}
\textbf{all-or-nothing}(S) = \sum_{e \in \partial{S}} w_e.
\end{equation}
However, this is only one among many cut functions that may be reasonable to minimize when solving cut-based hypergraph problems.

\subsubsection*{Hypergraph $s$-$t$ cuts}
Lawler was the first to introduce a hypergraph generalization of the $s$-$t$ cut problem~\cite{lawler1973}. Given an undirected hypergraph $\mathcal{H}$, Lawler considered how to find a minimum weight set of hyperedges to remove in order to separate two terminal nodes $s$ and $t$. In other words, the all-or-nothing cut~\eqref{eq:allornothing} is the penalty associated with any bipartition. Lawler showed that this problem can be solved in polynomial time by converting the problem into a graph $s$-$t$ cut problem in a \emph{directed} graph on a larger set of vertices. We explain the reduction in depth when we consider various graph reduction techniques in Section~\ref{sec:positive}. Liu and Wong later presented a simplification of this graph reduction~\cite{liuwong1998}, which was further simplified by Heuer et al.~\cite{heuer_et_al:LIPIcs:2018:8936}. Both results also apply to the all-or-nothing cut penalty.

\subsection{Submodular Function Minimization}
\label{sec:submodular}
Cut functions in graphs, and in turn the all-or-nothing hypergraph cut function~\eqref{eq:allornothing}, are well-known special cases of submodular functions~\cite{cunningham1985minimum}. Given a universe set $\Omega$, a set function $f$ is defined to be submodular if for all sets $A \subset \Omega$ and $B \subset \Omega$, it satisfies the following property:
\begin{equation}
\label{eq:submodularf}
f(A \cap B) + f(A \cup B) \leq f(A) + f(B).
\end{equation}
Gr{\"o}tschel et al.\ gave the first polynomial time and strongly polynomial time algorithms for minimizing general submodular functions~\cite{grotschel1981,grotschel2012geometric},
and later work developed improved algorithms~\cite{Iwata:2001:CSP:502090.502096,Iwata:2009:SCA:1496770.1496903,Orlin2009,Schrijver:2000:CAM:361537.361552}.

The minimum $s$-$t$ cut problem on a graph $G = (V,E)$ can be cast as a special case of submodular function minimization by defining 
\begin{align*}
f\colon S \subset V \backslash \{s,t\}\rightarrow \mathbb{R}\\ 
f(S) = \cut(S \cup \{s\})\,.
\end{align*}
An analogous approach can cast the all-or-nothing hypergraph $s$-$t$ cut problem as submodular minimization. Minimizing general submodular functions is a much broader problem than minimizing cut functions over graphs and hypergraphs, and therefore algorithms for the latter are much more efficient than general ones for submodular minimization. For this reason, there is a significant body of research on identifying set functions that can be represented or at least approximated by graph or hypergraph cut problems~\cite{cunningham1985minimum,devanur2013approximation,fujishige2001realization,jegelka2011fast,panli2017inhomogeneous,Yamaguchi:2016:RSS:2938781.2938842}. Submodularity and graph representability both play a central role in our results on minimizing generalized hypergraph cut functions.

\section{Technical Framework}
\label{sec:framework}
We now define our technical framework for solving hypergraph $s$-$t$ cut problems under generalized notions of cuts.
We first formalize the concept of hypergraph splitting functions, which assign cut penalties to different arrangements of nodes in a hyperedge.
After, we motivate and analyze a natural class of \emph{cardinality-based} splitting functions and later generalize to asymmetric
hypergraph cuts (Section~\ref{sec:nonsymmetric}) and \emph{multiway} splitting functions (Section~\ref{sec:multiway}).
Some of our notation and definitions are inspired by a model for inhomogeneous hypergraph cuts, where different ``splits'' of a hyperedge are approximated by weighted graph cuts on the same set of nodes~\cite{panli2017inhomogeneous}.
Our focus is on the computational complexity of generalized hypergraph $s$-$t$ and multiway cut problems, as well as algorithms for solving them.

\subsection{Hypergraph Splitting Functions}
\label{sec:splitting}
Let $\mathcal{H} = (V,E)$ denote a hypergraph where $V = \{v_1, \hdots, v_n\}$ is a set of vertices and each hyperedge $e \in E$ corresponds to a subset of $V$ of arbitrary size. A hyperedge is \emph{cut} if it spans more than one cluster in a partition of $V$. In graphs, there is no meaningful difference between the inherent weight of an edge and the cut penalty associated with separating its two nodes. However, in order to accommodate generalized cut penalties in hypergraphs, we define a \emph{hypergraph splitting function} that maps each possible configuration of a hyperedge to a nonnegative splitting penalty.
\begin{definition}
	\label{def:splitting}
For each $e \in E$, let $2^e$ denote the power set of $e$. A hyperedge splitting function on $e$ is any function $\vw_e\colon 2^e \rightarrow \mathbb{R}$ satisfying
\begin{align}
\label{eq:nn}
(\text{Non-negativity} ) &\hspace{1cm} \vw_e(S) \geq 0 & \text{ for all $S \subset e$.} \\
\label{eq:symmetric}
(\text{Symmetry} ) & \hspace{1cm}\vw_e(S) = \vw_e(e \backslash S) & \text{ for all $S \subset e$.} \\
\label{eq:zerononcut}
(\text{Non-split ignoring} )& \hspace{1cm}  \vw_e(e) = \vw_e(\emptyset) = 0. 
\end{align}
\end{definition}
In principle these requirements could be relaxed to obtain even broader generalizations of hypergraph cut problems. However, for the minimization problems we consider, it is most natural to only penalize \emph{cut} hyperedges (Property~\eqref{eq:zerononcut}), and avoid minimizing negative scores (Property~\eqref{eq:nn}), as these would not truly correspond to penalties. Property~\eqref{eq:symmetric} generalizes the fact that in undirected graphs, permuting the location of two nodes in a cut edge does not change the cut penalty. In Section~\ref{sec:nonsymmetric}, we will relax this requirement in order to consider an asymmetric version of the hypergraph $s$-$t$ cut problem. However, unless we explicitly state otherwise, the term \emph{splitting function} always refers to a symmetric function.

We will pay special attention to two classes of splitting functions.
The first class is \emph{submodular} splitting functions, which for all $S_1, S_2 \in 2^e$ satisfy 
\begin{equation}
\label{submodular}
\vw_e(S_1) + \vw_e(S_2) \geq \vw_e(S_1 \cap S_2) + \vw_e(S_1 \cup S_2).
\end{equation}
Submodular splitting functions are a more restricted class that can still encode natural structural properties for clustering applications.
Even so, in the most general setting, a submodular splitting function is parameterized by $2^{r-1}-1$ different penalty scores for an $r$-node hyperedge, which is the same number of scores needed to define a general splitting function. The reason is that submodular splitting functions are not necessarily \emph{anonymous}, in the sense that a node's name or identity may affect the splitting penalty. In constrast, many clustering applications do not place a special importance on any particular node; the goal is simply to cluster based on edge structure, without treating any node as inherently special.

To deal with these issues, we introduce a second special class of splitting functions.
\begin{definition}
A splitting function $\vw_e$ is cardinality-based if it satisfies
\begin{equation}
\label{cardinality}
\vw_e(S_1) = \vw_e(S_2) \text{ for all $S_1, S_2 \in 2^e$ with $|S_1| = |S_2|$.}
\end{equation}
\end{definition}
Cardinality-based splitting functions are anonymous: they do not distinguish between different types of nodes in a hyperedge. Instead, penalties are assigned based only on the number of nodes in each cluster. In addition to being a very natural model for applications, cardinality-based function are beneficial in that they can be parameterized by $\floor*{r /2}$ penalties for an $r$-node hyperedge. In other words, there can be at most a different penalty for each possible number of nodes on the \emph{small side} of a two-way split hyperedge. With this view in mind, we highlight an observation that will be relevant in future sections.
\begin{observation}
	\label{obs:23}
	Cardinality-based splitting functions on 2- and 3-node hyperedges are parameterized by a single penalty, and are therefore equivalent to the all-or-nothing splitting function.
\end{observation}
Despite this observation for small hyperedges, cardinality-based splitting functions have significantly more modeling power than the all-or-nothing splitting function in general. At the same time, they avoid having to determine $2^{r-1}-1$ scores, as is the case for defining a general splitting function or even a submodular splitting function. We also show in Section~\ref{sec:reduction} that the intersection of these splitting function classes has some remarkable properties.

Although this is the first formal definition of a cardinality-based splitting function, many existing hypergraph clustering and partitioning models indeed minimize an objective involving cardinality-based penalties. Table~\ref{tab:cardexamples} summarizes several cardinality-based splitting functions that have been (implicitly) used in practice.
\renewcommand{\arraystretch}{1.5}
\begin{table}[t]
  \caption{Examples of cardinality-based splitting functions appearing in previous literature on hypergraph clustering. In all cases, $S$ represents a subset of a hyperedge $e$.
    The first three examples are broadly used, and the citations provide a representative sample of research. In practice these are often scaled by a nonnegative weight.
    The latter two are more recent and are parameterized by some $\alpha \in (0,1)$.
    The discount count was designed for multiway partitions, but the two-cluster restriction here is also interesting and distinct from the all-or-nothing function.
    The L\&M submodular refers to a set of splitting functions used by Li and Milenkovic that is submodular and cardinality-based. Note that all of these splitting functions are the same on 2- and 3-node hyperedges (up to a scaling term).
    }
	\label{tab:cardexamples}
	\centering
\begin{tabular}{lll}
		\toprule 
	All-or-nothing & $\vw_e(S) = \begin{cases} 0 & \text{if $S \in \{e, \emptyset\}$} \\ 1 & \text{otherwise}\end{cases}$  & \cite{BensonGleichLeskovec2016,hadley1995,ihler1993modeling,lawler1973} \\
	Linear penalty & $\vw_e(S) = \min \{ |S|, |e \backslash S| \} $ & \cite{heuer_et_al:LIPIcs:2018:8936,hu1985multiterminal} \\
	Quadratic penalty & $\vw_e(S) = |S| \cdot|e \backslash S|$ & \cite{Agarwal2006holearning,hadley1995,vannelli1990Gomoryhu,Zhou2006learning,zien1999}\\
	Discount cut & $\vw_e(S) = \min \{ |S|^\alpha, |e \backslash S|^\alpha \} $  &  \cite{yaros2013imbalanced}\\
	L\&M submodular & $\vw_e(S) = \frac{1}{2} + \frac{1}{2} \cdot \min\left \{1, \frac{|S|}{ \floor*{\alpha |e|}}, \frac{|e\backslash S|}{\floor*{\alpha |e|}}\right \} $  & \cite{panli_submodular}\\
	\bottomrule
\end{tabular}
\end{table}

\subsubsection*{Scalar weights vs.\ splitting functions} 
In theory, every hyperedge in a hypergraph can be associated with its own unique splitting function. From this perspective, splitting functions can be viewed as a generalization of scalar hyperedge weights. In practice, however, a hypergraph will rarely be associated with a set of splitting functions that are somehow inherent to its structure. We therefore treat splitting functions as a way to describe the type of hypergraph cut problem we wish to solve, rather than viewing them as an inherent part of the hypergraph. Given a hypergraph $\mathcal{H} = (V,E)$, we will typically apply a single type of splitting function to all hyperedges when solving a hypergraph cut problem. For $r$-uniform hypergraphs, the same exact function can be applied. If the hypergraph comes with scalar hyperedge weights, these can be used to scale the splitting function at each hyperedge. In practice, it will often make sense to solve cut problems on the \emph{same} hypergraph using several {different} splitting functions, to highlight different types of clustering structure in the same hypergraph.

\subsection{The Generalized Hypergraph Minimum $s$-$t$ Cut Problem}
Let~$\mathcal{H} = (V,E)$ be a hypergraph and assume we have a splitting function~$\vw_e$ for each $e \in E$. If $\mathcal{H}$ comes with scalar hyperedge weights, we assume these have been incorporated directly into the splitting functions. For a set of nodes $S \subset V$, define the generalized hypergraph cut score of~$S$ to be
\begin{equation}
\cut_\mathcal{H}(S) = \sum_{e \in E} \vw_e(e \cap S) = \sum_{e \in \partial S} \vw_e(e \cap S).
\end{equation}
where $\partial S = \{ e \in E : e \cap S \neq \emptyset, e \cap \bar{S} \neq \emptyset \}$ denotes the set of cut hyperedges. 

\begin{definition}
Let $\{s,t\} \subset V$ be designated source and sink nodes in the hypergraph $\mathcal{H}$.
The generalized hypergraph minimum $s$-$t$ cut problem (\hstgen{}) is the following optimization problem:
\begin{equation}
\label{eq:hypermincut}
\begin{array}{ll}
\minimize_{S \in V} & \cut_\mathcal{H}(S) \\
{\normalfont \text{subject to }} & s \in S, t \in \bar{S}.
\end{array}
\end{equation}
\end{definition}
Alternatively, we can define the problem as a special case of function minimization.
Let $f\colon S \subseteq V\backslash\{s,t\} \rightarrow \mathbb{R}_+$ where $f(S) = \cut_\mathcal{H}(S \cup \{s\})$.
Then the generalized hypergraph minimum $s$-$t$ cut problem is simply
\begin{equation}
\label{eq:hypercutfun}
\minimize_{S\subseteq V\backslash\{s,t\}  } \,\,f(S).
\end{equation}
In cases where we apply the same splitting function to all hyperedges in $\mathcal{H}$, we will simply refer to the problem as \hc{}, preceded by the name of the splitting function. Thus, the hypergraph $s$-$t$ cut problem introduced by Lawler~\cite{lawler1973} can be referred to as \emph{all-or-nothing} \hc{}. We will pay special attention to \emph{cardinality-based} \hc{}, and \emph{submodular} \hc{}.

For {submodular} \hc{}, the hypergraph cut objective function~\eqref{eq:hypercutfun} is a sum of submodular functions, and therefore is itself submodular. This means that this special case can be solved in strongly polynomial time, using algorithms for general submodular function minimization~\cite{Iwata:2001:CSP:502090.502096,Iwata:2009:SCA:1496770.1496903,Orlin2009,Schrijver:2000:CAM:361537.361552}. However, algorithms for general submodular function minimization are much slower than algorithms for the minimum $s$-$t$ cut problem in graphs. A natural question to ask, then, is whether in special circumstances we can find more efficient algorithms for solving {submodular} \hc{}. Lawler's~\cite{lawler1973} work on the all-or-nothing penalty (which is a simple submodular splitting function) demonstrates that there is at least one special case that can be solved via reduction to the graph $s$-$t$ cut problem. Another interesting question to ask is whether or when \hstgen{} is tractable for non-submodular splitting functions. We address these questions in depth in the next sections.

\section{Positive Results via Graph Reduction}
\label{sec:positive}
A common technique for solving hypergraph cut problems is to reduce the hypergraph $\mathcal{H}$ into a related graph $G_\mathcal{H}$ and then apply existing algorithms and techniques for graph cut problems. As long as the node set of $G_\mathcal{H}$ includes all nodes from $\mathcal{H}$, partitioning $G_\mathcal{H}$ induces a partition on $\mathcal{H}$. The goal is therefore to construct $G_\mathcal{H}$ in such a way that graph cuts in $G_\mathcal{H}$ at least approximately model hypergraph cuts in $\mathcal{H}$.

We will use the term \emph{hyperedge expansion} to refer to any strategy for replacing a hyperedge with a small graph, which includes both the original nodes in the hyperedge as well as new auxiliary nodes. We begin this section by reviewing three existing techniques for hyperedge expansion, after which we demonstrate how each can be used to exactly solve the all-or-nothing \hc{} problem in 3-uniform hypergraphs. Inspired by the similarities and differences among these approaches, we present a general framework for hyperedge expansion specifically for minimum $s$-$t$ cut problems. Using this framework, we prove that instances of cardinality-based \hc{} can be exactly modeled by minimum $s$-$t$ cut problems in graphs, if and only if the cardinality-based splitting functions are also submodular.

\subsection{Three Previous Techniques for Hyperedge Expansion}
Figure~\ref{fig:expansions} displays three existing techniques for converting a hyperedge into a small graph.
\begin{figure}[t]
	\centering
	\subfloat[Hyperedge \label{fig:3edge}]
	{\includegraphics[width=.24\linewidth]{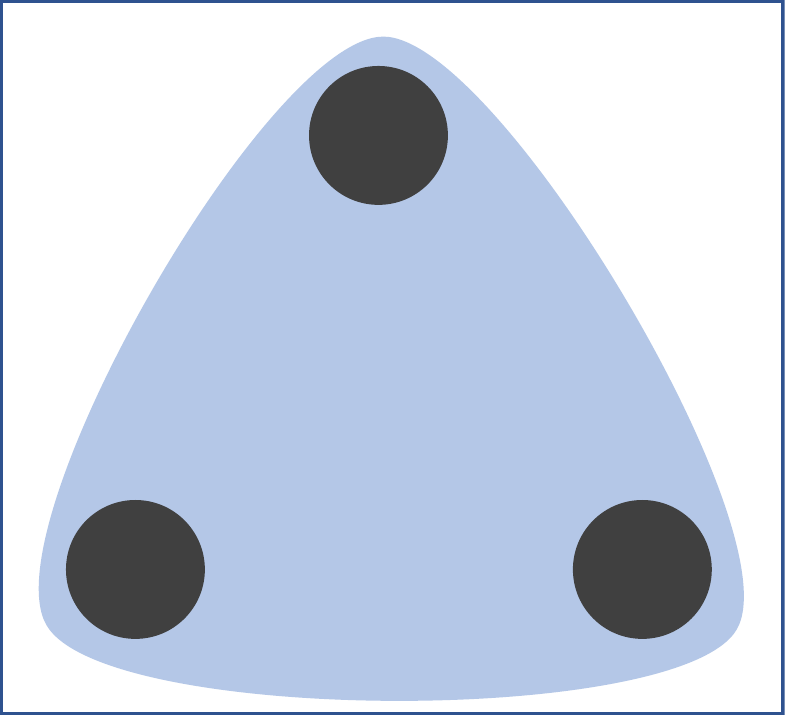}}\hfill
	\subfloat[Clique expansion \label{fig:3clique}]
	{\includegraphics[width=.24\linewidth]{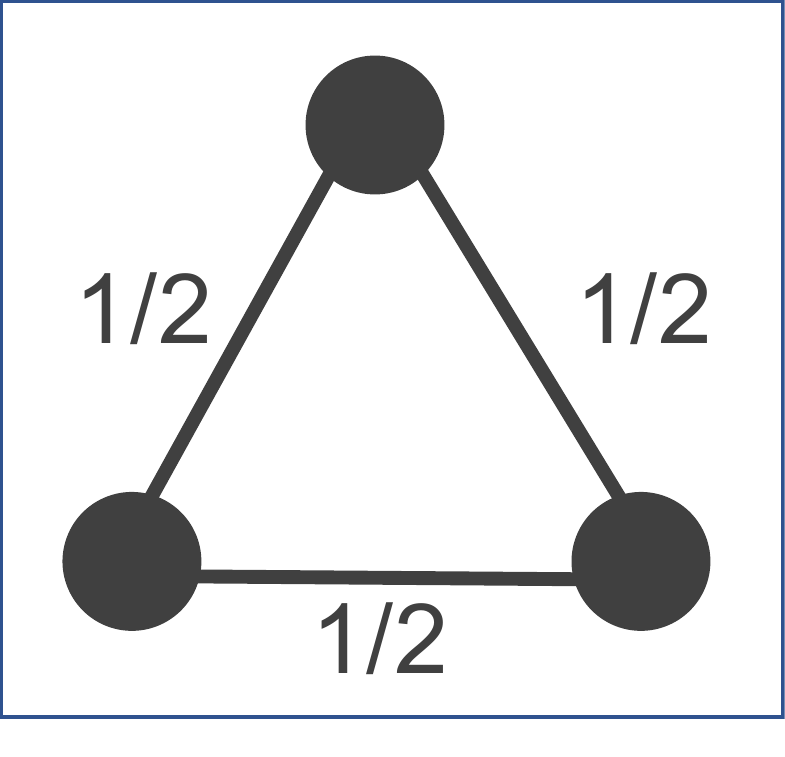}}\hfill
	\centering
	\subfloat[Star expansion\label{fig:3star}]
	{\includegraphics[width=.24\linewidth]{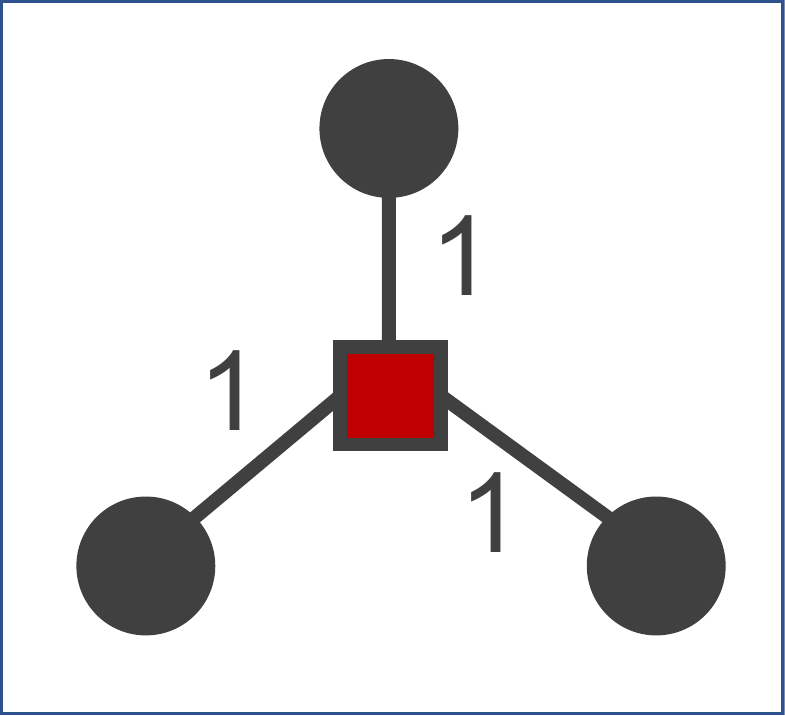}}\hfill
	\subfloat[Lawler expansion\label{fig:3lawler}]
	{\includegraphics[width=.24\linewidth]{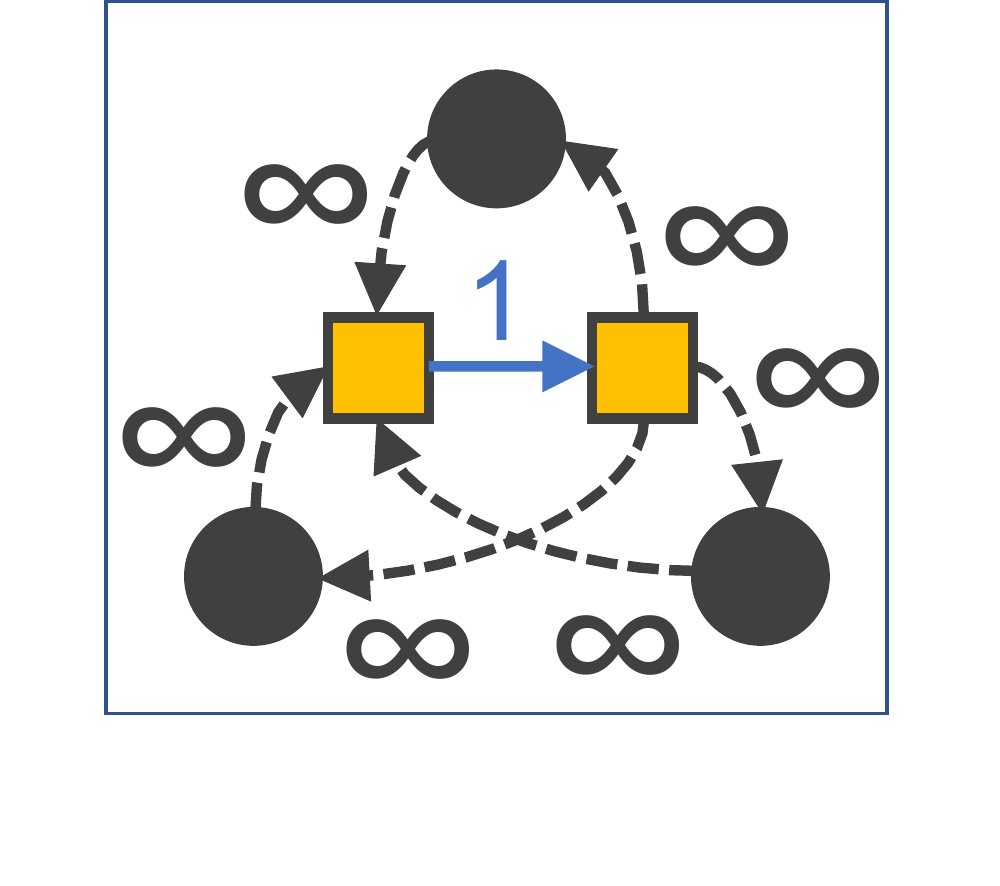}}
	\caption{Three existing techniques for converting a hyperedge into a set of edges and (possibly auxiliary) nodes. For 3-node hyperedges, each expansion technique can perfectly model the all-or-nothing hyperedge splitting function, if the displayed edge weights are used.}
	\label{fig:expansions}
\end{figure}

\textbf{Clique expansion.}
Arguably the most popular technique for reducing a hypergraph to a graph is the \emph{clique expansion}, which replaces a hyperedge $e \in E$ with a (possibly) weighted clique on all nodes in $e$. In the case of 3-node hyperedges, if all edges are given weight $1/2$, this weighting preserves the all-or-nothing hyperedge splitting penalty~\cite{ihler1993modeling,vannelli1990Gomoryhu}. In order to apply the clique expansion to an entire hypergraph, one can introduce an edge between a pair of nodes $i$ and $j$ for each hyperedge containing both nodes, and later merge all edges into one weighted edge by summing up the weights from each individual hyperedge expansion.

\textbf{Star expansion.} 
In the star expansion, a new node $v_e$ is introduced for every hyperedge $e$. Node $v_e$ is attached by an undirected edge to every node in the hyperedge $e$. When computing minimum cuts in the resulting graph, for any partition of $e$, node $v_e$ will be placed with the side of the cut that contains a majority of nodes from $e$. 

\textbf{Lawler expansion.}
We use the term \emph{Lawler expansion} to refer to the approach Lawler~\cite{lawler1973} used to convert an instance of the all-or-nothing hypergraph $s$-$t$ cut problem into a graph $s$-$t$ cut problem. Given a hyperedge $e \in E$, introduce two auxiliary nodes $e'$ and $e''$. For each $v \in e$, add a directed edge of weight infinity from $v$ to $e'$, and a directed edge of weight infinity from $e''$ to $v$. Finally, place a directed edge from $e'$ to $e''$ with weight 1. In the case of weighted hypergraphs, the edge from $e'$ to $e''$ is assigned the weight $w_e$ of hyperedge $e$. Note that the only way to separate the nodes of $e$ into an $s$-side cluster and a $t$-side cluster \emph{without} cutting an infinite weight edge, will require cutting the edge from $e'$ to $e''$. Thus, any splitting of the nodes in the reduced graph will result in the same penalty of 1 (or $w_e$ for weighted hypergraphs).

\subsubsection*{Illustration on  3-uniform hypergraphs}
The star, clique, and Lawler expansions all provide a way to exactly solve the minimum $s$-$t$ cut problems in 3-uniform hypergraphs under the all-or-nothing penalty. Splitting up any 3-node hyperedge into two clusters will always place two nodes in one cluster and one in the other. We can quickly see from Figure~\ref{fig:expansions} why this will lead to a penalty of exactly 1 when any of these expansions is used to reduce the hypergraph to a graph. Figure~\ref{fig:hypergraph9} illustrates the result of converting a 3-uniform hypergraph $s$-$t$ cut problem to a graph $s$-$t$ problem using each of these techniques. The minimum $s$-$t$ cut solution on all reduced graphs is the same. This can be seen as a consequence of Observation~\ref{obs:23}.

Although all of these expansion techniques lead to polynomial time algorithms for the 3-uniform problem, the computational complexity of solving the reduced minimum $s$-$t$ cut problem depends on which method is applied. The clique expansion has the advantage that it does not require auxiliary vertices. However, if the original hypergraph is unweighted, all edges in the star expansion will be unweighted as well. Since unweighted minimum $s$-$t$ cuts are easier to compute than weighted $s$-$t$ cuts~\cite{even1975network}, depending on the edge structure of the reduced hypergraph, it may be more efficient to apply the star expansion. Among the three expansion techniques considered here, the Lawler expansion is the least efficient approach for solving the minimum $s$-$t$ cut problem in 3-uniform hypergraphs, due to the larger number of auxiliary vertices and (weighted) edges that it requires. However, this is the only approach among the three that can be used to exactly solve the all-or-nothing \hc{} problem in hypergraphs with arbitrary-sized hyperedges, and thus it plays an important role in solving higher-dimensional problems.
\begin{figure}[t]
	\centering
	\subfloat[Hypergraph \label{fig:hyper}]
	{\includegraphics[width=.24\linewidth]{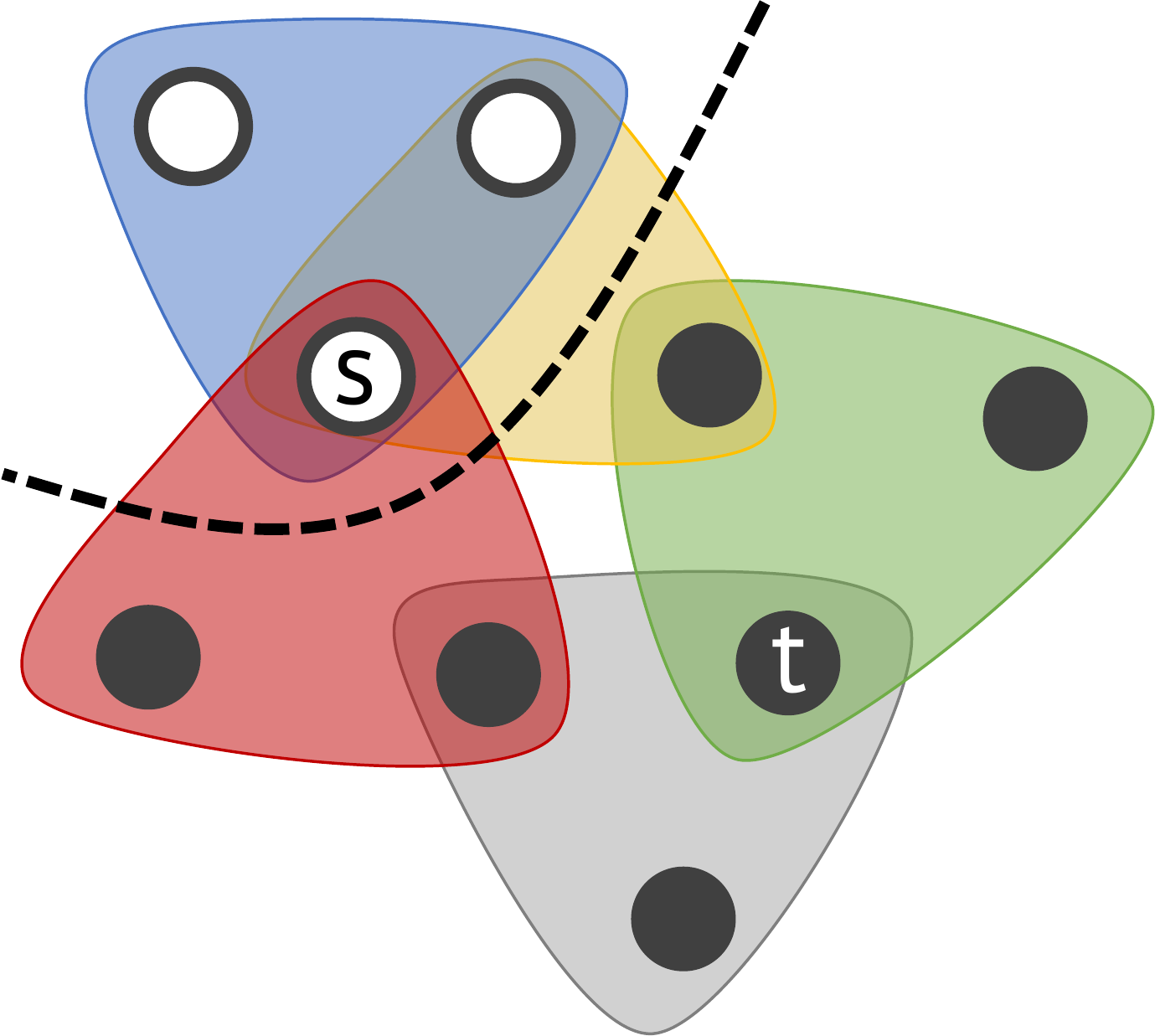}}\hfill
	\subfloat[Clique expansion \label{fig:clique}]
	{\includegraphics[width=.24\linewidth]{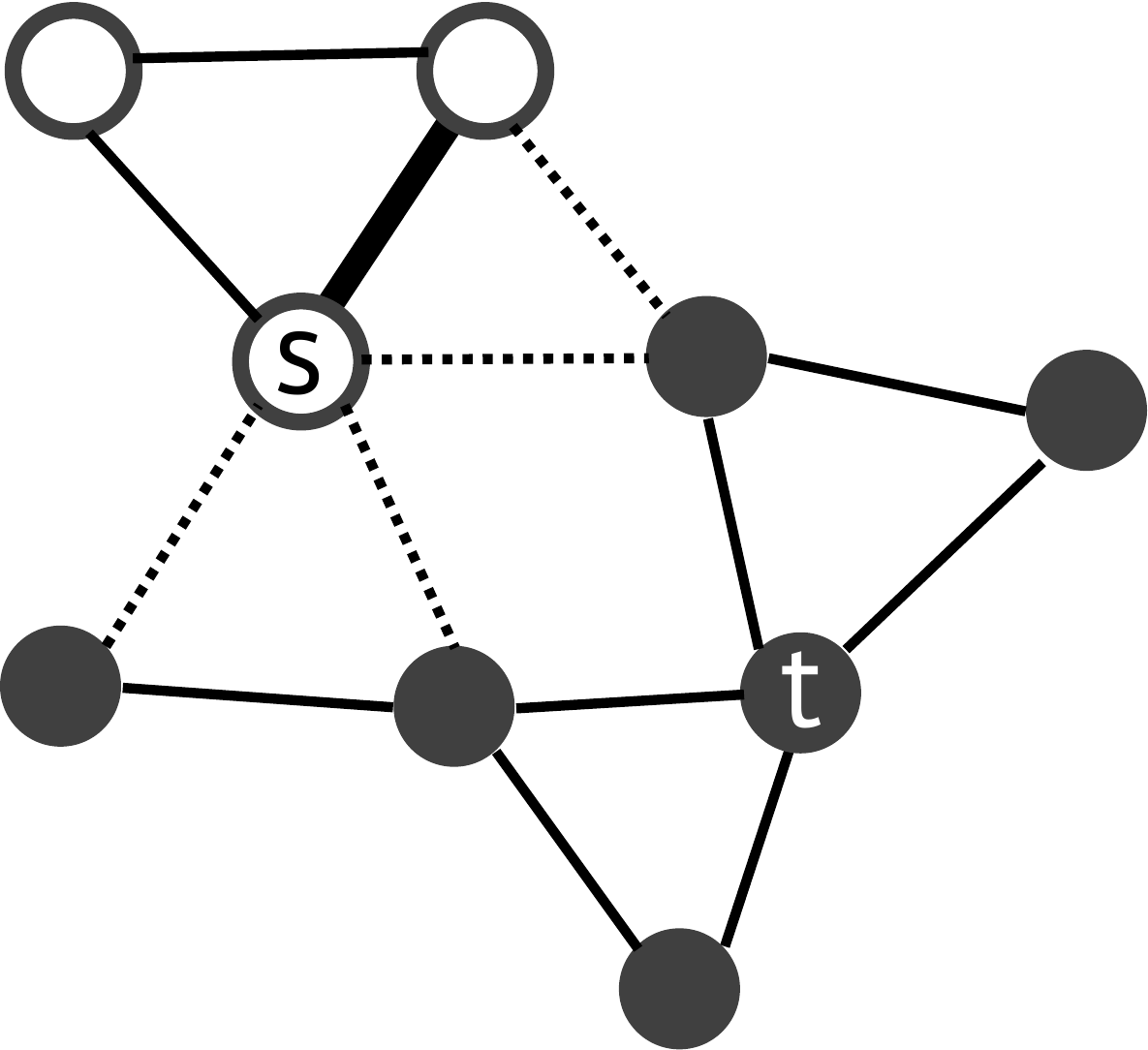}}\hfill
	\centering
	\subfloat[Star expansion\label{fig:star}]
	{\includegraphics[width=.24\linewidth]{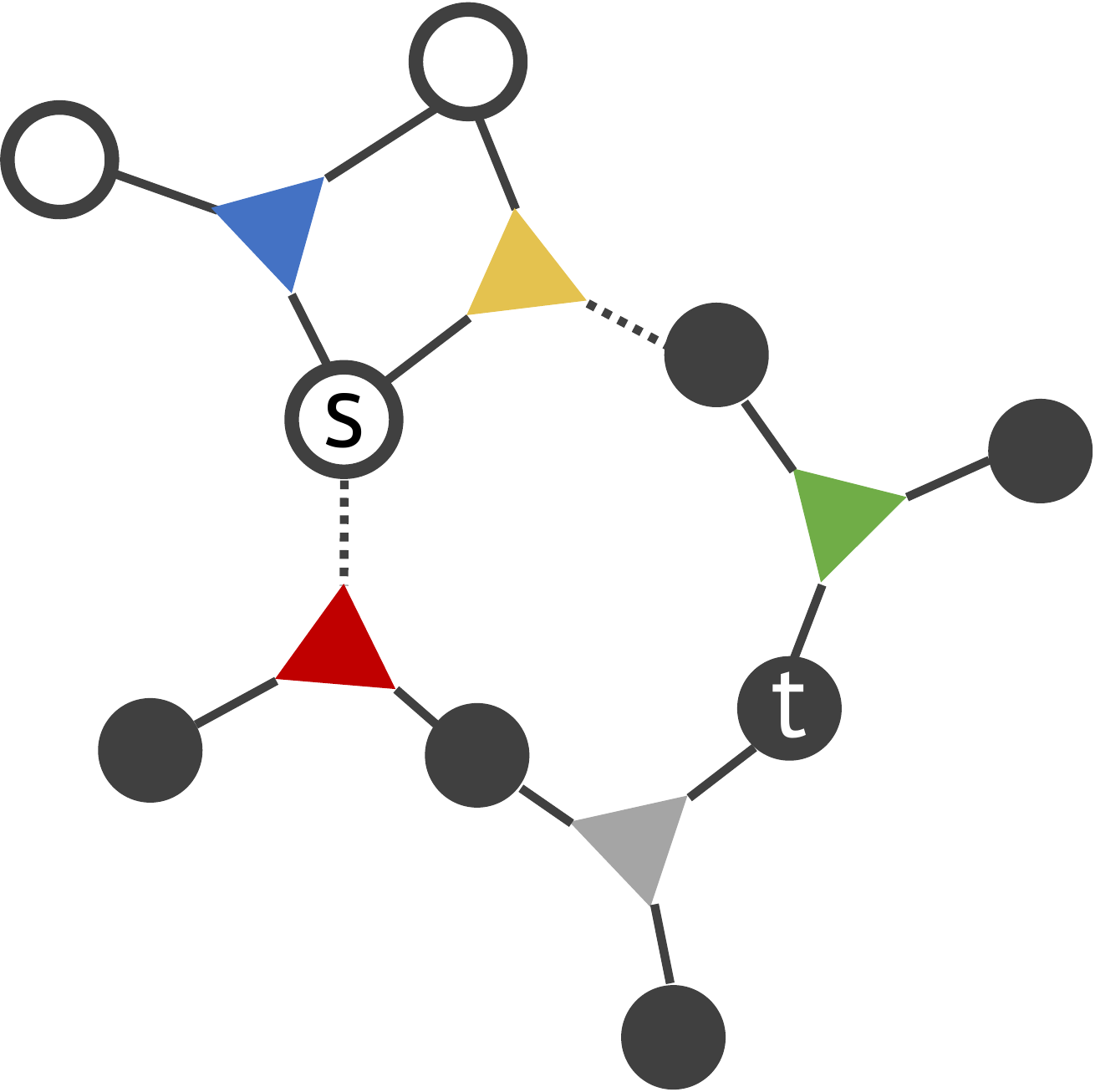}}\hfill
	\subfloat[Lawler expansion\label{fig:lawler}]
	{\includegraphics[width=.24\linewidth]{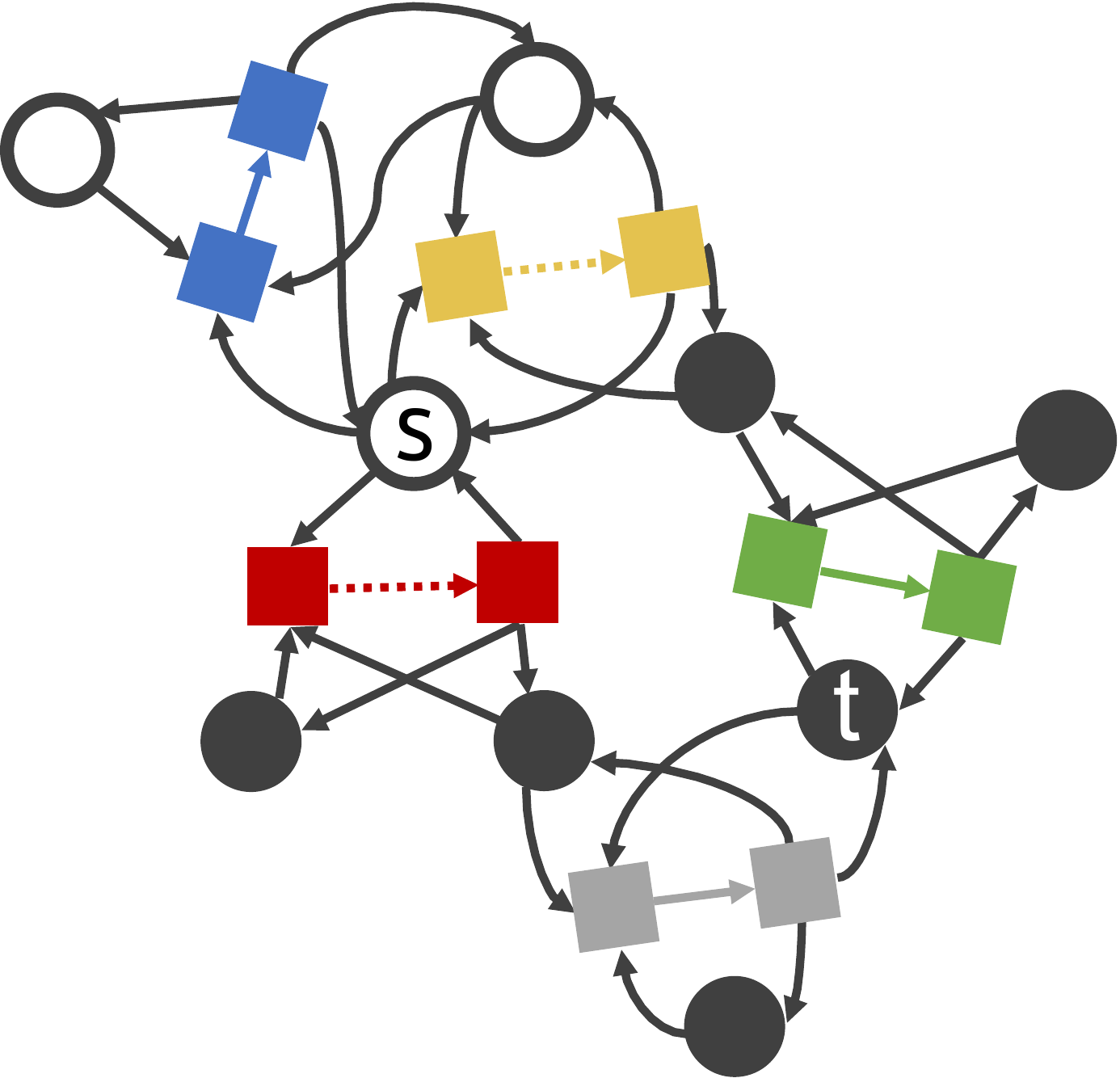}}
	\caption{A small hypergraph minimum $s$-$t$ cut problem (a) is converted into a graph using three different expansion techniques. Because the hypergraph is 3-uniform, solving the minimum $s$-$t$ cut problem on any of the resulting graphs will produce the minimum $s$-$t$ cut partition of the hypergraph under the all-or-nothing splitting function. In the clique expansion (b), all edges have weight $1/2$, except one edge with weight 1 (drawn with a thicker line). All edges in the star expansion (c) have weight 1. For the Lawler expansion (d), all black edges have weight infinity. The minimum $s$-$t$ cut score is two. We illustrate the solution with the minimum number of source side nodes, shown in white. Dotted lines indicate cut edges for each expansion technique.}
	\label{fig:hypergraph9}
\end{figure}

\subsection{Generalized Framework for Hyperedge Expansion}
\label{sec:stgadget}
The star and clique expansion are recognized as two approaches for the same high-level goal of reducing a hypergraph to a graph~\cite{Agarwal2006holearning,Hein2013,Huang2015,ihler1993modeling,zien1999}. The Lawler and star expansions have also appeared together~\cite{heuer_et_al:LIPIcs:2018:8936}, albeit to a lesser extent. These three approaches have not previously been simultaneously viewed as different strategies for the same overall goal. Motivated by these three expansions, we present a new generalized framework for reducing hypergraphs to graphs in a minimum-cut preserving way. 

\subsubsection*{Relation to previous work on hyperedge expansion}
Previously, Ihler et al.~\cite{ihler1993modeling} considered a broad framework for hyperedge expansion, with a nearly identical goal of understanding when hypergraphs can be represented by graphs with the same min-cut properties. Their conclusions were largely negative:\ except in the case of 3-node hyperedges, the all-or-nothing cut penalty cannot be exactly modeled by introducing undirected edges and auxiliary nodes. Recent research in this direction has thus focused on graph representations that approximate hypergraph cuts as best as possible~\cite{panli2017inhomogeneous}.
While our motivation is similar, our approach differs in two key ways, leading instead to several positive results for graph reduction, specifically for the \hstgen{} problem. First, we allow the addition of directed edges, given that directed $s$-$t$ cut problems are also well-defined, and since this admits Lawler-type expansions in our analysis. Second, since we do not restrict to the all-or-nothing splitting function, what Ihler et al.\ viewed as a limitation of certain hyperedge expansions turns out to be a useful feature in our framework.
If a hyperedge expansion technique \emph{does not} perfectly model the all-or-nothing splitting penalty, then it must model a different type of splitting penalty that may in fact be useful in different applications.
By combining different hyperedge expansion techniques, we will show how to model a broad range of hypergraph $s$-$t$ cut problems as graph $s$-$t$ cut problems.

\subsubsection*{Hypergraph $s$-$t$ gadgets}
We formally define the concept of a \emph{hyperedge $s$-$t$ gadget}, a type of hyperedge expansion that can be used to reduce \hstgen{} problems to related minimum $s$-$t$ cut problems on graphs.
\begin{definition}
	\label{def:gadget}
	Let  $e \in E$ be a hyperedge in a hypergraph $\mathcal{H} = (V,E)$. 
	A {hyperedge $s$-$t$ cut gadget} on $e$ is a graph $G_e = (V', E')$ with node set $V' =  e \cup \hat{V}$, where $\hat{V}$ is an auxiliary node set, and $E'$ is set of weighted and possibly directed edges. The gadget is associated with a \textbf{gadget splitting function} $\hat{\vw}_e\colon S \subseteq e \rightarrow \mathbb{R}^+$ defined by
	\begin{equation}
	\label{eq:gadgetfunction}
	\hat{\vw}_e(S) = \minimum_{\substack{ T \subseteq V'\\T\cap e = S}} \, \,\cut_{G_e}(T) \,,
	\end{equation}
	where $\cut_{G_e}(T) = \sum_{i \in S} \sum_{j\in V' \backslash T} w_{ij}$ is the standard graph cut function on $G_e$.
\end{definition}
The gadget splitting function gives a formal way to compare hyperedge splitting scores $\vw_e(S)$ in the hypergraph, with a gadget splitting function score $\hat{\vw}_e(S)$, for any $S \subseteq e$. The minimization in~\eqref{eq:gadgetfunction} encodes how, when solving the minimum $s$-$t$ cut problem on a graph formed by concatenating many hypergraph $s$-$t$ cut gadgets, we always arrange auxiliary vertices in a way that leads to a locally minimal penalty at each hyperedge gadget. In other words, given any fixed bipartition $\{S, e\backslash S\}$ of a hyperedge $e$ with $S \subset e$, the gadget splitting function implicitly ``moves'' any auxiliary nodes of the gadget in a way that yields the smallest penalty, given this fixed bipartition of $e$.

Gadget splitting functions are directly related to hyperedge splitting functions (Section~\ref{sec:splitting}), though an important distinction must be made. A \emph{hyperedge} splitting function is any type of penalty function defined on a hyperedge, satisfying properties~\eqref{eq:nn},~\eqref{eq:symmetric}, and~\eqref{eq:zerononcut} in Definition~\ref{def:splitting}. A \emph{gadget} splitting function is a function defined using a small {graph}, designed for the purpose of \emph{modeling} a hyperedge splitting function. We distinguish between these two types of functions by including a hat~(~$\hat{}$~) over the gadget splitting functions: $\hat{\vw}_e$. We will say that a gadget splitting function $\hat{\vw}_e$ \emph{models} a hyperedge splitting function $\vw_e$ if $\hat{\vw}_e = \vw_e$. We will also say that a \hstgen{} problem is \emph{graph reducible} if each of its hyperedge splitting functions can be modeled by some gadget splitting function. 

\renewcommand{\arraystretch}{1.5}
\begin{table}[t]
  \caption{Examples of cardinality-based gadget splitting functions derived from common expansions (see also Figures~\ref{fig:expansions}~and~\ref{fig:hypergraph9}).
    A hypergraph $s$-$t$ gadget is a small graph $G_e = (V', E')$ constructed from a hyperedge $e \in E$. The node set $V'$ is made up of the original nodes in $e$ along with a set of \emph{auxiliary vertices} $\hat{V}$. The edges in $E'$ can be both weighted and directed. The gadget splitting function is the result of applying equation~\eqref{eq:gadgetfunction} from Definition~\eqref{def:gadget} to $G_e$. In the case of the Lawler gadget, this provides another way to formalize Lawler's observation that the all-or-nothing splitting function can be modeled using a small directed graph. We have displayed the unit-weight version for each gadget. Other weightings can be obtained by scaling all gadget edges by a nonnegative weight.}
	\centering
	\begin{tabular}{llll}
		\toprule 
		Gadget name & $\hat{V}$ & $|E'|$  & Gadget Splitting Function\\
		\midrule
		Lawler gadget & \{$e'$, $e''$\} & $2|e| + 1$ &  $\hat{\vw}_e(S) = \begin{cases} 0 & \text{if $S \in \{e, \emptyset\}$} \\ 1 & \text{otherwise}\end{cases}$ \\
		Clique gadget & $\emptyset$ & ${|e| \choose 2}$ & $\hat{\vw}_e(S) = |S| \cdot|e \backslash S|$ \\
		Star gadget & $v_e$ & $|e|$ & $\hat{\vw}_e(S) = \min \{ |S|, |e \backslash S| \}$ \\
		\bottomrule
	\end{tabular}
\label{tab:gadgetsplitting}
\end{table}

The star, clique, and Lawler expansions can all be viewed as hypergraph $s$-$t$ cut gadgets. We summarize the auxiliary vertex set, number of edges, and the gadget splitting function for each in Table~\ref{tab:gadgetsplitting}. We use the term \emph{gadget} rather than \emph{expansion} to emphasize the fact that we are now concerned with how these expansions model hyperedge splitting penalties in \hstgen{} problems. Comparing Table~\ref{tab:gadgetsplitting} to Table~\ref{tab:cardexamples}, we see that the Lawler gadget models the all-or-nothing penalty, the star expansion models the linear penalty, and the clique expansion models the quadratic penalty. In fact, the reason the linear and quadratic penalties arise in previous work (see references in Table~\ref{tab:cardexamples}) is precisely because the clique and star expansions are often used to approximate hyperedges in certain hypergraph cut problems. 

A gadget splitting function can be \emph{cardinality-based} in the same way as a hyperedge splitting function. For a gadget function to be cardinality-based, the graph $G_e = (V', E')$ must have an edge and auxiliary node structure that is symmetric with regard to the nodes in $e$. In other words, we should still get the same gadget splitting function even if we permute the node labels.
Table~\ref{tab:gadgetsplitting} only considers the simplest versions of the Lawler, clique, and star gadgets. One can generalize them to not be cardinality-based by allowing the edges to have different weights; the clique expansion has been extended in this way to model hypergraphs with special nodes~\cite{panli2017inhomogeneous}.

\subsection{Gadget Splitting Functions are Submodular}
The clique gadget does not involve any auxiliary vertices, and therefore from Definition~\ref{def:gadget} we can see that its gadget splitting function is just a cut function on a small graph. However, for any gadget involving one or more auxiliary nodes, the gadget splitting function is not a cut function, since evaluating it requires solving a small optimization problem over different arrangements of auxiliary nodes. Despite this, we show that gadget splitting functions, just like graph cut functions, are submodular. We prove this result for a broader class of functions whose evaluation involves a minimization problem like the one in~\eqref{eq:gadgetfunction}. Let $U$ represent a ``universe" set, and let $V \subset U$ represent a fixed subset. Assume that $f\colon U \rightarrow \mathbb{R}$ is a submodular function, and define a new function $g\colon X \subset V \rightarrow \mathbb{R}$ as follows:
\begin{equation}
\label{eq:gfun}
g(X) = \minimum_{\substack{Y \subseteq U \\ Y \cap V = X}} \, f(Y)\,.
\end{equation}

\begin{theorem}
	\label{thm:gsub}
	The function $g$ is submodular.
\end{theorem}
\begin{proof}
	By definition, $g$ is submodular if for any $A, B \subset V$, it satisfies
	\begin{equation}
	g(A\cap B) + g(A\cup B) \leq g(A) + g(B).
	\end{equation}
	Fix any $A$ and $B$. There must exist sets $A'$ and $B'$ that are subsets of $U$ such that
	\begin{align*}
	g(A) &= f(A') \text{ with } A = A' \cap V\\	
	g(B) & = f(B')  \text{ with } B = B' \cap V.
	\end{align*}
	By submodularity of $f$, we know that
	\[f(A' \cap B') + f(A' \cup B') \leq f(A') + f(B') = g(A) + g(B). \]
	Thus, to prove submodularity of $g$, it suffices to show that
	\begin{align}
	\label{gf1}
	g(A\cap B) &\leq  f(A' \cap B')\\
	\label{gf2}
	g(A\cup B) &\leq f(A' \cup B').
	\end{align}
	To show~\eqref{gf1}, note that $(A' \cap B') \cap V = (A' \cap V) \cap (B' \cap V) = A \cap B$, so
	\[ g(A\cap B) = \minimum_{\substack{Y \subset U \\ Y \cap V = A\cap B}} f(Y) \leq f(A' \cap B'). \]
	Similarly, to show \eqref{gf2}, note that $(A' \cup B') \cap V = (A' \cap V) \cup (B' \cap V) = A \cup B$, so
	\[ g(A\cup B) = \minimum_{\substack{Y \subset U \\ Y \cap V = A\cup B}} f(Y) \leq f(A' \cup B'). \]
\end{proof}

Given that the gadget splitting function~\eqref{eq:gadgetfunction} is a special case of function~\eqref{eq:gfun}, we immediately obtain the following result.
\begin{corollary}
	\label{cor:submodgadget}
Every hyperedge $s$-$t$ cut gadget splitting function is submodular. Therefore, if a \hstgen{} problem is graph reducible, the splitting function for each of its hyperedges is submodular.
\end{corollary}
This corollary immediately restricts the class of \hstgen{} problems that can be solved via reduction to a minimum $s$-$t$ cut problem in graphs.
Next we turn to sufficient conditions for graph reducibility in the case of cardinality-based penalties. 

\subsection{Submodular $+$ Cardinality-Based Implies Graph Reducible}
\label{sec:reduction}
Consider an $r$-node hyperedge $e \in E$ with splitting function $\vw_e$. If the splitting function is cardinality-based, then it is characterized completely by $q = \floor*{\frac{r}{2}}$ weights, which we will denote by $w_i$ for $i = 1, 2, \hdots q$, where
\begin{equation}
\vw_e(S) = w_i \text{ for every $S \subset e$ such that $\min \{ |S|, |e \backslash S| \} = i$}.
\end{equation}
In other words, for cardinality-based splitting functions, we only need to consider how many nodes of a hyperedge are on the \emph{small side} of a split. Our goal is to understand which $r$-node hyperedges with cardinality-based splitting functions can be modeled using a hypergraph $s$-$t$ cut gadget.
Corollary~\ref{cor:submodgadget} says that if such a splitting function can be modeled by an $s$-$t$ cut gadget, then it must be submodular. 
The question, then, is which cardinality-based submodular splitting functions can be modeled by $s$-$t$ cut gadgets? The answer, perhaps surprisingly, is all of them.

\subsubsection*{A new cardinality-based gadget}
We start by introducing a new hypergraph~$s$-$t$ cut gadget, which is similar in spirit to the Lawler gadget, and whose gadget splitting function depends on an integer parameter $b$.
We call this the cardinality-based gadget, or CB-gadget (Figure~\ref{fig:gcbgadget}).
\begin{figure}[t]
	\centering
	\includegraphics[width=.5\textwidth]{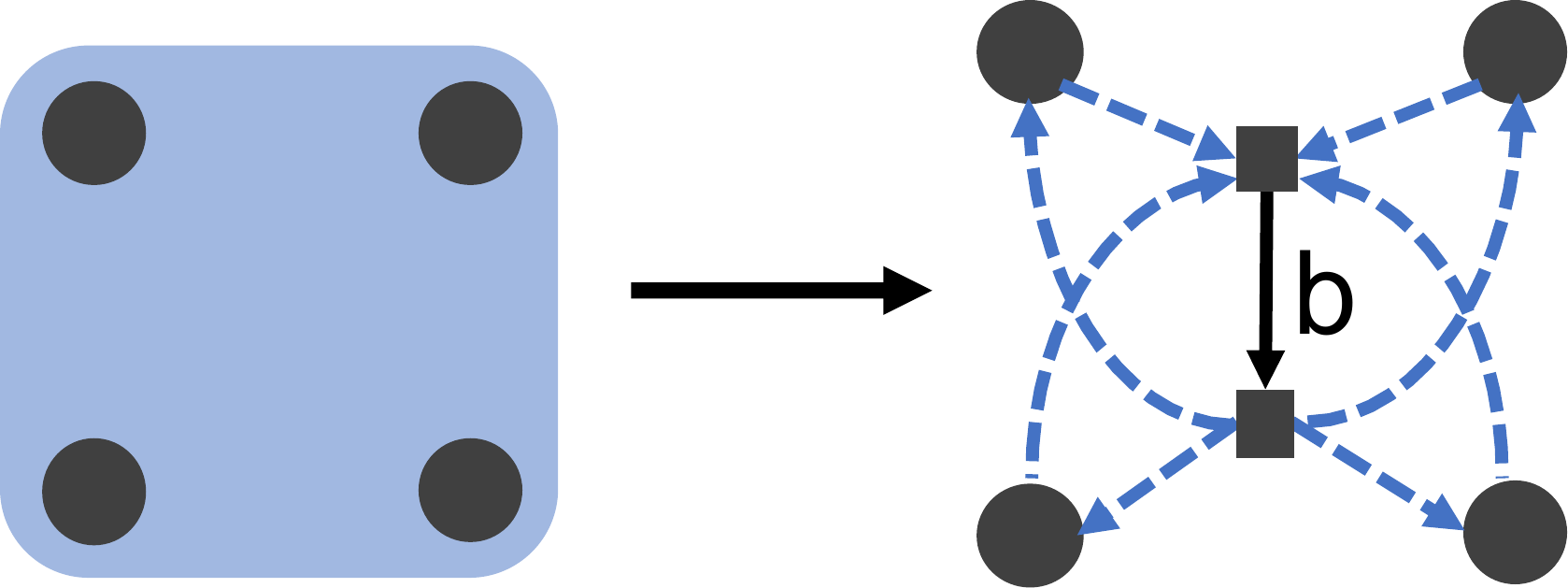}
	\caption{In the cardinality-based gadget, all edges have weight 1, except for the edge between auxiliary nodes, which has weight $b \in \mathbb{N}$.} 
	\label{fig:gcbgadget}
\end{figure}
For a hyperedge $e \in E$, the CB-gadget with parameter $b$ is formed as follows:
\begin{enumerate}
	\item Introduce two auxiliary vertices $e'$ and $e''$.
	\item For each $v \in e$, add a directed edge of weight 1 from $v$ to $e'$ and a directed edge of weight 1 from $e''$ to $v$.
	\item Add a directed edge of weight $b \in \mathbb{N}$ from $e'$ to $e''$.
\end{enumerate}	
When $b = 1$, the CB-gadget models the all-or-nothing splitting penalty, just as the standard Lawler expansion, and has the additional aesthetic appeal
of not needing edge weights. The only difference is that for the CB-gadget with $b = 1$, the auxiliary nodes may end up on different sides of a given bipartition, as compared to where they end up if the Lawler gadget were used. The overall splitting scores and the placement of nodes from $e$ remain the same. The benefit of the CB-gadget comes from setting $b$ to be different integers.
\begin{theorem}
\label{thm:gcbsplit}
 Given a hyperedge $e \in E$ and a subset of nodes $S \in e$, the CB-gadget with parameter $b$ has the following gadget splitting function:
\begin{equation}
\label{eq:genlawler}
\hat{\vw}_e(S) = \min \{|S|, |e \backslash S|, b  \}.
\end{equation}
\end{theorem}
\begin{proof}
 If $|S| \geq |e \backslash S|$, then it is better to put $e'$ and $e''$ in the same cluster as $S$ and cut all edges from $e''$ to $e\backslash S$. If $|S| < |e \backslash S|$, it is cheaper to cut all the edges from $S$ to $e'$. If $b$ is smaller that both $|S|$ and $|e \backslash S|$, then the cheapest cut for this bipartition is to put $e'$ with $S$, and $e''$ with $\bar{S}$, and cut the edge from $e'$ to $e''$. 
\end{proof}

\subsubsection*{Combining gadgets}
We can use a set of CB-gadgets with different parameters $b$ as a \emph{basis} set for constructing more sophisticated gadgets. Let $e $ be an $r$-node hyperedge with a cardinality-based and submodular splitting function $\vw_e$. Introduce $q = \floor{r/2}$ different CB-gadgets, with $e'_j$ and $e''_j$ representing the auxiliary nodes of the $j$th CB-gadget. Assign the weight from $e'_j$ to $e''_j$ to be $b_j = j$, and then scale all edge weights in this gadget by a scaling factor $c_j \geq 0$. The resulting combined gadget is made up of the original node set $e$, plus $2q$ auxiliary nodes and $q(2|e| + 1)$ edges.

Given a fixed set of scaling weights $c_1, c_2, \hdots c_q$, let $\hw_i$ be the gadget splitting function output for the combined gadget when there are $i$ nodes on the small side of the bipartition of $e$. If $i = 1$, then the $j$th gadget will return a penalty of 1 times the weight $c_j$, resulting in a combined penalty of $c_1 + c_2 + \hdots c_q = \hat{w}_1$.
In general, when there are $i$ nodes on the small side of the split, the combined gadget splitting score is
\begin{equation}
\hw_i = \sum_{j = 1}^q A_{ij} c_j\,,
\end{equation}
where $A_{ij} = \min \{i, j\}$. Let 
$\vc = [c_1 \,\, c_2 \,\, \cdots \,\, c_q]^T$
be a vector of scaling weights, and let 
$\hat{\vw} = [\hw_1 \,\, \hw_2 \,\, \cdots \,\, \hw_q]^T$
be the gadget splitting function penalties for a given choice of $\vc$. Define $\mA$ to be the $q \times q$ matrix whose $ij$ entry is $A_{ij}$. Note then that $\mA \vc = \hat{\vw}$. We write this matrix equation out explicitly for a large $q$ to illustrate a pattern in the relationship between $\vc$ and $\hat{\vw}$.
\begin{equation}
\label{matrixeq}
\begin{bmatrix}
1 & 1 & 1 & \cdots & 1 \\
1 & 2 & 2 &  \cdots & 2\\
1 & 2 & 3 &  \cdots & 3\\
\vdots & \vdots & \vdots  & \ddots & \vdots\\
1 & 2 & 3  &\cdots & q \\
\end{bmatrix}
\begin{bmatrix}
c_1 \\ c_2 \\ c_3 \\\vdots \\ c_q
\end{bmatrix}
= 
\begin{bmatrix}
\hw_1 \\ \hw_2 \\ \hw_3  \\\vdots \\ \hw_q
\end{bmatrix}.
\end{equation}
For a fixed $\vc$, this matrix equation tells us the exact gadget splitting function for the combined gadget. More importantly, by inverting the system~\eqref{matrixeq}, we can completely characterize which splitting functions our combined gadget can model. The inverse of $\mA$ is a tridiagonal matrix where $\mA^{-1}_{qq} = 1$, $\mA^{-1}_{ii} = 2$ for $i = 1, 2, \hdots, q-1$, and all entries directly above and below the main diagonal are $-1$. The inverted system is
\begin{equation}
\label{invmatrixeq}
\begin{bmatrix}
2 & -1  & \cdots & 0 & 0\\
-1 & 2 & \cdots & 0 & 0\\
\vdots  & \vdots  & \ddots & \vdots  & \vdots\\
0 & 0 &\cdots & 2 &-1 \\
0 &  0 &\cdots & -1 &1 \\
\end{bmatrix}
\begin{bmatrix}
\hw_1 \\ \hw_2  \\ \vdots \\\hw_{q-1} \\ \hw_q
\end{bmatrix}
= 
\begin{bmatrix}
2\hw_1 - \hw_2 \\ 2\hw_2 - \hw_1 - \hw_2  \\\vdots \\ 2\hw_{q-1} -\hw_{q-2}- \hw_{q}\\ \hw_q - \hw_{q-1}
\end{bmatrix}
= \begin{bmatrix}
c_1 \\ c_2  \\\vdots\\ c_{q-1} \\ c_q
\end{bmatrix}.
\end{equation}

The edges in the reduced graph must all be positive, in order to apply algorithms for the minimum $s$-$t$ cut problem, and therefore we must enforce $c_j \geq 0$ for all $j = 1, 2, \hdots q$. This means that in order for a hyperedge splitting function $\vw_e$ to be modelable by our combined gadget, its penalty scores $w_i$ for $i = 1, 2, \hdots, q$ must satisfy:
\begin{align}
\label{firstconstraint}
2w_1 &\geq w_2 \\
\label{mainconstraint}
2w_j &\geq w_{j-1} + w_{j+1} \text{ for $j = 2, \hdots, q-1$} \\
\label{extra}
w_q & \geq w_{q-1} \,.
\end{align}
Furthermore, from~\eqref{matrixeq}, the nonnegativity of $\vc$, and the structure of $\mA$,  the splitting scores for the combined gadget will always satisfy $0 \leq \hw_1 \leq \hw_2 \leq \hdots \leq \hw_q$. Therefore, in order for a hyperedge splitting function $\vw_e$ to be modeled by our combined gadget, it must satisfy the same inequality (which subsumes~\eqref{extra}):
\begin{equation}
\label{monotone}
0 \leq w_1 \leq w_2 \leq \hdots \leq w_q.
\end{equation}
All of these constraints are satisfied by submodular cardinality-based penalties.
\begin{lemma}
	\label{lem:cbsub}
	All cardinality-based submodular hyperedge splitting functions on $r$-node hyperedges satisfy inequalities~\eqref{firstconstraint}, \eqref{mainconstraint}, and~\eqref{monotone}.
\end{lemma}
\begin{proof}
	Let $e = \{v_1, v_2, \hdots v_r\} $ be an $r$-node hyperedge with a submodular and cardinality-based splitting function $\vw_e$, and let $q = \floor*{\frac{r}{2}}$. To prove the inequalities in~\eqref{mainconstraint}, first choose any $j \in \{2,3, \hdots (q-1)\}$. Define two sets of nodes with $j$ nodes each: $S_j = \{ v_1, v_2, v_3, \hdots ,v_j\}$ and $T_j = \{ v_2, v_3, \hdots v_j, v_{j+1}\}$.
	Since the splitting function $\vw_e$ is submodular, we see that
	\begin{equation}
	2w_j = \vw_e(S_j) + \vw_e(T_j) \geq \vw_e(S_j \cap T_j) + \vw_e(S_j \cup T_j) = w_{j-1} + w_{j+1}.
	\end{equation}
	Thus, \eqref{mainconstraint} is satisfied. If $S_1 = \{v_1\}$ and $T_1 = \{v_2\}$, then $S_1 \cap T_1 = \emptyset$. Therefore, since $\vw_e(\emptyset) = 0$, we also satisfy inequality~\eqref{firstconstraint}. To prove constraint~\eqref{monotone}, which subsumes~\eqref{extra}, we need to apply the symmetry constraint satisfied by all hyperedge splitting functions~\eqref{eq:symmetric}. For any $i \in \{1, 2, \hdots q-1\}$, define sets
	$S_i = \{ v_1, v_2, v_3, \hdots ,v_{r - (i+1)}\}$ and $T_i = \{ v_{r-2i}, v_{r-2i + 1}, \hdots , v_{r-i}\}$. Observe that 
	\begin{align*}
	|S_i| &= r - (i+1) \implies \vw_e(S_i) = w_{i+1} \\
	|T_i| &= (r-i) - (r-2i) + 1 = (i+1) \implies \vw_e(T_i) = w_{i+1}\\
	|S_i \cap T_i| &= i \implies \vw_e(S_i \cap T_i) = w_i \\
	|S_i \cup T_i| &= (r-i) \implies \vw_e(S_i \cup T_i) = w_i\,.
	\end{align*}
	By the definition of submodularity we know that $2w_{i+1} \geq 2w_{i} \implies w_{i+1} \geq w_i$. 
\end{proof}
Combining Lemma~\ref{lem:cbsub} and Corollary~\ref{cor:submodgadget}, we conclude that inequalities~\eqref{firstconstraint}, \eqref{mainconstraint}, and~\eqref{monotone} are in fact both sufficient and necessary conditions for a cardinality-based splitting function to be submodular. These two result together completely characterize the set of cardinality-based \hc{} problems that are graph reducible. We end with a summarizing theorem.
\begin{theorem}
	\label{thm:iffsub}
	Let $\mathcal{H} = (V,E)$ be a hypergraph. If each $e \in E$ is associated with a cardinality-based splitting function $\vw_e$, then the \hstgen{} problem on $\mathcal{H}$ is graph reducible if and only if $\vw_e$ is submodular for every $e \in E$.
\end{theorem}

\subsection{Examples on Real Data}
\label{sec:experiments}
We illustrate our graph reduction techniques by solving a range of $s$-$t$ cut problems on a hypergraph constructed from real data.

\subsubsection*{Dataset}
We consider data obtained from Math Stack Exchange, an online forum for discussing math questions (\url{https://math.stackexchange.com/})~\cite{BensonE11221}.%
\footnote{Original data available at \url{https://www.cs.cornell.edu/~arb/data/threads-math-sx/index.html}.}
Each entry in the dataset corresponds to a post on the forum about a math question, which is associated with 1 to 5 different tags related to the topic of the post (e.g., ``invariance'', ``topology'', ``hypergraphs''). We associate each tag with a node in a hypergraph. A set of tags appearing in the same post defines a hyperedge. Discarding posts with only one tag, we obtain a hypergraph $\mathcal{H}$ with 1,629 nodes and 169,259 hyperedges with 2 to 5 nodes each. 
Figure~\ref{fig:math1} in the introduction illustrates all 4-node hyperedges in the dataset containing the tag ``hypergraphs.''

\subsubsection*{Constructing a hypergraph $s$-$t$ cut problem}
In practice, simply choosing one node to be the source $s$ and another to be the sink $t$ typically produces minimum $s$-$t$ cut problems where the optimal solution places one terminal node in a cluster by itself. In order to obtained more balanced and meaningful bipartitions of the dataset, we introduce super-source and super-sink nodes, and connect each terminal to a designated node \emph{plus} its neighbors in the hypergraph. In more detail, we first choose two nodes $s$-seed and $t$-seed corresponding to tags in the dataset. We attach the super-source node $s$ to $s$-seed, and all nodes that share a hyperedge with $s$-seed, but not with $t$-seed. We similarly attach the super-sink $t$ to $t$-seed and its neighbors that do not also neighbor $s$-seed. All edges adjacent to $s$ and $t$ are given infinite weight. Our construction is related to the Graph Mincut algorithm of Blum and Chawla for semi-supervised learning~\cite{Blum:2001:LLU:645530.757779}, as well as other graph-based learning techniques that connect super-source and super-sink nodes to different subsets of nodes in an input graph before solving an $s$-$t$ cut problem~\cite{Andersen:2008:AIG:1347082.1347154,LangRao2004,Orecchia:2014:FAL:2634074.2634168}. Thus, while the construction is more sophisticated than simply identifying a source and sink node in $\mathcal{H}$, it is a better reflection of how our framework might be used in applications. 

\subsubsection*{Results}
We solve the cardinality-based $s$-$t$ cut problem for several different pairs of tags from the Math Stack Exchange hypergraph. Since minimum $s$-$t$ cut solutions are the same up to a multiplicative scaling of edge weights, we first fix $w_1 = 1$. The submodular region then corresponds to $w_2 \in [1,2]$, so we compute $s$-$t$ cut solutions with $w_2$ varying from 1 to 2 in increments of $0.05$. Minimum $s$-$t$ cuts may not be unique, so we always consider the cut with the minimum number of source-side nodes. 

Recall that hyperedges with three or fewer nodes are characterized by a single splitting penalty (Observation~\ref{obs:23}), and therefore, varying $w_2$ will only affect penalties at hyperedges with four or five nodes. Despite this, we observe significant differences in $s$-$t$ cut solutions as $w_2$ changes. Given a fixed $s$-seed and $t$-seed, let $S^*$ be the source-side solution set when $w_2 = 1$ (i.e., the all-or-nothing solution). For values of $w_2 \in [1,2]$, we compute the Jaccard similarity between the solution $S_{w_2}$ and the all-or-nothing solution: $\textbf{Jaccard}(S_{w_2},S^*) = \frac{|S_{w_2} \cap S^*|}{|S_{w_2} \cup S^*|}$. In Figure~\ref{fig:a}, we plot Jaccard similarity curves for a variety of different ($s$-seed, $t$-seed) pairs selected from the hypergraph. For some pairs, we observe noticeable differences in Jaccard scores as $w_2$ varies, while in other cases the curves remain mostly constant. Overall, Jaccard scores tend to steadily decrease as the gap between $w_2$ and $w_1$ increases. However, the decrease is not always monotonic, as can be seen in the green curve in Figure~\ref{fig:a}. 

\begin{figure}[t]
	\centering
	\subfloat[Jaccard Scores\label{fig:a}]
	{\includegraphics[width=.32\linewidth]{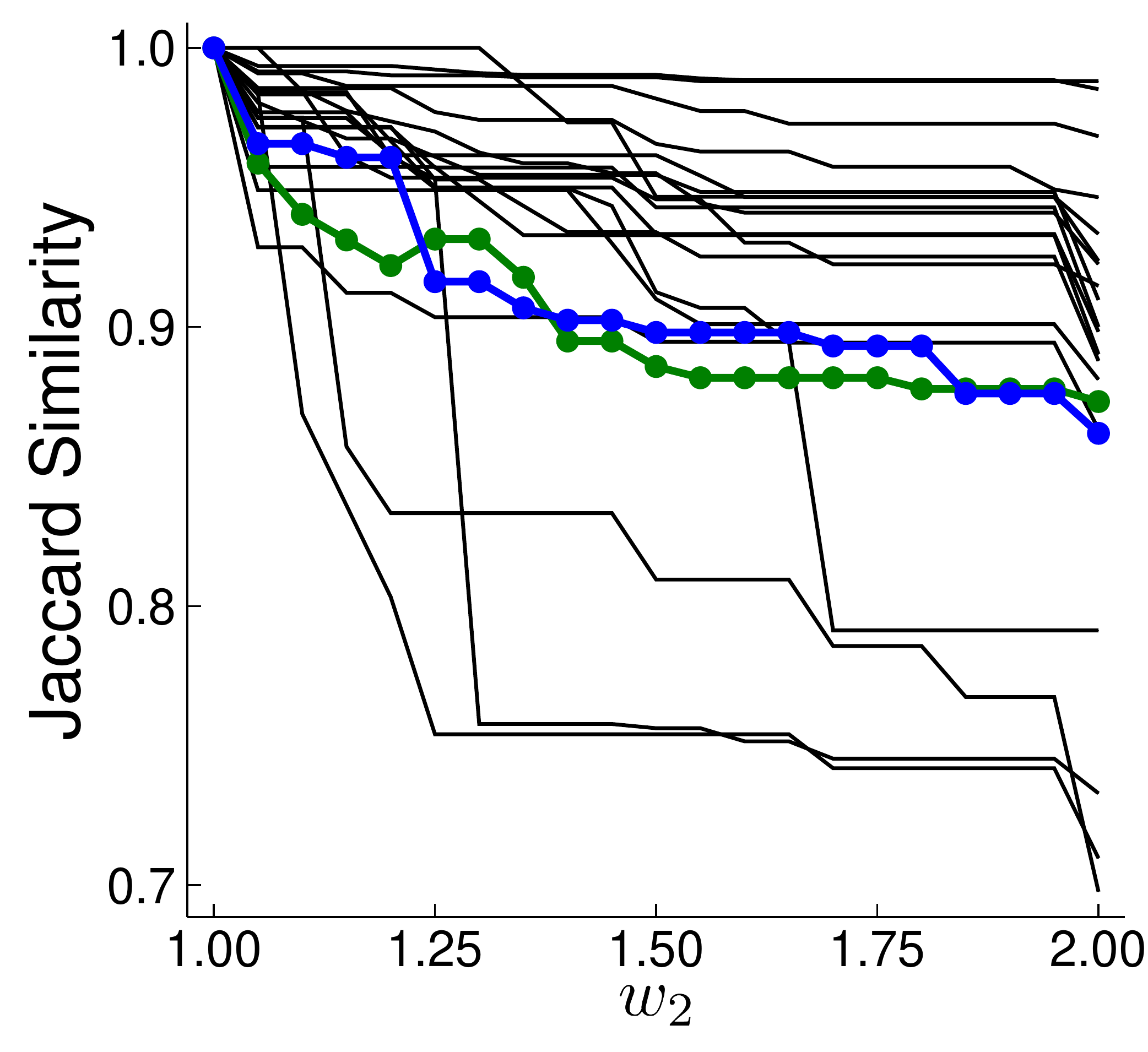}}\hfill
	\hspace{.01\linewidth}
	\subfloat[Green curve boundary\label{fig:b}]
	{\includegraphics[width=.335\linewidth]{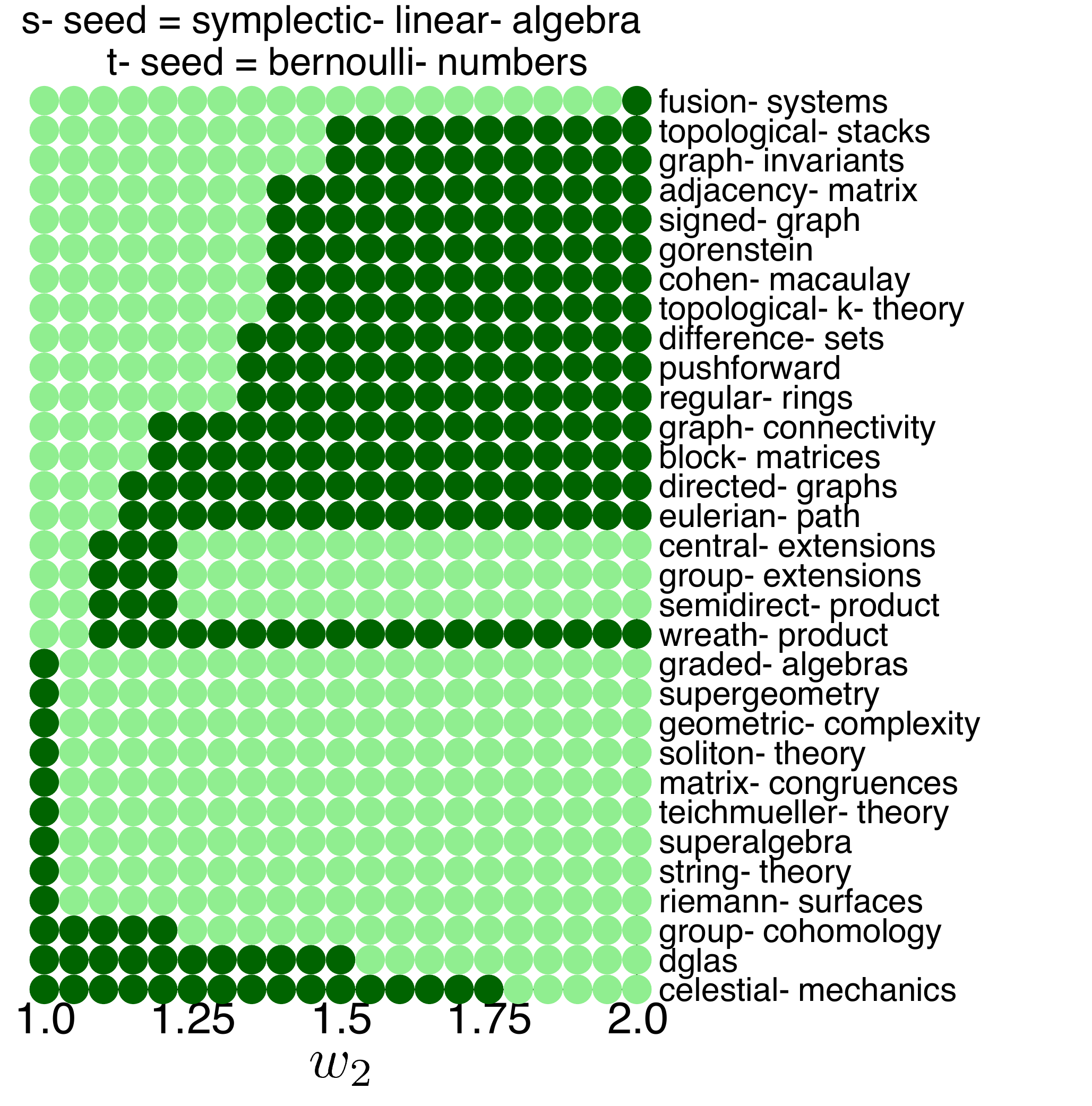}}\hfill
	\subfloat[Blue curve boundary\label{fig:c}]
	{\includegraphics[width=.325\linewidth]{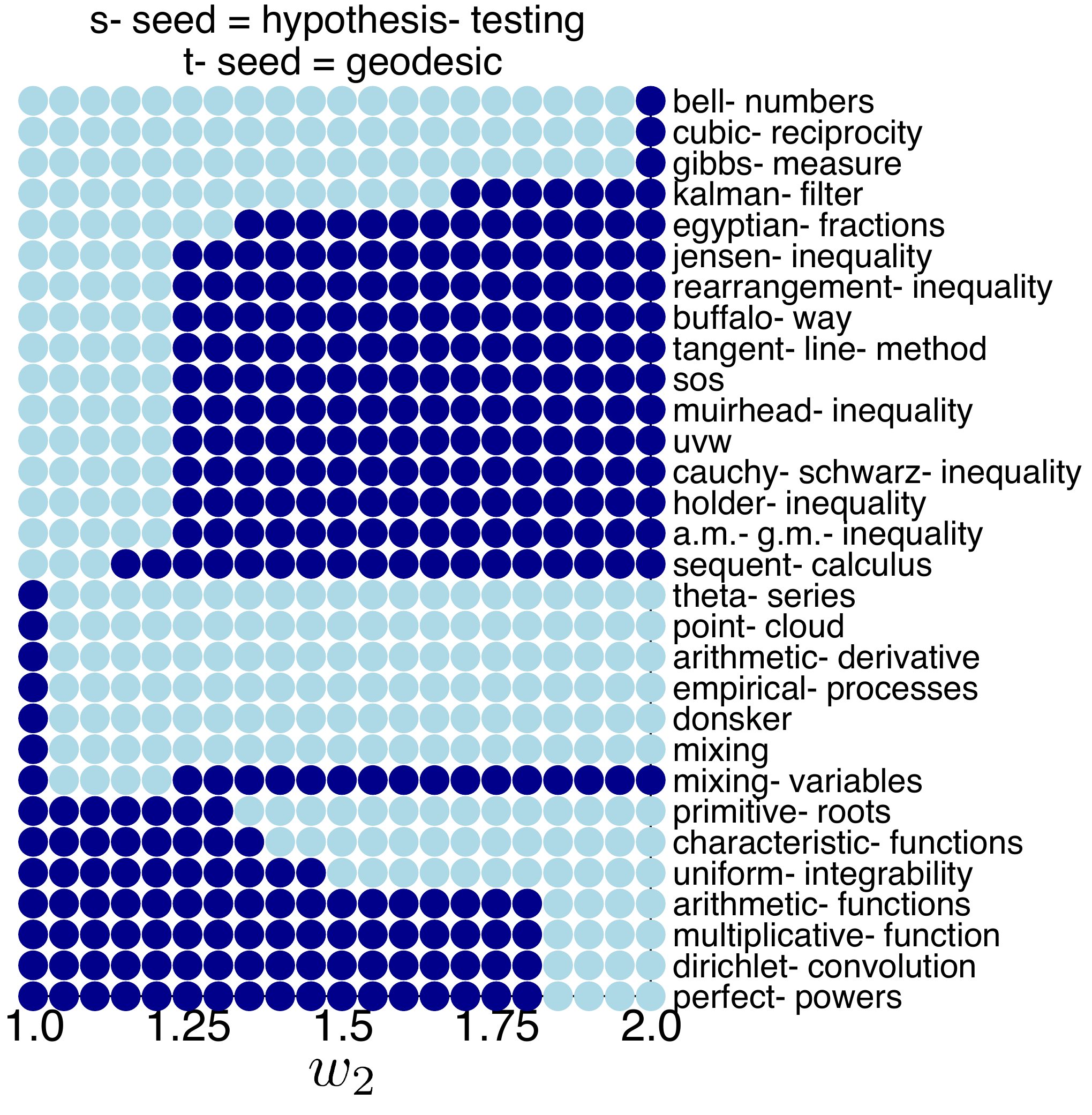}}
	\caption{Each curve in (a) corresponds to a sequence of cardinality-based minimum $s$-$t$ cut solutions on the Math Stack Exchange hypergraph, computed as $w_2$ ranges from 1 to 2. The super-source $s$ is attached via an infinite weight edge to one tag, and its neighbors. The super-sink is similarly attached to a different tag and its neighbors. Jaccard similarities are computed between the all-or-nothing (i.e., $w_2 = w_1$) solution, and the solution for the given $w_2$ in the horizontal axis.  Plots (b) and (c) show the changing cluster assignment for all of the nodes on the boundary of the blue and green curve respectively. Dark circles indicate the node is on the source side; light circles indicate a sink-side assignment. We also list the tags corresponding to each boundary node.}
	\label{fig:mathstack2}
\end{figure}

The green curve in Figure~\ref{fig:a} corresponds to setting $s$-seed and $t$-seed to the ``symplectic-linear-algebra'' and ``bernoulli-numbers'' tags respectively. For this problem, there are 30 nodes that switch sides at least once as $w_2$ changes. Figure~\ref{fig:b} illustrates where each of these boundary nodes is assigned for each value of $w_2$. Figure~\ref{fig:c} is a similar plot when the two seed nodes are ``hypothesis-testing'' and ``geodesic''. In many cases, there are interesting patterns in the tags on the boundary between $s$ and $t$ terms. For example, in Figure~\ref{fig:b}, we observe a number of tags associated with topics in graph theory (e.g., ``graph-invariants'', ``adjacency-matrix'', ``signed-graph'', ``directed-graphs''). The boundary nodes in Figure~\ref{fig:c} include six different tags associated with well-known mathematical inequalities (e.g., ``jensens-inequality'', ``rearrangement-inequality'', ``cauchy-schwartz inequality'').

In other experiments, we observed even lower Jaccard similarity scores as $w_2$ grew, as well as cases where a larger number of nodes in the hypergraph switched from source- to sink-side or vice versa.
We also noticed many examples of thematic boundary tags, relating to one or more specific subtopics in mathematics. For example, when we set ($s$-seed, $t$-seed) = (``random functions'',``svd''), 17 nodes switch sides at least once as $w_2$ varies from 1 to 2, and nearly all of these related to probability theory including ``polya-urn-model'', ``birth-death-process'', ``stopping-times'', ``brownian-motion'', three tags involving the word ``martingales'', and four involving ``stochastic''. Overall, these experiments show that solving the hypergraph $s$-$t$ cut problem with different cardinality-based splitting functions can lead to a range of different cuts in the same dataset. Furthermore, exploring differences in these cut solutions can uncover meaningful patterns in a dataset that would not be detected by solving only the all-or-nothing hypergraph $s$-$t$ cut problem.

\subsection{Asymmetric \hc{}}
\label{sec:nonsymmetric}
Our definition of a hyperedge splitting function includes a symmetry requirement, i.e., $\vw_e(S) = \vw_e(e \backslash S)$ for all $S \subseteq e$ (Definition~\ref{def:splitting}, Property \ref{eq:symmetric}). In principle, we can remove this requirement to obtain a well-defined notion of an \emph{asymmetric} \hc{} problem, for which splitting penalties at hyperedges depends both on \emph{how} the nodes of a hyperedge are split, as well as \emph{which} nodes are clustered on the source- or sink-side of a cut. This is a natural formulation to consider, for example, in $s$-$t$ cut problems where the source-side of a cut is intended to represent a cluster of nodes possessing a certain function or property. Consider, for example, a 3-node hyperedge $\{a,b,c\}$ where we have some prior reason to believe node $a$ possesses some property $X$ of interest, but $b$ and $c$ do not. If the goal of the $s$-$t$ cut problem is to identify a cluster of nodes with property $X$, then clustering $a$ with the source and $\{b,c\}$ with the sink should be treated differently from placing $\{b,c\}$ with the source and $a$ with the sink. In other cases, we may have a preference for including more (or fewer) nodes on the source-side of a cut, in which case a 1--2 split of a 3-node hyperedge should be treated differently from a 2--1 split. These distinctions cannot be modeled with symmetric splitting functions. 

\subsubsection*{Formal definitions}
An \emph{asymmetric} hyperedge splitting function is a function $\vy_e\colon 2^e \rightarrow \mathbb{R}_+$ that satisfies the splitting function requirements of Definition~\ref{def:splitting} except for the symmetry constraint~\eqref{eq:symmetric}. Thus, for a hyperedge $e$ and a subset $S \subset e$, an asymmetric hyperedge splitting function $\vy_e$ can have $\vy_e(S) \neq \vy_e(e\backslash S)$.
However, if $S = e$, we still have $\vy_e(S) = \vy_e(e\backslash S) = 0$ so that penalized hyperedges are in fact cut. Notions of cardinality-based~\eqref{cardinality} and submodular~\eqref{submodular} splitting functions naturally extend to this setting. An asymmetric cardinality-based splitting function on an $r$-node hyperedge is characterized by $r-1$ penalty scores $y_i$ for $i \in \{1,2, \hdots, r-1\}$, where $y_i$ is the penalty for placing $i$ nodes on the source-side of the cut. 

If the function is also submodular, it satisfies several properties that are closely related to inequalities~\eqref{firstconstraint}, \eqref{mainconstraint}, and~\eqref{monotone} considered in Theorem~\ref{thm:iffsub}.
\begin{lemma}
	\label{lem:asub}
	If $\vy_e$ is an asymmetric cardinality-based submodular splitting function on an $r$-node hyperedge $e = \{v_1, v_2, \hdots, v_r\}$, then its splitting penalties satisfy
	\begin{align}
	\label{a1}
	2y_1 &\geq y_2 \\
	\label{a2}
	2y_j &\geq y_{j-1} + y_{j+1} \text{ for $j = 2, \hdots, r-3$} \\
	\label{a3}
	2y_{r-2} &\geq y_{r-1}.
	\end{align}
\end{lemma}
\begin{proof}
	By definition of submodularity, for all sets of nodes $A,B \subseteq 2^e$,
	\begin{equation}
	\label{eq:submod}
	\vy_e(A) + \vy_e(B) \geq \vy_e(A \cap B) + \vy_e(A\cup B).
	\end{equation}
	If we set $A = \{v_1\}$, and $B = \{v_2\}$, then $\vy_e(A) = \vy_e(B) = y_1$, $\vy_e(A\cup B) = y_2$, and $\vy_e(A\cap B) = 0$, so inequality~\eqref{eq:submod} reduces to inequality~\eqref{a1}. Constraints~\eqref{a2} and~\eqref{a3} are similarly shown by using $A = \{v_1, v_2, \hdots, v_j\}$ and $B = \{v_2, v_3,\hdots , v_{j+1}\}$ for $j \in \{2, 3, \hdots , r-2\}$. 
\end{proof}

\subsubsection*{Graph reducibility for asymmetric cardinality-based functions}
Just as we did for symmetric splitting functions, we can model asymmetric splitting functions using hypergraph $s$-$t$ gadgets. Definition~\ref{def:gadget} of hypergraph $s$-$t$ gadgets does not explicitly consider any notion of symmetry, so it can be directly applied to the asymmetric setting. As before, we say that an asymmetric splitting function $\vy_e$ is \emph{modeled} by a gadget with gadget splitting function $\hat{\vy}_e$ if $\vy_e = \hat{\vy}_e$. An asymmetric \hc{} problem is \emph{graph reducible} if the splitting function at every hyperedge can be modeled by a gadget splitting function. 
\begin{theorem}
	\label{thm:iffns}
	An instance of asymmetric cardinality-based \hc{} is graph reducible if and only if every splitting function is submodular.
\end{theorem}
\begin{proof}
First, Theorem~\ref{thm:gsub} and Corollary~\ref{cor:submodgadget} apply in the same way to the asymmetric setting, and thus any cardinality-based splitting function that can be modeled by an $s$-$t$ gadget is submodular. To prove sufficiency, we consider a new asymmetric CB-gadget for a hyperedge $e$, parameterized by two integers $a$ and $b$, and whose construction is shown in Figure~\ref{fig:nss}.
\begin{figure}[h]
	\label{fig:nss}
	\begin{minipage}{0.45\textwidth}
		\includegraphics[width=\linewidth]{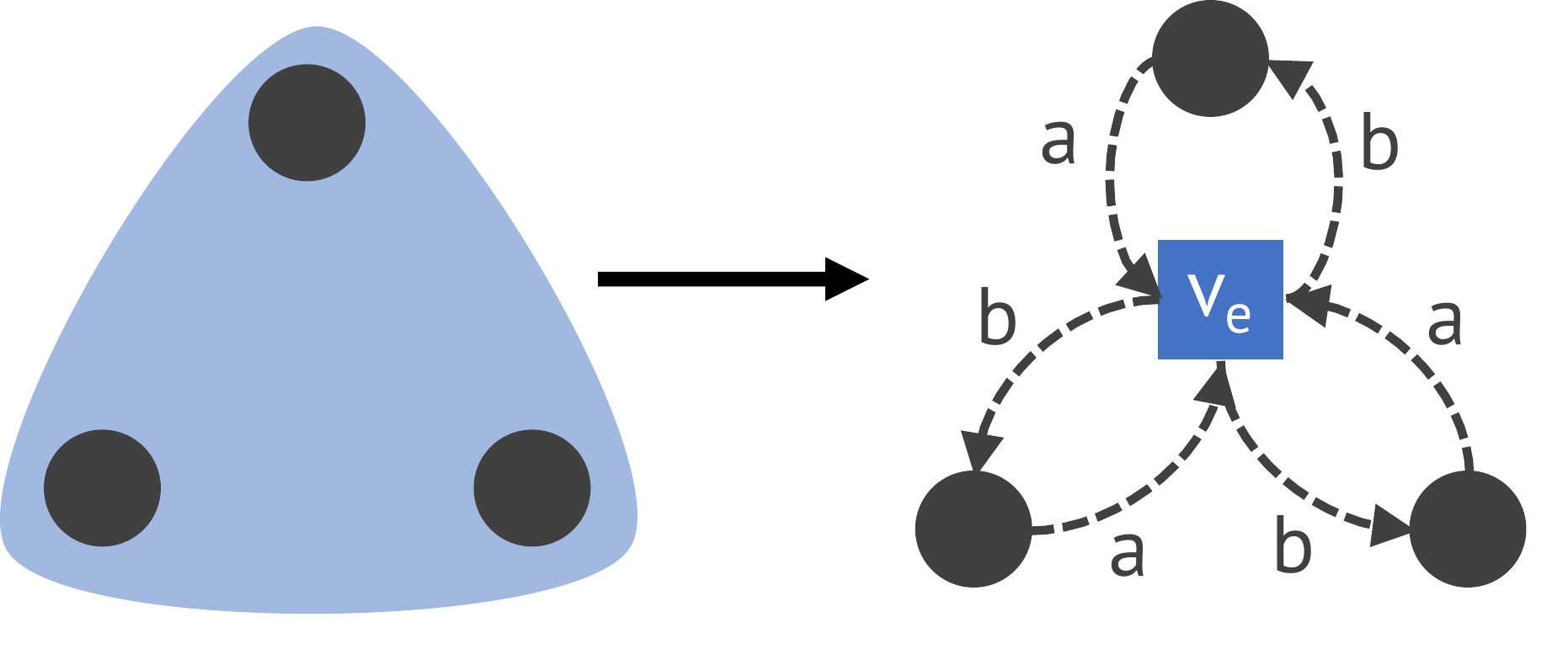}
	\end{minipage}\hfill
	\begin{minipage}{0.55\textwidth}		
		\begin{itemize}
			\item Introduce one new auxiliary nodes $v_e$.
			\item For each $v \in e$, add a directed edge $(v,v_e)$ with weight a, and a directed edge $(v_e,v)$ with weight $b$.
		\end{itemize}
	\end{minipage}
	\vspace{.25cm}
	\caption{}
\end{figure}

The gadget splitting function for the asymmetric CB-gadget is 
\begin{equation}
\label{eq:nonsymmgadgets}
\hat{\vy}_{a,b}(S) = 
\min \{|S|\cdot a,  |e\backslash S| \cdot b\}.
\end{equation}

To model an $r$-node asymmetric submodular splitting function, we use $r-1$ asymmetric CB-gadgets with carefully chosen parameters $a$ and $b$. For $j \in \{1, 2, \hdots , r-1\}$, define the $j$th gadget to be an asymmetric star gadget with parameters $a = r-j$ and $b = j$.
Let $\hat{y}_i^{(j)}$ denote the penalty that gadget $j$ assigns to a hyperedge split with $i$ nodes on the source-side of a split. From the construction of the $j$th gadget and from~\eqref{eq:nonsymmgadgets} we see that
\begin{equation}
\label{asgadget}
\hat{y}_i^{(j)} = \min \{i \cdot (r-j),  (r-i) \cdot j\} = 
\begin{cases}
i \cdot (r-j) & \text{if $i < j$} \\
(r-i) \cdot j & \text{if $i \geq j$.} 
\end{cases}
\end{equation}
Scaling the $j$th gadget by a multiplicative weight $c_j \geq 0$ and combining all gadgets results in a larger gadget with $r-1$ new auxiliary vertices and $2\cdot j \cdot r$ directed edges. As we did for the symmetric case in Section~\ref{sec:reduction}, let
$\vc = [c_1 \,\, c_2 \,\, \cdots \,\, c_{r-1}]^T$ store scaling weights, $\hat{\vy} = [\hy_1 \,\, \hy_2 \,\, \cdots \,\, \hy_{r-1}]$ store splitting penalties for the combined gadget, and define a matrix $\mA = (A_{ij})$ where $A_{ij} = \hat{y}_i^{(j)}$. The splitting penalties of the combined gadget are then given by the linear system~$\mA \vc = \hat{\vy}$. We illustrate both $\mA$ and its inverse when $r = 6$:
\begin{equation}
\label{matrixeq_multi}
\mA = 
\begin{bmatrix}
5  &   4   & 3  &  2   &  1\\
4   &  8   &  6  &   4   &  2 \\
3   &  6   &  9   &  6   &  3 \\
2   &  4   &  6    & 8   &  4 \\
1   &  2 &    3   &  4   &  5 \\
\end{bmatrix}\,,
\hspace{.5cm}
\mA^{-1} = 
\frac{1}{r}
\begin{bmatrix}
  2   & -1   &  0  &  0   &  0 \\
-1   &  2   & -1   &  0   &  0 \\
0   & -1   &  2   & -1   &  0 \\
0   &  0   & -1   &  2   & -1 \\
0   &  0   &  0   & -1   & 2\\
\end{bmatrix}\,.
\end{equation}
In general, the inverse of $\mA$ is a tridiagonal matrix with a value of $\frac{2}{r}$ on each diagonal entry and $- \frac{1}{r}$ along off diagonals. Inverting the system $\mA \vc = \hat{\vy}$ and constraining the right hand side to be greater than zero produces a set of inequalities that defines the class of submodular splitting functions we can model with this approach. As can be seen from the non-zero pattern in $\mA^{-1}$, this set of inequalities is
	\begin{align}
	2\hy_1 &\geq \hy_2 \\
	2\hy_j &\geq \hy_{j-1} + \hy_{j+1} \text{ for $j = 2, \hdots, r-3$} \\
	2\hy_{r-2} &\geq \hy_{r-1}.
	\end{align}
Lemma~\ref{lem:asub} proves that this set of inequalities is satisfied by every asymmetric submodular cardinality-based splitting function. Thus, submodularity is both a necessary and sufficient condition for modeling asymmetric cardinality-based splitting functions with $s$-$t$ gadgets.
\end{proof}

\subsection{Graph Reducibility of General Submodular Penalties}
Although we are primarily focused on cardinality-based splitting functions, some applications assign penalties that treat certain nodes in a hyperedge differently from others~\cite{panli2017inhomogeneous}.
Our result on the graph reducibility of submodular cardinality-based splitting functions raises a natural open question.
While Corollary~\ref{cor:submodgadget} proves that \emph{all} graph reducible hyperedge splitting functions are submodular, Theorems~\ref{thm:iffsub} and~\ref{thm:iffns} only prove the converse for cardinality-based functions, in the symmetric and asymmetric cases respectively. A natural question to ask then, is whether all submodular \hc{} problems are graph reducible. If so, how we can construct a hypergraph $s$-$t$ cut gadget for modeling an arbitrary submodular hyperedge splitting function? We state this question as an open conjecture which we believe may hold for both symmetric and asymmetric splitting functions. 
\begin{conjecture}
	\label{con:submod}
	Any submodular splitting function can be modeled by a hypergraph $s$-$t$ gadget.
\end{conjecture}
While we believe this is true, we note an important caveat. If we apply the same reduction strategy we used for cardinality-based submodular functions, then this reduction will be exponential in the size of the hyperedge. For an $r$-node cardinality-based function, the number of splitting penalties is linear in $r$, which led us to introduce a linear number of CB-gadgets, and thus a linear number of auxiliary vertices. Submodular splitting functions are instead characterized by an exponential number of splitting penalties in $r$. Thus, a direct extension of our previous techniques would require introducing $2^{O(r)}$ gadgets and auxiliary vertices. 
Even so, a positive answer to Conjecture~\ref{con:submod} would still provide a significant benefit for solving submodular \hc{} problems, since in many applications, the maximum hyperedge size $r$ is small. Even if the runtime for solving such a problem via graph reduction is exponential in $r$, for constant $r$ this would be an improvement over using generic submodular minimization algorithms. 

It is worth noting that a positive answer to Conjecture~\ref{con:submod} would provide a conceptually simple way to solve the following problem.
\begin{definition}
	Let $f$ be a symmetric submodular function with a ground set $V$, and let $i,j \in V$. The $i$-$j$ \textbf{submodular minimization problem} is the task of finding a set $A \subset V$ that contains $i$ but not $j$, such that $f(A)$ is minimized.
\end{definition}
The $i$-$j$ submodular minimization problem is another generalization of the graph $s$-$t$ cut problem, and is closely related to symmetric submodular function minimization. 
Independent of Conjecture~\ref{con:submod}, this problem can be cast as a submodular hypergraph $s$-$t$ cut problem. Let $f \colon 2^V \rightarrow \mathbb{R}$ be a symmetric submodular function. We can assume without loss of generality that $f$ is nonnegative, and that $f(V) = 0$. To see why, note that for any $A \subset V$, submodularity and symmetry imply that 
\[
f(A) + f(V\backslash A) \geq f(A\cup (V\backslash A) ) + f(A\cap (V\backslash A)) \implies f(A) \geq f(U).
\]
We could therefore create a new function $g(A) = f(A) - f(V)$, so that $g(V) = g(\emptyset) =0$, and $g(A) \geq 0$ for all $A \subset V$. Assume then that to begin with, the function $f$ is nonnegative and satisfies $f(V) = f(\emptyset) = 0$. Construct a hypergraph made up of a single edge $e$ containing all of $V$, with a submodular splitting function $\vw_e = f$. Solving the hypergraph $s$-$t$ cut problem with $s = i$ and $t = j$ will solve the $i$-$j$ submodular minimization problem. Thus, if Conjecture~\eqref{con:submod} is true, the $i$-$j$ submodular minimization problem can be cast as a graph $s$-$t$ cut problem. At least conceptually, this would provide a simpler solution than applying symmetric submodular function minimization algorithms. However, the runtime of this approach would depend on the number of auxiliary nodes required to model a submodular splitting function on a large hyperedge. 

We end this section with a proof of Conjecture~\ref{con:submod} for the special case of 3-node hyperedges, as well as a partial answer for the 4-node case.

\subsubsection*{Three-node submodular splitting functions}
Conjecture~\ref{con:submod} is true for 3-node hyperedges with general submodular splitting functions.
The 3-node symmetric case was proved using a clique expansion~\cite{panli2017inhomogeneous}; here we show that a weighted and directed star gadget is sufficient to model the more general \emph{asymmetric} case.
Let $e = \{1,2,3\}$ be a 3-node hyperedge with splitting penalties $p_1$, $p_2$, $p_3$, $p_{23}$, $p_{13}$, and $p_{12}$, where $p_X$ is the penalty for placing the set $X \subset e$ on the source-side of the cut. 
The splitting function of $e$ is submodular if and only if these penalties satisfy the following inequalities:
\begin{align*}
p_1 \leq p_{12} + p_{13},& \hspace{.5cm}  p_2 \leq p_{12} + p_{23}, \hspace{.5cm} p_3 \leq p_{13} + p_{23}, \\
p_{23} \leq p_{2} + p_3,& \hspace{.5cm}  p_{13} \leq p_1 + p_3, \hspace{.5cm}  p_{12} \leq p_1 + p_2.
\end{align*}
\begin{figure}[h]
	\centering
	\includegraphics[width=.5\textwidth]{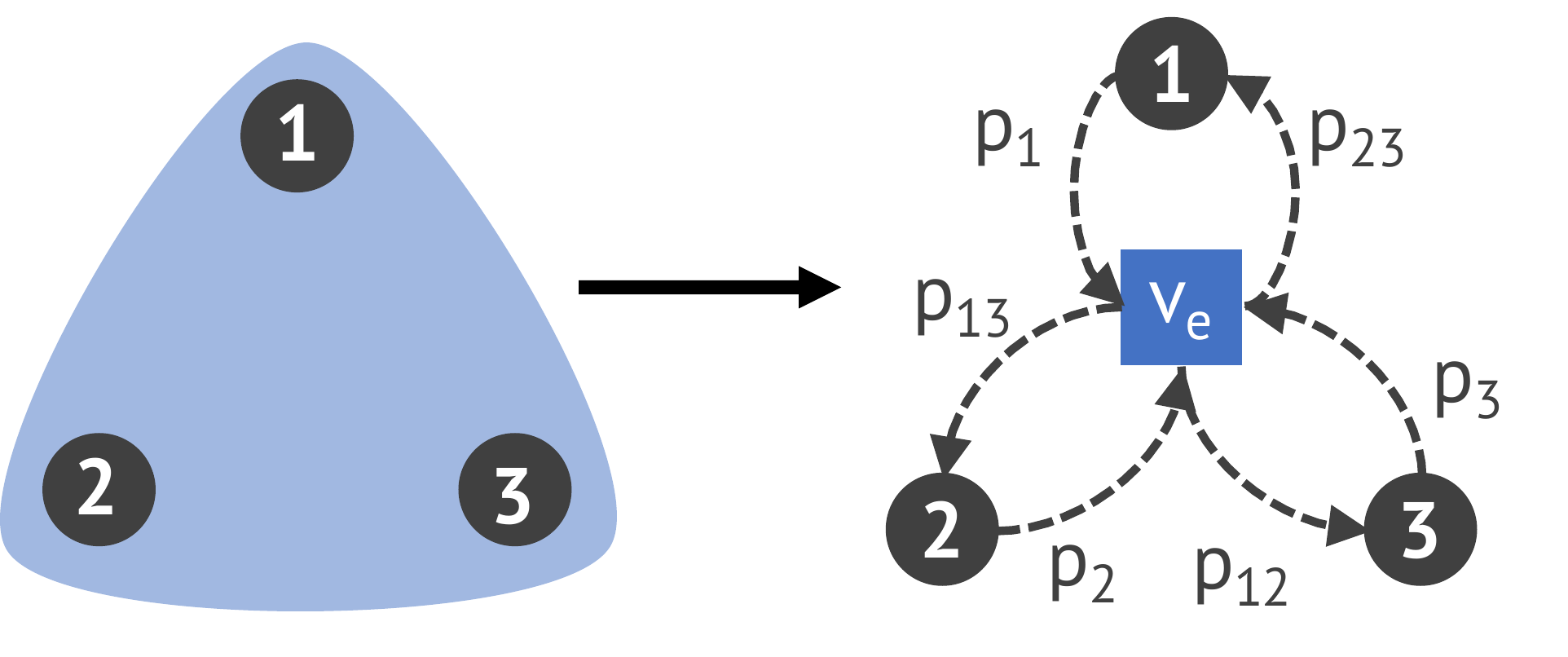}
	\caption{}	\label{fig:ssg}
\end{figure}
Given these penalty scores, we construct the weighted and directed star gadget in Figure~\ref{fig:ssg}. 
For node $i \in \{1,2,3\}$, we add a directed edge $(i,v_e)$ with weight $p_i$, and a directed edge $(v_e,i)$ with weight $p_{jk}$, where $j$ and $k$ are indices of the two other nodes in $e$. It is a simple exercise to check that any split of this gadget will produce a cut score equal to the splitting penalty of the original hyperedge. For example, consider placing node $1$ on the source-side and $\{2,3\}$ with the sink.
We must cut either the edge from $1$ to $v_e$, or the edges from $v_e$ to nodes 2 and 3.
By submodularity, $p_1 \leq p_{12} + p_{13}$, so the gadget splitting penalty will be $p_1$, as desired.
We can also see what happens for symmetric splitting functions when $p_1 = p_{23}$, $p_2= p_{13}$, and $p_3 = p_{12}$.
In this case, the gadget collapses to an weighted and undirected gadget with three edges
that models a symmetric submodular splitting function.

\subsubsection*{Four-node symmetric submodular splitting functions} 
\begin{figure}[t]
	\centering
	\includegraphics[width=.8\textwidth]{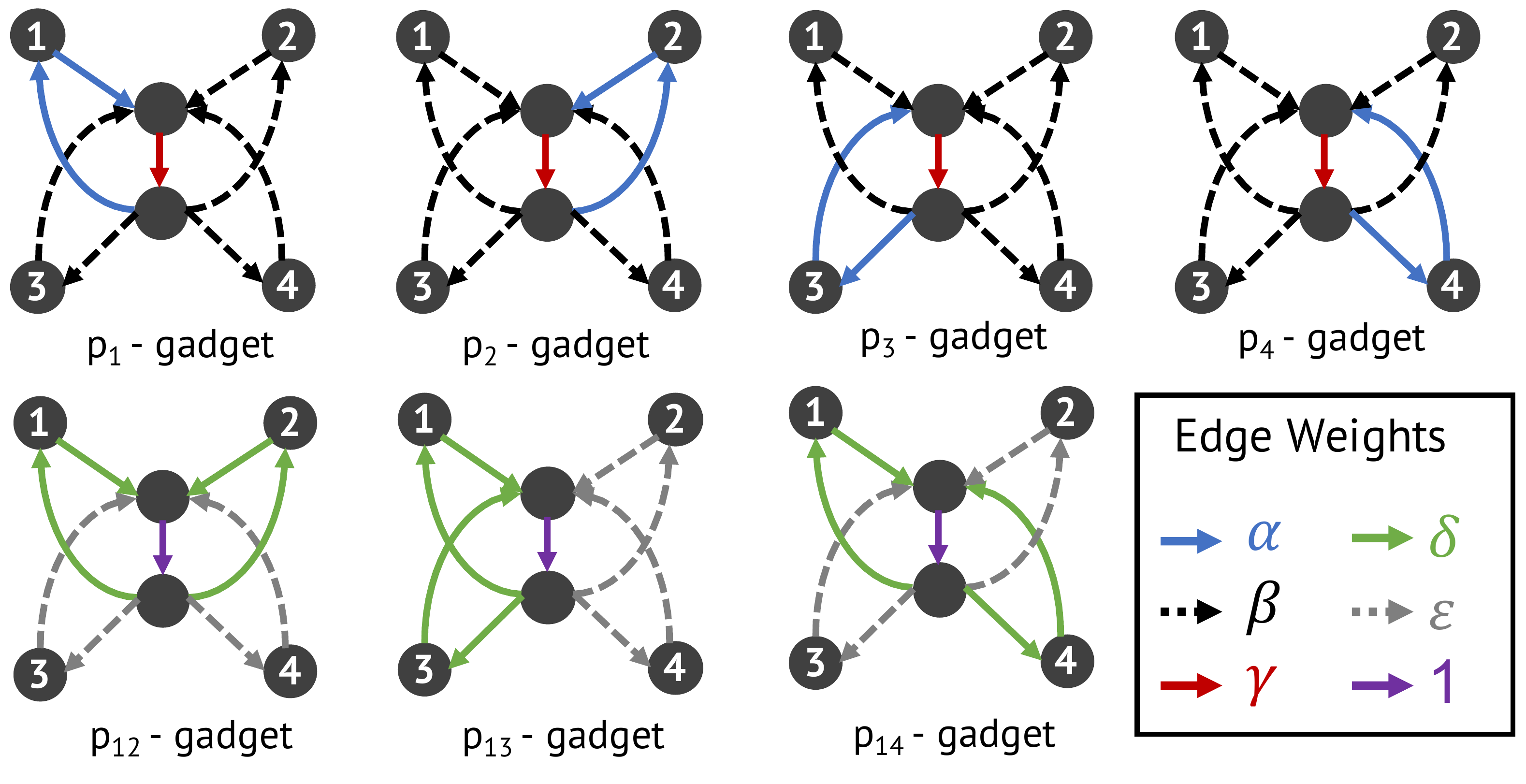}
	\caption{``Basis'' gadgets for modeling submodular hyperedges with 4 nodes. The top four are {Type-1} basis gadgets, and the bottom three are {Type-2}. The edge weight parameters are $\{\alpha, \beta, \gamma, \delta, \varepsilon\}$. After fixing edge weight parameters, we can scale the basis gadgets by nonnegative scaling weights: $c_1$, $c_2$, $c_3$, $c_4$, $c_{12}$, $c_{13}$, and $c_{14}$. This leads to a $7 \times 7$ linear system that defines the gadget splitting function for a combined gadget. If edge parameters are chosen carefully, the matrix can be inverted, leading to a system of inequalities that defines a class of submodular splitting functions that can be modeled for the chosen set of edge parameters.}
	\label{fig:four_node_basis}
\end{figure}
The complexity of the problem increases considerably for a four-node hyperedge $e = \{1,2,3,4\}$ with a general submodular splitting function. We do not have a proof of graph reducibility for this case, but we present here a partial answer for how to model 4-node \emph{symmetric} submodular splitting functions. In this case, we must assign a penalty for each set of one or two nodes that could be on the same side of the split. The splitting function is determined by $7 = 2^3 - 1$ penalties: $p_1$, $p_2$, $p_3$, $p_4$, $p_{12}$, $p_{13}$, and $p_{14}$, where $p_i$ represents the penalty for placing node $i \in \{1,2,3,4\}$ by itself, and $p_{ij}$ is the penalty for putting nodes $i$ and $j$ together. Observe that $p_{34} = p_{12}$, $p_{23} = p_{14}$, and $p_{24} = p_{13}$. 

To model this hyperedge, we introduce 7 carefully constructed \emph{submodular basis gadgets}, one for each of the possible splitting penalties for a submodular hyperedge splitting function.
After, we follow the ideas in Section~\ref{sec:reduction} and take nonnegative linear combinations of the basis gadgets.
This leads to a system of equations that is invertible under certain conditions.
Inverting the system will produce a set of inequalities that defines a class of submodular hyperedge splitting functions that can be modeled using the basis gadgets. 

Figure~\ref{fig:four_node_basis} illustrates a general class of basis gadgets. The basic edge and node structure for each gadget is the same as our CB-gadget.
However, we consider a wider range of different possible edge weights, in order to assign penalties that distinguish between different specific node subsets, rather than assigning penalties that depend only on the number of nodes on the small side of the split. In total, we use five edge parameters ($\alpha$, $\beta$, $\gamma$, $\delta$, and $\gamma$), and we define {Type-1} and {Type-2} basis gadgets. Type-1 gadgets correspond to the one-node penalties $\{p_1, p_2, p_3, p_4\}$ and have a special weight $\alpha$ for the edges adjacent to a single special node. For example, the Type-1 gadget for penalty $p_1$ has weight $\alpha$ for edges adjacent to node 1, another edge weight $\beta$ for edges adjacent to nodes 2, 3, and 4, and a third edge weight $\gamma$ for the edge between auxiliary vertices. These weights are also used for the other Type-1 gadgets. Type-2 gadgets correspond to penalties $p_{12}$, $p_{13}$, and $p_{14}$, and in a similar fashion have edge weights $\delta$ and $\varepsilon$, which distinguish between edges based on whether or not they are adjacent to a certain \emph{pair} of nodes. Given that we will take positive linear combinations of basis gadgets, we only need to be concerned with the relative scale between the edge parameters we assign. Therefore, without loss of generality we can fix the weight for one of the edge types. For this reason, for Type-2 gadgets, we always fix the weight of the edge from one auxiliary vertex to another to be 1. (We could use different edge weights for each of the basis gadgets; however, the parameter space is already challenging to navigate, so for our partial answer we restrict to using the same parameters for Type-1 gadgets, and the same parameters for all Type-2 gadgets.)

By trying different values for edge parameters and then taking linear combinations of basis gadgets, we can set up a matrix equation that defines a gadget splitting function corresponding to the combination of basis gadgets. As an example, consider setting $\alpha = \gamma = 1$, $\beta = \delta = 1/2$, and $\varepsilon = 1/4$. It is not hard to reason through the different penalties that each gadget would assign to each bipartition of the hyperedge. For the given parameters, the resulting linear system is:
\begin{equation}
\label{matrixeq_gen}
\begin{bmatrix}
1 & 	\frac{1}{2} & 	\frac{1}{2} & 	\frac{1}{2} & 	1 & 	1 & 	1   \\ 
\frac{1}{2} & 	1 & 	\frac{1}{2} & 	\frac{1}{2} & 	1 & 	1 & 	1   \\ 
\frac{1}{2} & 	\frac{1}{2} & 	1 & 	\frac{1}{2} & 	1 & 	1 & 	1   \\ 
\frac{1}{2} & 	\frac{1}{2} & 	\frac{1}{2} & 	1 & 	1 & 	1 & 	1   \\ 
\frac{1}{2} & 	\frac{1}{2} & 	\frac{1}{4} & 	\frac{1}{4} & 	\frac{1}{2} & 	\frac{3}{4} & 	\frac{3}{4}   \\ 
\frac{1}{2} & 	\frac{1}{4} & 	\frac{1}{2} & 	\frac{1}{4} & 	\frac{3}{4} & 	\frac{1}{2} & 	\frac{3}{4}   \\ 
\frac{1}{2} & 	\frac{1}{4} & 	\frac{1}{4} & 	\frac{1}{2} & 	\frac{3}{4} & 	\frac{3}{4} & 	\frac{1}{2}   \\ 
\end{bmatrix}
\begin{bmatrix}
c_1 \\ c_2 \\ c_3 \\ c_4 \\ c_{12} \\ c_{13} \\ c_{14} 
\end{bmatrix}
=
\begin{bmatrix}
\hp_1 \\ \hp_2 \\ \hp_3 \\ \hp_4 \\ \hp_{12} \\ \hp_{13} \\ \hp_{14} 
\end{bmatrix}
= \hat{\vp}.
\end{equation}
where the vector $\hat{\vp}$
is the set of penalties which completely defines the resulting gadget splitting function. For the given choice of edge parameters, the above matrix is invertible. Inverting the system and constraining the right hand side to be nonnegative gives inequalities defining the class of submodular cardinality-based splitting functions that we can model by taking nonnegative linear combinations of these specific gadgets:
\begin{equation}
\label{invmatrixeq_gen}
\begin{bmatrix}
4 & 	0 & 	0 & 	0 & 	-2 & 	-2 & 	-2 & \\ 
2 & 	2 & 	0 & 	0 & 	-2 & 	-2 & 	-2 & \\ 
2 & 	0 & 	2 & 	0 & 	-2 & 	-2 & 	-2 & \\ 
2 & 	0 & 	0 & 	2 & 	-2 & 	-2 & 	-2 & \\ 
-2 & 	1 & 	-1 & 	-1 & 	-1 & 	3 & 	3 & \\ 
-2 & 	-1 & 	1 & 	-1 & 	3 & 	-1 & 	3 & \\ 
-2 & 	-1 & 	-1 & 	1 & 	3 & 	3 & 	-1 & \\ 
\end{bmatrix}
\begin{bmatrix}
\hp_1 \\ \hp_2 \\ \hp_3 \\ \hp_4 \\ \hp_{12} \\ \hp_{13} \\ \hp_{14} 
\end{bmatrix}
= \vc.
\end{equation}
Constraining the right hand side to be nonnegative, first row of equation~\eqref{invmatrixeq_gen} says that the set of edge penalties we can model with this approach must satisfy $4\hp_1 \geq 2\hp_{12} + 2\hp_{13} + 2\hp_{14}$. Other constraints on modelable penalties can be derived from the other seven rows. Our choice of edge parameters makes it possible to model a wide range of submodular splitting functions. However, it does not accommodate \emph{all} possible submodular hyperedge splitting functions. For example, consider the following submodular splitting function:
\begin{equation}
\vw_e(S) = \begin{cases} 
0 &\text{ if $S \in \{1, \{2,3,4\} \}$}\\
2 & \text{ otherwise.}
\end{cases}
\end{equation}
This is in fact the splitting function resulting from creating a clique on nodes 2, 3, and 4, and leaving node 1 detached from the rest. While this splitting function is clearly graph-reducible, the edge penalties do not satisfy the system of inequalities defined by inverting system~\eqref{matrixeq_gen}, and therefore cannot be modeled by applying the 7 basis gadgets with $\alpha = \gamma = 1$, $\beta = \delta = 1/2$, and $\varepsilon = 1/4$. Despite systematically checking a wide range of parameter settings, we were unable to identify settings for $\{\alpha, \beta, \gamma, \delta, \varepsilon\}$ which did not lead to a similar counterexample. Furthermore, it is not clear whether it is possible cover the entire submodular region by using different edge parameter settings to cover different subregions.
The submodular region is a complicated subset of $\mathbb{R}^7$, and it is not clear how to even find a minimal set of inequalities that characterizes this space.

If it is indeed possible to model all submodular splitting functions with a single fixed set of 7 basis functions, we conjecture that this may require using different edge weights for each gadget, rather than sharing weights across Type-1 gadgets and sharing weights across Type-2 gadgets. A more sophisticated approach may be needed to determine a more versatile set of basis functions, since it is infeasible to check a wide enough range of parameter settings by brute force, even for 4-node hyperedges. However, a working solution for the 4-node case may help illuminate a useful pattern that could be generalized to larger hyperedges.

\section{NP-Hard Regimes for Hypergraph $s$-$t$ Cuts}
\label{sec:tractability}
We now turn to several hardness results for classes of hypergraph minimum $s$-$t$ cut problems, focusing specifically on \emph{symmetric} cardinality-based splitting functions. Since enforcing the symmetry constraints only makes the problem more specific, any hardness result for symmetric splitting functions automatically implies the same hardness result for the more general class of splitting functions that are not required to be symmetric.

\subsection{Hardness Results for Cardinality-Based Splitting Functions}
\label{sec:nphardcb}
Our hardness results for cardinality-based \hc{} are based on reduction from \mcut{}. Given an unweighted and undirected graph $G = (V,E)$, the \mcut{} problem seeks a bipartition of $V$ that maximizes the number of cut edges. We restrict our attention to $r$-uniform hypergraphs in which every hyperedge has the \emph{same} {cardinality-based} splitting function. We refer to the corresponding hypergraph $s$-$t$ cut problem as $r$-CB \hc{}. As before, let $w_i$ denote the penalty associated with splitting up an $r$-node hyperedge in such a way that there are $i \leq \floor{r/2}$ nodes on the small side of the split. We begin with a result for 4-uniform hypergraphs. A cardinality-based splitting function for 3-node hyperedges is simply the all-or-nothing penalty, so 4-uniform hypergraphs provide the smallest example of cardinality-based hypergraph cut problems that could be NP-hard.
\begin{theorem}
	\label{thm:4maxcut}
	The 4-CB \hc{} problem is NP-hard for any $w_2 < w_1$.
\end{theorem}
\begin{proof}
	Consider an instance of $\mcut$ given by a graph $\hat{G} = (\hat{V}, \hat{E})$. To reduce this to an instance of 4-CB \hc{}, introduce source and sink nodes $s$ and $t$, and for each $(u,v) \in \hat{E}$, introduce a 4-node hyperedge $(s,t,u,v)$ (see Figure~\ref{fig:4max}).
\begin{figure}
	\begin{minipage}[c]{0.45\linewidth}
		\centering
		\includegraphics[width=\textwidth]{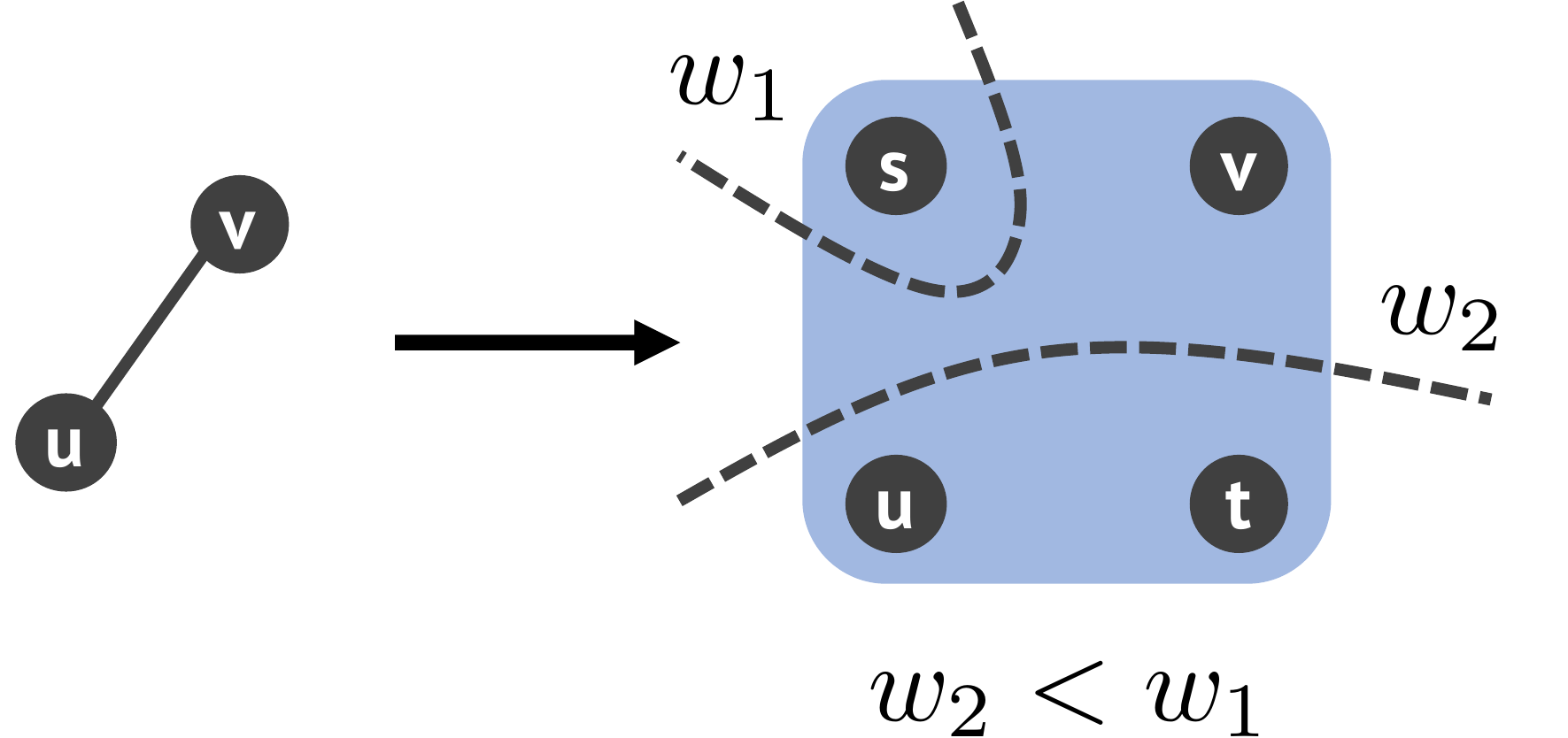}
	\end{minipage}
	\hspace{0.5cm}
	\begin{minipage}[c]{0.4\linewidth}
		\caption{We can reduce an instance of \mcut{} to 4-uniform cardinality-based hypergraph $s$-$t$ cut by introducing a hyperedge $(s,t,u,v)$ for every edge $(u,v)$ in the \mcut{} instance. For any $w_2 < w_1$, the problems are equivalent.}
	\label{fig:4max}
	\end{minipage}
\end{figure}
 The number of hyperedges in the 4-CB \hc{} problem is equal to the number of edges in $\hat{G}$, and they all must be cut in one way or another since all hyperedges involve both $s$ and $t$. Since $w_2 < w_1$, it is cheaper to separate $u$ and $v$, so that one is on the $s$-side and the other is on the $t$-side. The goal of 4-CB \hc{} is therefore equivalent to  bipartitioning the nodes $\hat{V}$ in a way that maximizes the number of $w_2$ splits, which is equivalent to maximizing the number of cut edges $(u,v) \in \hat{E}$ in $\hat{G}$. Thus, this instance of 4-CB \hc{} is equivalent to \mcut{} on $\hat{G}$.
\end{proof}

\subsubsection*{General strategy for identifying NP-hard regimes}
We can generalize the proof technique in Theorem~\ref{thm:4maxcut} to identify other NP-hard regimes for $r$-CB~\hc{} when $r > 4$. The general proof strategy is to replace each edge $(u,v)$ from an instance of \mcut{} with a small hypergraph gadget $\mathcal{H}_{uv}$, and then identify splitting penalty regimes that implicitly reward hyperedge splits that separate $u$ and $v$. Our proof of Theorem~\ref{thm:4maxcut} fits this paradigm, with a hypergraph gadget defined by a single hyperedge $(u,v,s,t)$, and splitting penalties satisfying $w_2 < w_1$. To generalize this result, we define the $(r,j)$-maxcut-gadget for an edge $(u,v)$, where $r$ denotes hyperedge size, and $j \leq \floor{r/2}$ is an index for a splitting penalty $w_j$.
\begin{definition}
	\label{def:rjgadget}
	Let $r \geq 4$ be an integer, $(u,v)$ be an edge in an instance of \mcut{}, and $\{s,t\}$ be source and sink nodes. The $(r,j)$-maxcut-gadget on $(u,v)$ is a hypergraph defined by the following disjoint auxiliary node sets and hyperedges.
\[
\begin{minipage}[c]{0.55\linewidth}
\textbf{Auxiliary Nodes}
\begin{itemize}
	\item $A$: set of $j - 2$ auxiliary nodes
	\item $B$: set of $r - j -2$ auxiliary nodes
	\item $U$: set of $j +1$ auxiliary nodes
	\item $V$: set of $r-j+1$ auxiliary nodes.
\end{itemize}
\end{minipage}\hfill
\begin{minipage}[c]{0.4\linewidth}

\textbf{Hyperedges}
\begin{itemize}
	\item $h_{st} = \{u,v,s,t, A, B \}$
	\item $h_u = \{u, B, U\}$
	\item $h_v = \{v, A, V\}$.
\end{itemize}
\end{minipage}
\]
\end{definition}
A visualization of the $(r,j)$-maxcut-gadget is given in Figure~\ref{fig:wjgadget}. By design, all hyperedges in the gadget have $r$ nodes. When reducing an instance of \mcut{} defined on a graph $\hat{G} = (\hat{V}, \hat{E})$ to an instance of $r$-CB \hc{}, the same source and sink nodes will be shared by all maxcut-gadgets. Each node $u \in \hat{V}$ will appear in multiple gadgets (equal to the number of edges $u$ is in), and each auxiliary node will be unique to the maxcut-gadget for which it was introduced. We prove the following Lemma regarding special splitting functions that can be applied to the $(r,j)$-maxcut-gadget. This generalizes the regime $w_2 < w_1$ that we considered for Theorem~\ref{thm:4maxcut}.
	\begin{figure}[h]
	\begin{minipage}[c]{0.55\linewidth}
		\centering
		\includegraphics[width=\linewidth]{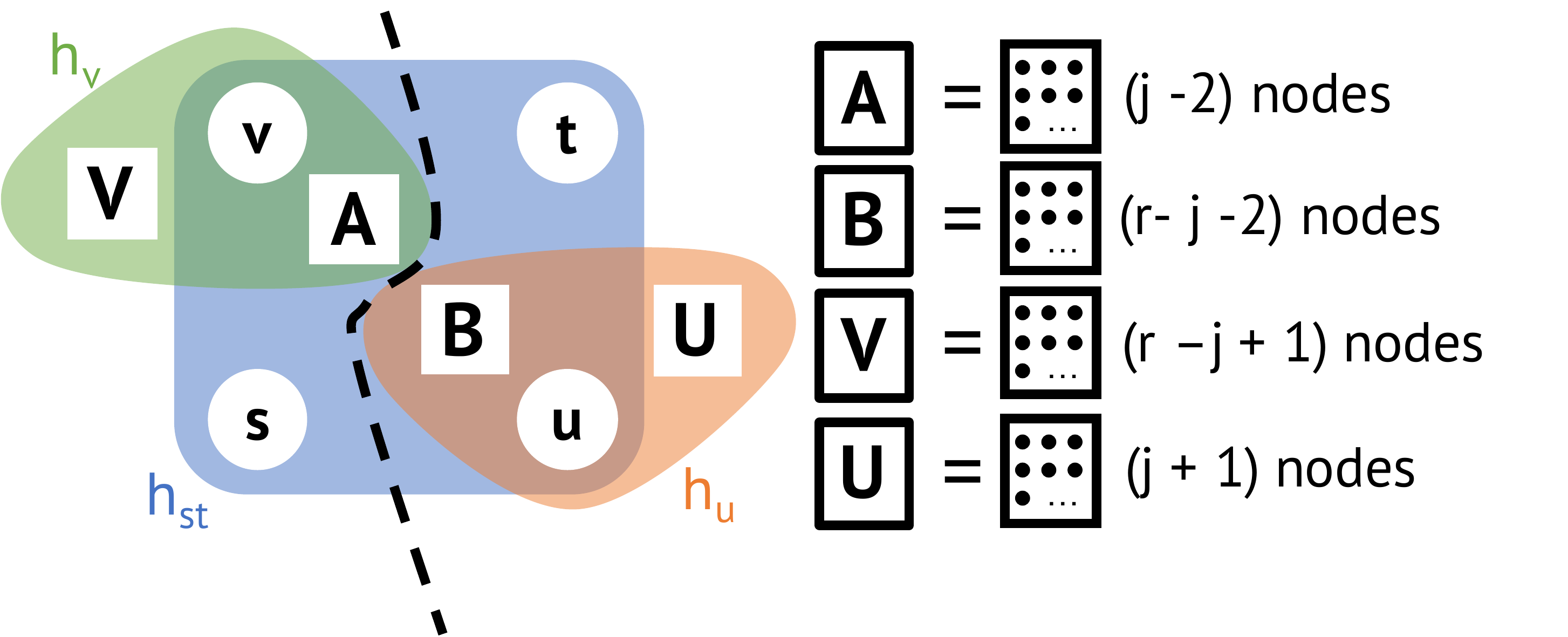}
	\end{minipage}
	\hspace{0.1cm}
	\begin{minipage}[c]{0.4\linewidth}
		\caption{If we use a splitting function satisfying~\eqref{NPhardweights}, the minimum $s$-$t$ cut in the $(r,j)$-maxcut-gadget $\mathcal{H}_{uv}$ is $w_j$, which is achieved only when $u$ and $v$ are separated. 
	When $u$ and $v$ are constrained to be together, the minimum $s$-$t$ cut is some value $y >w_j$.}
		\label{fig:wjgadget}
	\end{minipage}
\end{figure}

\begin{lemma}
	\label{lem:rj}
	For $r \geq 4$ and $j \in \{2, 3, \hdots \floor{r/2} \}$, let $\mathcal{H}_{uv}$ be the $(r,j)$-maxcut-gadget for an edge $(u,v)$ in an instance of \mcut{}. Associate each hyperedge in the gadget with a cardinality-based splitting function with penalties satisfying
	\begin{equation}
	\label{NPhardweights}
	0 < w_j  < w_1 \text{ and } w_j \leq w_i \text{ for all $i \neq 1$}.
	\end{equation} 
	Then the minimum $s$-$t$ cut in $\mathcal{H}_{uv}$ is $w_j$, which can only be achieved when $u$ and $v$ are clustered apart. If we constrain $u$ and $v$ to be on the same side of the split, the minimum $s$-$t$ cut score is a value strictly greater than $w_j$.
\end{lemma}
\begin{proof}
If $\{v,s\}$ are clustered together and $\{u,t\}$ are clustered together, then the optimal bipartition of $\mathcal{H}_{uv}$ places $\{u, t, B, U\}$ together and $\{v, s, A, V\}$ together. This leads to a penalty of $w_j$, since there are $|A| + 2 = j$ nodes on one side of the split of hyperedge $h_{st}$, and $|B| + 2 = r-j$ nodes on the other side. Hyperedges $h_u$ and $h_v$ are not split, so they incur no penalty. This is the minimum possible penalty we can incur at this gadget, since $s$ and $t$ must be split up in some way, and $w_j$ is the smallest among all splitting penalties by~\eqref{NPhardweights}. Similarly, we can show this minimum penalty of $w_j$ is achieved if instead $\{s,u\}$ are clustered together and $\{t,v\}$ are clustered together.

However, if $u$ and $v$ are both clustered with $s$ or both clustered with $t$, there are two possibilities. Either one of $\{s,t\}$ is clustered by itself, or at least one of $\{h_u, h_v\}$ will be cut in addition to $h_{st}$. The first case leads to a penalty of $w_1 > w_j$. In the second case, the penalty will be at least $2w_j > w_j$, since at least two hyperedges are cut. Given the symmetric relationship between $s$ and $t$ in $\mathcal{H}_{uv}$, this penalty will be the same whether $\{u,v\}$ are clustered with $s$ or with $t$. Thus, the minimum $s$-$t$ cut penalty is achieved only when $u$ and $v$ are clustered apart, and in every other case the penalty will be strictly greater. 
\end{proof}

Lemma~\ref{lem:rj} leads to the following generalized hardness result.
\begin{theorem}
	Every $r$-CB \hc{} problem with a splitting function satisfying property~\eqref{NPhardweights} is NP-hard.
\end{theorem}
\begin{proof}
  Let $\hat{G} = (\hat{V}, \hat{E})$ represent an instance of $\mcut$. Introduce two terminal nodes $s$ and $t$, and for each $(u,v) \in \hat{E}$, construct an $(r,j)$-maxcut-gadget with a splitting function satisfying property~\eqref{NPhardweights}. This produces a hypergraph $\mathcal{H}_r$ with $3|\hat{E}|$ hyperedges, representing an instance of $r$-CB \hc{}. The same terminal nodes $s$ and $t$ are shared across all maxcut-gadgets in $\mathcal{H}_r$, and each node $u \in \hat{V}$ shows up in exactly $d_u$ gadgets, where $d_u$ is the degree of node $u$ in $\hat{G}$. Each auxiliary vertex that is introduced shows up in exactly one maxcut-gadget. 
	
	At optimality, the $\mcut$ and $r$-CB \hc{} objectives depend only on the underlying bipartition of $\hat{V}$. This is true for \mcut{} simply by definition. The reason this also holds true for $r$-CB \hc{} is that each auxiliary node is associated with a unique $(r,j)$-maxcut-gadget. Therefore, given any bipartition of $\hat{V}$, we can arrange all auxiliary nodes in such a way that the penalty at each gadget is minimized, subject to the placement of nodes from $\hat{V}$ (each of which may appear in multiple gadgets). 
	Consider a fixed set $S \subset \hat{V}$, which defines a set of cut edges. If an edge $(u,v)$ is cut, by Lemma~\ref{lem:rj}, we know there is a penalty of $w_j$ at its gadget. If $(u,v)$ is not cut, then the penalty at its gadget is some value $y$. The value $y$ may differ depending on the splitting function. However, as long as~\eqref{NPhardweights} is satisfied, Lemma~\ref{lem:rj} guarantees that $y > w_j$. Furthermore, this value $y$ will be the same across all gadgets $\mathcal{H}_{uv}$ for which $u$ and $v$ are clustered together. Therefore, the minimum possible cut score in $\mathcal{H}_r$ that has $S$ on one side and $\hat{V} - S$ on the other, is given by
	\begin{equation}
	\label{eq:SstH}
	\textbf{Hyper-st-Cut}(S) = \sum_{uv \in \partial S} w_j + \sum_{ uv \notin \partial S } y = |\hat{E}| y - \sum_{uv \in \partial S} (y - w_j)\,,
	\end{equation}
	where $\partial S$ denotes the set of edges that cross between $S$ and $\hat{V} - S$. Meanwhile, the $\mcut$ score for the set $S$ is
	\begin{equation}
	\label{eq:maxcutobj}
	\textbf{MaxCut}(S) = |\partial S| = \sum_{uv \in \partial S} 1.
	\end{equation}
	Finally, since $(y-w_j) > 0$,
	$\argmax_S \,\, \textbf{MaxCut}(S) = \argmin_{S} \,\, \textbf{Hyper-st-Cut}(S)$.
\end{proof}



\subsection{Hardness Results in Needy-Node Hypergraphs}
Recall that for three-node hyperedges, cardinality-based splitting functions are equivalent to the all-or-nothing penalty. We briefly deviate from our study of cardinality-based hyperedge splitting functions to show that in the most general case, \hstgen{} is in fact NP-hard to approximate to within any multiplicative factor, even for three-uniform hypergraphs. We show this by first reducing a general instance of boolean satisfiability to a special case of \hstgen{} that we call \emph{Needy-Node} \hc{}. We then consider a related reduction from 3-SAT which allows us to prove this same hardness result holds even in the case of 3-uniform hypergraphs. 

\subsubsection*{Boolean satisfiability}
Boolean satisfiability (SAT) is an NP-complete problem that is frequently used in NP-hardness reductions~\cite{Garey:1990:CIG:574848}. An instance of SAT is given by a formula of boolean variables $x_1, x_2, \hdots x_n$ and their negations $\neg x_1, \neg x_2, \hdots, \neg x_n$ (both of which are called \emph{literals}) and the goal is to find an assignment of variables to \emph{true} and \emph{false} so that the formula evaluates to \emph{true}. The formula is said to be in conjunctive normal form if the formula is expressed as a conjunction of clauses, where a clause is a disjunction of literals. Informally, a formula is in conjunctive normal form if it is an AND ($\land$) of ORs ($\lor$), e.g. $(x_1 \lor x_2 \lor \neg x_3) \land (\neg x_2 \lor x_3) \land (\neg x_1 \lor x_4 \lor x_3 \lor x_2)$. In the special case where all clauses involve exactly 3 literals, the problem is known as 3SAT, and remains NP-complete. 

\subsubsection*{Needy-Node \hc} 
Let $\mathcal{H} = (V, E)$ be a hypergraph, and define hyperedge splitting functions so that every hyperedge $e \in E$ is associated with a special ``needy-node'' $z_e \in e$, such that the splitting function $\vw_e$ is
\begin{equation}
\label{eq:needynode}
\vw_e(S) = \begin{cases} 1 & \text{ if $z_e$ is clustered by itself} \\
				0 & \text{ otherwise.}
				\end{cases}
\end{equation}
In other words, a penalty is incurred at $e$ if and only if the needy node is not clustered with at least one other node from the hyperedge. A node may play the role of needy node in some hyperedges, without being the needy node for all the hyperedges it is in. Solving the minimum $s$-$t$ cut problem on $\mathcal{H}$ with these splitting functions is called \nnhc.
\begin{figure}[t]
	\centering
	\subfloat[Penalty incurred if needy node is alone.\label{fig:nn1}]
	{\includegraphics[width=.3\linewidth]{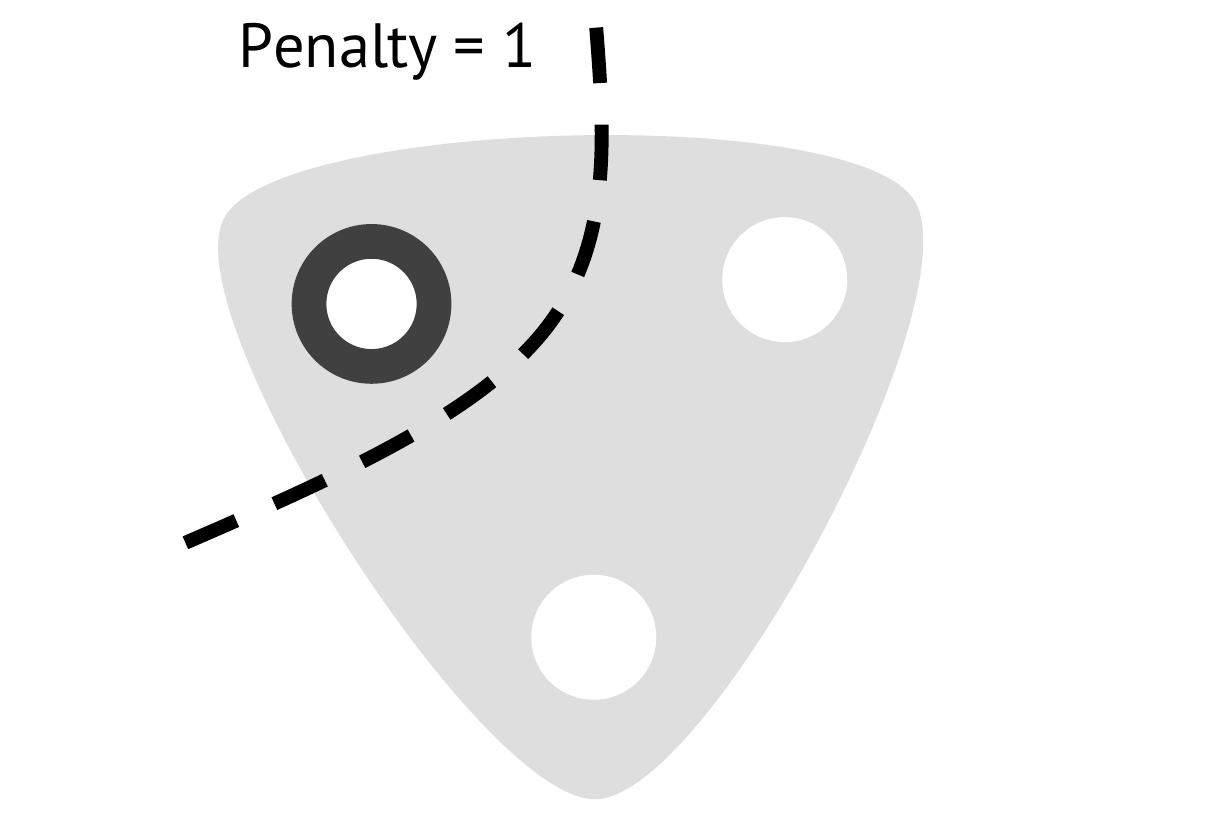}}\hfill
	\subfloat[One of $\{x_i, \neg x_i\}$ is true, the other is false. \label{fig:nn2}]
	{\includegraphics[width=.3\linewidth]{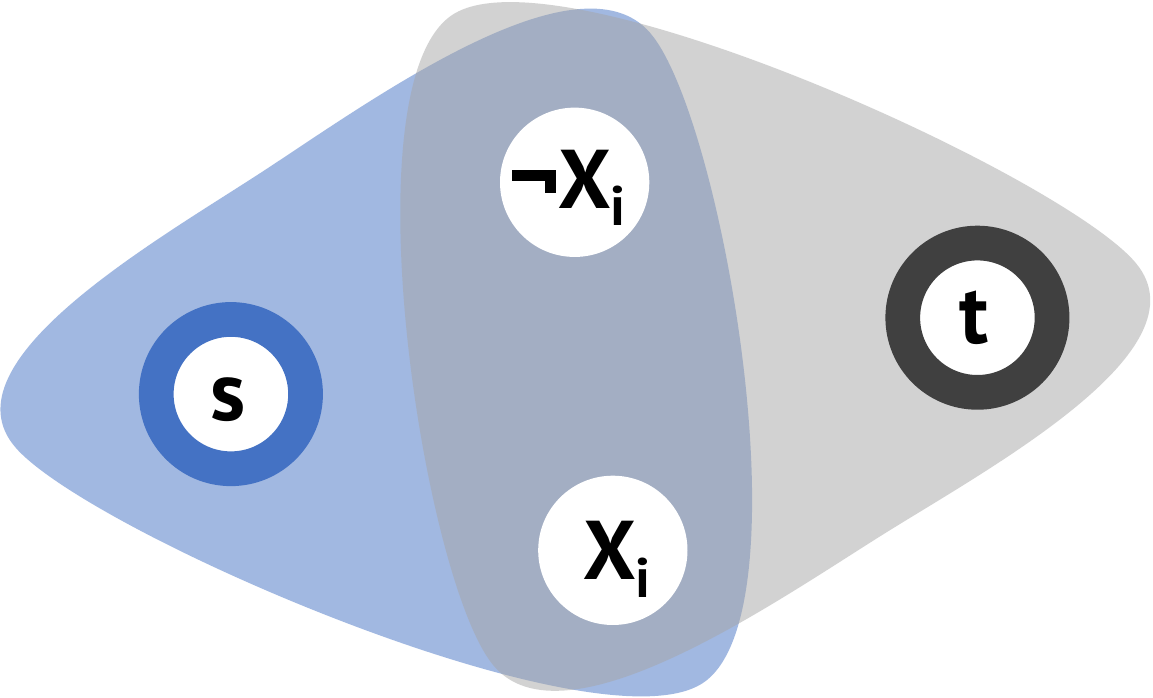}}\hfill
	\subfloat[At least one literal will be true for clause $(x_i \lor x_j \lor x_k)$ \label{fig:nn3}]
	{\includegraphics[width=.3\linewidth]{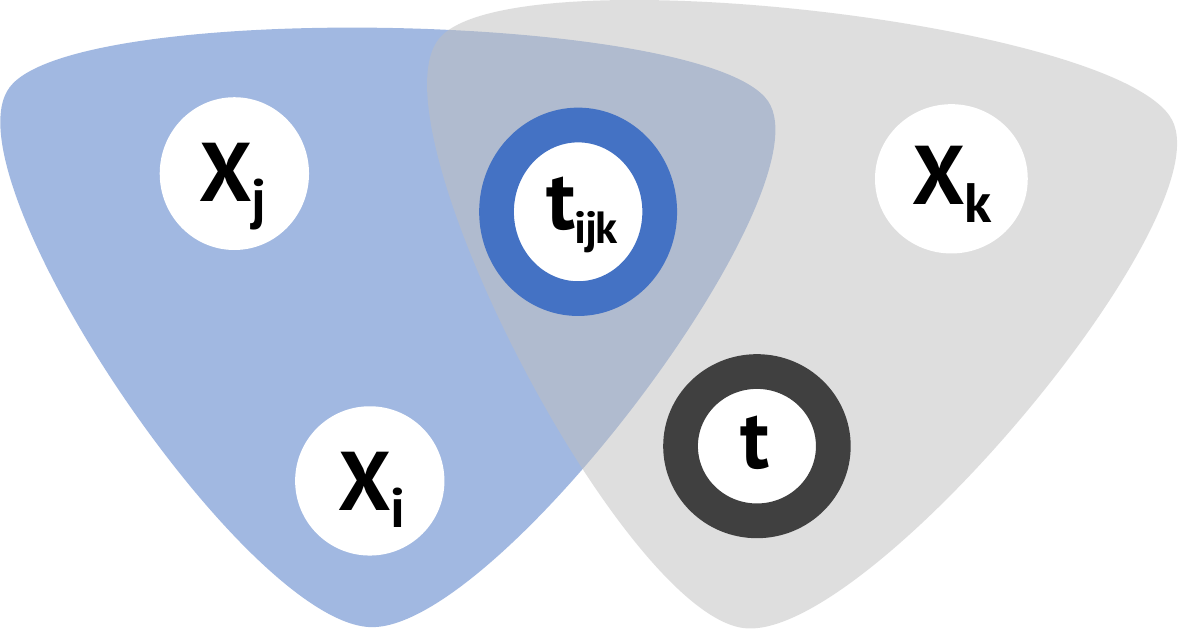}}\hfill
	\caption{Any instance of 3SAT can be cast as \nnhc{}. Each circled node is the needy-node for the hyperedge of the corresponding color. Clustering a node with the sink $t$ is equivalent to assigning \emph{true} to the literal. Clustering $x_i$ with $s$ means assigning the literal to \emph{false}.}
	\label{fig:3sat}
\end{figure}

\begin{theorem}
	\label{thm:sat}
	Any instance of SAT can be reduced in polynomial time and space to an instance of \nnhc.
\end{theorem}
\begin{proof}
Consider a SAT problem on variables $x_1, x_2, \hdots , x_n$ and their negations $\neg x_1 \neg x_2, \hdots , \neg x_n$ in conjunctive normal form. Construct an instance of \nnhc{} as follows:
\begin{itemize}
	\item Introduce a node for each literal $x_i$ and another node for $\neg x_i$.
	\item Introduce a source node $s$ and sink node $t$, corresponding to \emph{false} and \emph{true} assignments respectively.
	\item For each $(x_i, \neg x_i)$ pair, introduce a hyperedge $(s,x_i, \neg x_i)$ where $s$ is the needy-node, and another hyperedge $(t,x_i, \neg x_i)$ where $t$ is the needy-node.
	\item For each clause in the instance of SAT, introduce a hyperedge on the nodes defined for all literals in the clause, plus node $t$ as a needy node. 
\end{itemize}
With this construction, the minimum $s$-$t$ cut in the resulting graph will be zero if and only if there is a satisfying assignment for the SAT problem. The two hyperedges involving nodes $\{x_i, \neg x_i\}$ ensure that one of these nodes will be on the sink side of the cut, and the other will be on the source side. We can think of this as assigning one to \emph{true}, and the other to \emph{false}. The second type of hyperedge we introduce ensures that for every clause in the SAT problem, at least one of the nodes associated with the literals will be on the sink side of the cut, i.e.\ the \emph{true} side.
\end{proof}

We can adapt the above result in order to get an NP-hardness result even if the graph is just 3-uniform. We do this by beginning with an instance of 3SAT, which is more restrictive than general SAT but still NP-complete. Figure~\ref{fig:3sat} illustrates the construction. As before, we introduce a source and sink node $s$ and $t$, and introduce hyperedges $(x_i, \neg x_i, s)$ and $(x_i, \neg x_i, t)$ for each literal $x_i$. Any clause in the instance of 3SAT will involve only three literals. For a clause $(x_i \lor x_j \lor x_k)$ (which may also involve negative literals), we will introduce a new node $t_{ijk}$. We then define two hyperedges: $(t_{ijk}, x_i, x_j)$ with needy node $t_{ijk}$, and $(t, t_{ijk}, x_k)$ with needy node $t$. (We can arbitrarily select any two nodes from $(x_i, x_j, x_k)$ to be in the hyperedge with $t_{ijk}$, as long as we place the third node with $t$ and $t_{ijk}$ in the other hyperedge.) This construction guarantees that at least one of ($x_i$, $x_j$, $x_k$) will end up with node $t$ if we solve the minimum $s$-$t$ cut problem and get a solution with zero penalty. Thus, the Needy-Node \hc{} problem will be zero if and only if there is a satisfying assignment. We can conclude the following theorem.
\begin{theorem}
	\label{thm:3sat}
	Needy-Node \hc{} in 3-uniform hypergraphs is at least as hard as 3SAT, and is therefore NP-hard.
\end{theorem}
Note that in Theorems~\ref{thm:sat} and~\ref{thm:3sat}, we have shown that it is NP-hard to detect whether a zero solution exists for the Needy-Node \hc{} instances we have defined. Thus, for these instances, it is NP-hard to even approximate the generalized hypergraph $s$-$t$ cut problem to within any multiplicative factor.

\subsection{Tractability Regions and Open Questions}
\label{sec:openstcut}
\begin{figure}[t]
	\centering
	\subfloat[4--5 node hyperedes \label{fig:tractable2}]
	{\includegraphics[width=.31\linewidth]{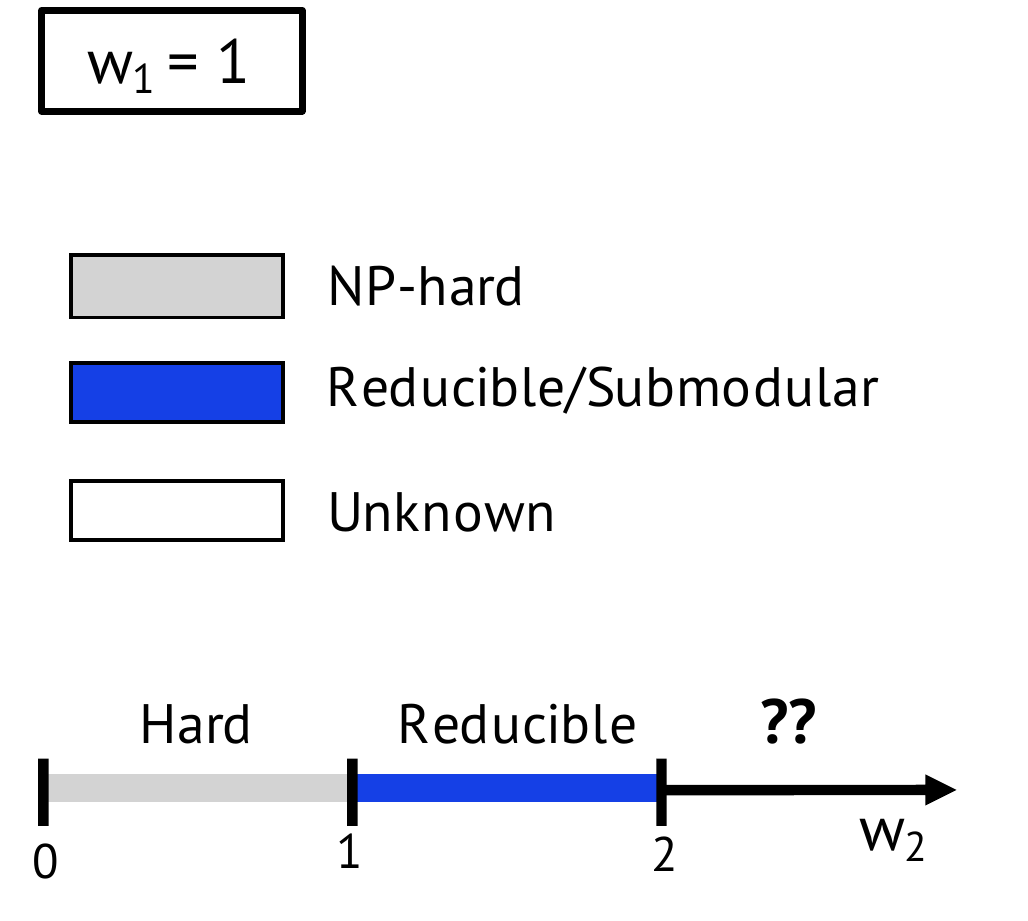}}\hfill
	\subfloat[6--7 node hyperedes \label{fig:tractable3}]
	{\includegraphics[width=.31\linewidth]{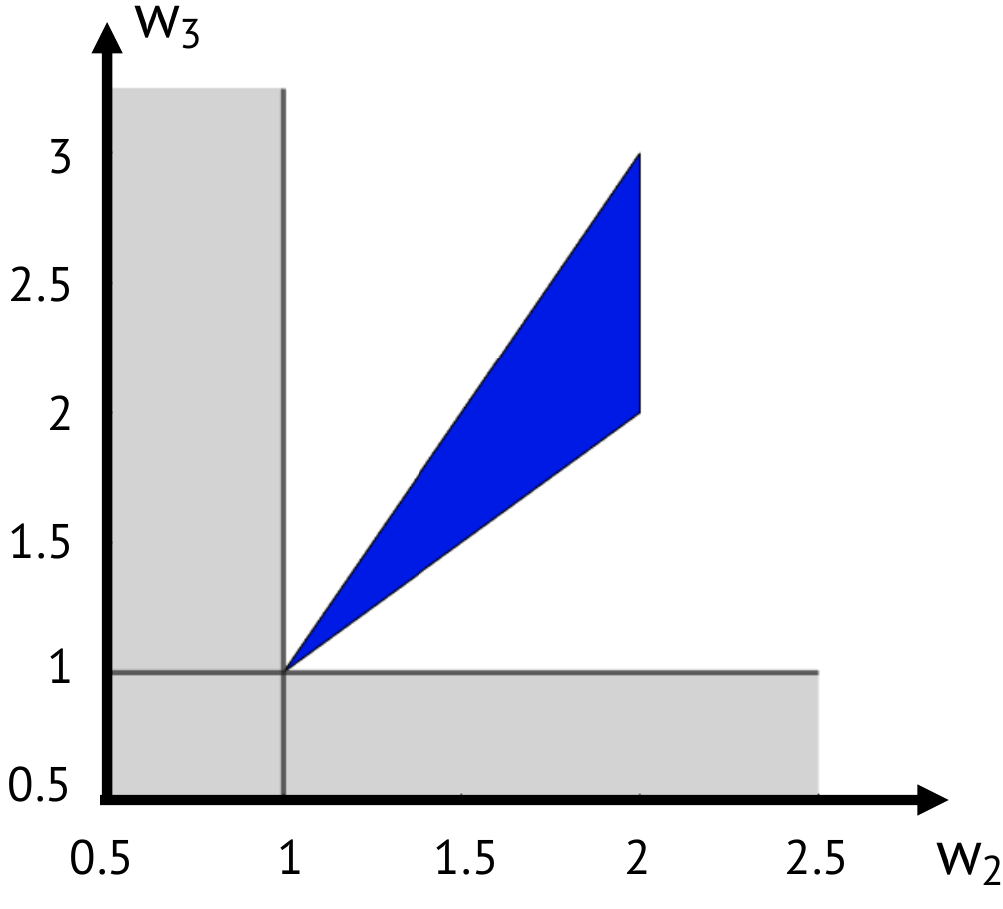}}\hfill
	\subfloat[7--8 node hyperedges\label{fig:tractable4}]
	{\includegraphics[width=.36\linewidth]{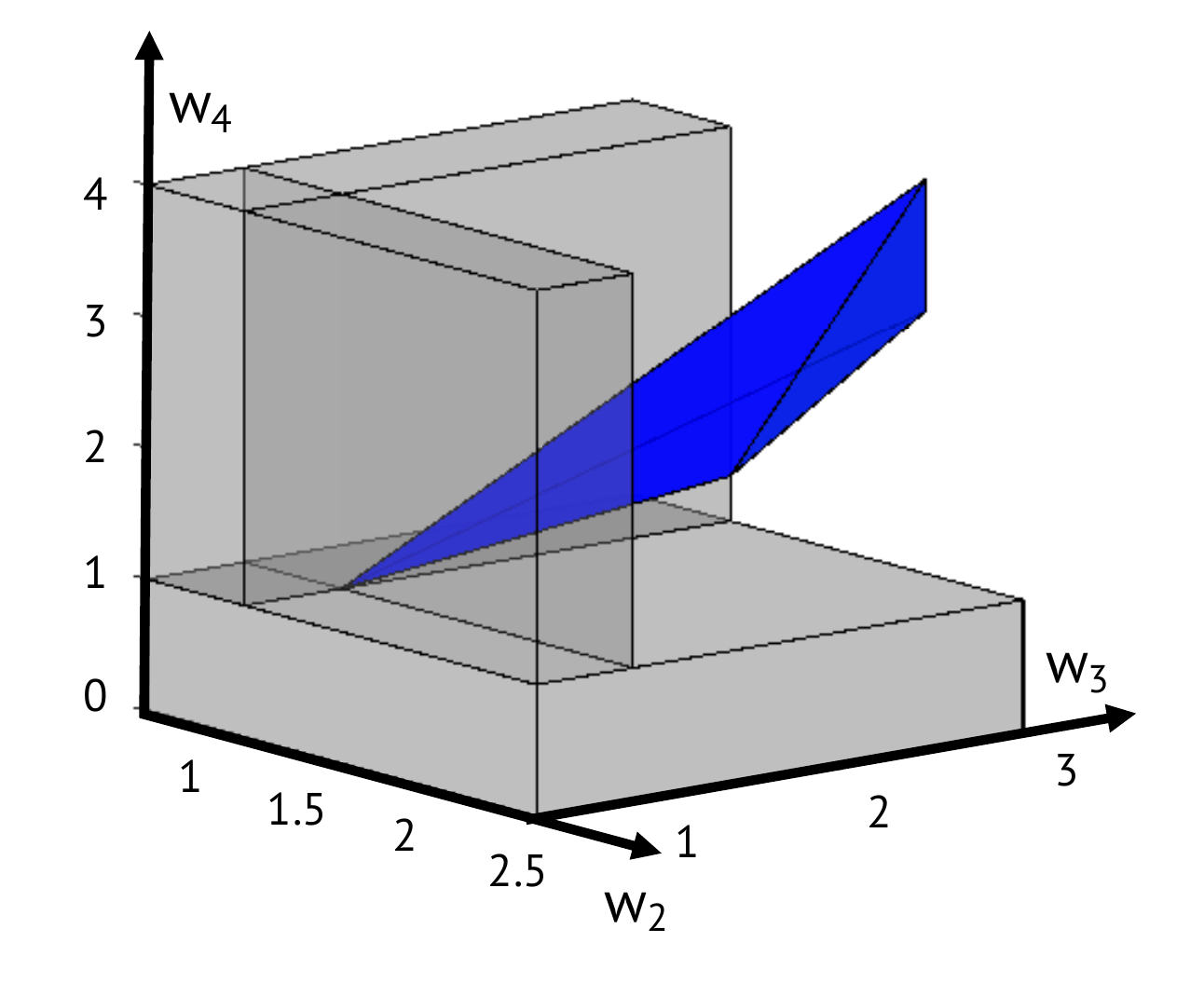}}\hfill
	\caption{The submodular region of cardinality-based hyperedge splitting functions for hyperedges of 4--5 nodes (a), 6--7 nodes (b), and 8--9 nodes (c), with $w_1 = 1$ fixed. 
		NP-hard regions for the $r$-uniform \hstgen{} problem with cardinality-based splitting functions are shaded in gray.}
	\label{fig:tractability}
\end{figure}
Cardinality-based hyperedge splitting functions on $r$-node hyperedges are completely characterized by $\floor*{r/2}$ penalty parameters. At this point we have identified regions of the parameter space that are tractable via graph reduction (the submodular region), as well as regions for which cardinality-based \hc{} is NP-hard. Figure~\ref{fig:tractability} visualizes the submodular and NP-hard regions for hyperedges between 3 and 9 nodes. 

A number of open questions arise naturally from our results on hardness, tractability, and graph reducibility for cardinality-based \hc{}. What is the complexity of solving the objective in the parameter regimes not considered thus far? Are there examples of the problem that are not graph reducible but still tractable? Are approximation algorithms possible for problems falling in the NP-hard regions? We end this section with a few observations and partial answers, though these questions remain largely unanswered and provide several clear directions for future research.


\subsubsection*{Tractability outside of the submodular region}
So far the only tractable instances of cardinality-based \hc{} that we have found are graph reducible. These exactly correspond to the submodular region, and we have shown that no instances outside of this region can be solved by graph reduction. It turns out that there is a small class of degenerate problems \emph{outside} the submodular region for which the problem is still tractable. 
\begin{definition}
\dhstgen{} is any cardinality-based \hc{} problem for which $w_1 = 0$. 
\end{definition}
\dhstgen{} is trivial to solve by placing the source (or sink) in a cluster by itself, as this has a penalty of zero.
Furthermore, if $w_i > 0$ for some $i > 1$, this parameter setting is \emph{not} submodular and therefore the problem is not graph reducible. This degenerate parameter setting and the resulting trivial solution does not lead to a particularly useful algorithm for \hstgen{}, but it has several interesting ramifications. First of all, this immediately rules out the possibility that cardinality-based \hc{} is tractable if and only if the splitting functions are submodular. More generally, it rules out the possibility that \hstgen{}  is tractable if and only if it is graph reducible.

\dhstgen{} also leads to interesting open complexity questions for the case where hyperedges have only 4 or 5 nodes, which involves only two parameters: $w_1$ and $w_2$. 
If $w_1 > 0$ and $w_2 \rightarrow \infty$, cardinality-based \hc{} on 4-uniform hypergraphs converges to the following problem:
\begin{definition}
	\label{def:nesst}
\textbf{No Even Split Hypergraph $s$-$t$ Cut} (\nesst{}) is the problem of separating a 4-uniform hypergraph into two clusters, with terminal nodes $s$ and $t$ on opposite sides, in a way that minimizes the number of cut hyperedges, and strictly avoids 2-2 hyperedge splits.
\end{definition}
A similar problem could also be defined on 5-uniform hypergraphs, where the goal is to avoid 2-3 hyperedge splits. This problem always has a solution, given that one can always place either the source or sink in a cluster by itself. \dhstgen{} and \nesst{} are in fact closely related. Fixing $w_1 = 1$ and taking $w_2 \rightarrow \infty$ leads to an instance of \nesst{}. This is nearly the same as fixing $w_2 = 1 $ and taking $w_1 \rightarrow 0^+$, except that in the limit we instead get \dhstgen{}. However, although solving \dhstgen{} is trivial, we do not know of an efficient solution to \nesst{}, nor a reduction for showing that it is NP-hard. Understanding the computational complexity of this problem is a natural next step in filling in the unknown parameter regions for cardinality-based \hc{}. 


\subsubsection*{Approximability within the NP-hard region}
Another natural question is whether we can obtain approximate solutions for cardinality-based \hc{} even in the NP-hard regions. A simple approximation guarantee of $\frac{w_1}{w_2}$ can be obtained for any problem with $w_2 < w_1$ by simply solving the all-or-nothing cut penalty (scaled by $w_1$), since the penalty at each hyperedge will then be off by at most a factor of $\frac{w_1}{w_2}$. As $w_2 \rightarrow 0$, this penalty becomes increasingly worse, though this is not a surprising fact given the following result:
\begin{theorem}
	\label{thm:5card}
Cardinality-based \hc{} on 5-uniform hypergraphs with $w_1 = 1$ and $w_2 = 0$ is NP-hard to approximate to within any multiplicative factor.
\end{theorem}
\begin{proof}
To prove the result, we will reduce an instance of Monotone Not-All-Equal 3-satisfiability (\textsc{Mon-NAE-3SAT}) to 5-uniform cardinality-based \hc{} where $w_1 = 1$ and $w_2 = 0$. An instance of \textsc{Mon-NAE-3SAT} is an instance of 3SAT in conjunctive normal form, with boolean variables $x_1$, $x_2$, $\hdots$ , $x_n$ but \emph{not} their negations (hence, \emph{monotone}). Given such an input, finding an assignment from variables to \emph{true} and \emph{false} in such a way that all clauses contain at least one \emph{true} and at least one \emph{false} variable is NP-hard~\cite{Schaefer:1978:CSP:800133.804350}. 

To reduce this problem to 5-uniform cardinality-base \hc{}, introduce a source and sink pair $s$ and $t$ as well as a node $i$ for each variable $x_i$. For each clause $(x_i \lor x_j \lor x_k)$, add a hyperedge $e = (s,t,i,j,k)$, and assign it a cardinality-based splitting function with parameters $w_2 = 0$ and $w_1 = 1$. Observe that there is a zero-penalty error for the resulting hypergraph $s$-$t$ cut problem if and only if there is a satisfying assignment for the SAT problem. If there is a satisfying assignment, then for any variable $x_i$, if $x_i$ is true, place node $i$ with the sink $t$, otherwise cluster it with the sink. Since each clause of the form $(x_i \lor x_j \lor x_k)$ has either two \emph{true} variables and one \emph{false}, or one \emph{true} and two \emph{false}, the hyperedge $(s,t,i,j,k)$ will be a 2-3 split. The other direction follows using similar arguments. Thus, it is NP-hard to detect a zero-penalty solution for 5-uniform cardinality-based \hc{}, and so the problem is NP-hard to approximate to within any multiplicative factor. 
\end{proof}
Although this result is not hard to show for 5-uniform hypergraphs, how to adjust it for a similar reduction on 4-uniform hypergraphs is unclear. Nevertheless, this results suggests that any approximation factor we can obtain for 4- and 5-uniform cardinality-based \hc{} will get increasingly worse as $w_2 \rightarrow 0$.

\section{Generalized Hypergraph Multiway Cuts}
\label{sec:multiway}
Until now we have only considered hypergraph cut problems with exactly two terminal nodes that must be separated. We now turn to the multiterminal setting, where we are given a hypergraph $\mathcal{H} = (V,E)$ with a set of $k > 2$ terminals $\{t_1, t_2, \hdots , t_k\} \subset V$. The goal of is to form $k$ clusters, with a terminal node in each cluster, in a way that minimizes the sum of hyperedge splitting penalties. This problem is NP-hard even in the graph case~\cite{Dahlhaus94thecomplexity}, where there is no ambiguity in the definition of an edge cut, though it permits several approximation algorithms~\cite{calinescu2000,karger2004rounding}. There are also several generalizations of the standard graph multiway cut objective, including directed graph~\cite{garg1994multiway,Zhao2005} and node-weighted variants~\cite{garg1994multiway,GARG200449}. The problem has also been studied in the hypergraph setting, under different hypergraph generalizations of the all-or-nothing splitting penalty~\cite{chekuri2011,ene-hmp-2014,Okumoto2012,Zhao2005}. Here we consider the goal of separating terminal nodes in order to minimize a generalized hypergraph cut function.

We begin by defining a multiway generalization of hyperedge splitting functions, which can assign penalties to hyperedge partitions involving more than two clusters. As we did in the two-terminal setting, we consider special subclasses of splitting functions that are motivated by previous work and are a natural fit for clustering applications. We prove that for a class of \emph{move-based} functions, which generalize the cardinality-based functions we considered for the two-terminal problem, the hypergraph multiway cut problem can be reduced to an instance of node-weighted multiway cut~\cite{garg1994multiway} (for which there are approximation algorithms) over a wide range of penalty parameters. However, we also identify a parameter regime for which move-based hypergraph multiway cut is NP-hard to approximate.

\subsection{Hypergraph Multiway Splitting Functions}
Consider a hypergraph $\mathcal{H} = (V,E)$. For each $e \in E$, let $\mathcal{P}_e$ denote the set of partitions (i.e., clusterings) of $e$. We denote the clusters of a partition $P \in \mathcal{P}_e$ by $(e_1, e_2, \hdots e_k)$, some of which may be empty, where $e_i \subset e$ for $i \in \{1,2, \hdots, k\}$. We will use $|P|$ to denote the number of nonempty clusters in $P$. The fact that $P$ is a partition means $e_i \cap e_j = \emptyset$ for all $i \neq j$, and $\cup_{i = 1}^k e_i = e$. 
Let $S_k$ denote the set of permutations on the set $[k] = \{1,2, \hdots, k\}$. For a partition $P = (e_1, e_2, \hdots, e_k)$, and for any $\pi \in S_k$, let $P_\pi = (e_{\pi(1)}, e_{\pi(2)}, \hdots, e_{\pi(k)})$. We can generalize the definition of a hyperedge splitting function to the multiway cut setting as follows:
\begin{definition}
	A $k$-way splitting function on $e \in E$ is any function $\vz_e : \mathcal{P}_e \rightarrow \mathbb{R}$, which for all $P = (e_1, e_2, \hdots, e_k) \in \mathcal{P}_e$ satisfies:
	\begin{align}
	\label{eq:nnmc}
	(\text{Non-negativity} ) &\hspace{.5cm} \vz_e(P) \geq 0 \\
	\label{eq:symmetricmc}
	(\text{Permutation Invariance} ) &\hspace{.5cm} \vz_e(P) = \vz_e(P_\pi) \text{ for any $\pi \in S_k$ }\\
	\label{eq:zerononcutmc}
	(\text{Non-split ignoring} )& \hspace{.5cm}  \vz_e(P) = 0 \text{ if $|P| = 1$}.
	\end{align}
\end{definition}
Importantly, if we restrict to $k = 2$ clusters, we recover our earlier definition of splitting functions for \hc{} problems~\eqref{def:splitting}. In general, we will use the term \emph{multiway splitting function} when we do not wish to specify a value for $k$. We also define another term that is useful for characterizing different types of multiway splitting functions.
\begin{definition}
	\label{def:signature}
	The {signature} of a partition $P \in \mathcal{P}_e$ is an ordered tuple of cluster sizes, ordered in decreasing size. More formally, the signature of $P$ is 
	\begin{equation}
	\label{eq:signature}
	{\normalfont \textbf{signature}}(P) = (|e_1|, |e_2|, \hdots , |e_r|)
	\end{equation}
	where $|e_i| \geq |e_{i+1}|$ for $i = 1, 2, \hdots, r-1$. 
\end{definition}
For example, the only possible signatures for 3-node hyperedge splits are $(3,0,0)$, $(2,1,0)$, and $(1, 1, 1)$. For 4-node hyperedges, the possible signatures are $(4,0,0,0)$, $(3, 1, 1,0)$, $(2, 2,0,0)$, $(2, 1, 1,0)$, $(1, 1, 1, 1)$. 
\begin{definition}
	\label{def:multicardinality2}
	A multiway splitting function $\vz_e$ is \textbf{signature-based} if 
	\begin{equation}
	\vz_e(P_1) = \vz_e(P_2) \text{ for all $P_1, P_2 \in \mathcal{P}_e$ with ${\normalfont \textbf{signature}}(P_1) = {\normalfont \textbf{signature}}(P_2)$}.
	\end{equation}
\end{definition}
We further introduce two natural subclasses of signature-based multiway splitting functions of special interest.
\begin{definition}
	\label{def:multipartition}
	\textbf{Cluster-based} multiway splitting functions depend only on $|P|$, the number of non-empty clusters of a partition $P$.
\end{definition}
\begin{definition}
	\label{def:movebased}
	\textbf{Move-based} multiway splitting functions depend only on the number of nodes from a hyperedge that are not in the largest cluster. Formally, for a partition $P = (e_1, e_2, \hdots, e_k)$ with $i_{\max} = \argmax_i \,\, |e_i|$, move-based functions depend only on $\sum_{i = 1, i\neq i_{\max}} |e_i|$.
\end{definition}
Intuitively, move-based splitting functions count the minimum number of individual node-moves that are needed for all the nodes in $e$ to be in the same cluster. 

All of these definitions remain valid if we restrict to only $k = 2$. Although move-based splitting functions are a special case of signature-based splitting functions, both are exactly equivalent to cardinality-based splitting functions if we reduce to two terminal nodes. Cluster-based splitting functions, on the other hand, simply reduce to the all-or-nothing splitting penalty if we restrict to two terminal nodes.

\subsubsection*{Examples} Table~\ref{tab:multiwaysplits} outlines a number of non-standard multiway cut penalties that have been considered in previous work all of which fit within the framework we have defined here. All of these are signature-based. All-or-nothing, sum of external degrees, $K-1$ penalty, and rainbow split are all cluster-based. All-or-nothing and rainbow split are also move-based. The discount cut is move-based when $\alpha = 1$, and is neither move-based nor cluster-based otherwise.
\renewcommand{\arraystretch}{1.75}
\begin{table}[t]
	\caption{Examples of hyperedge cut penalties that fit within our multiway splitting function framework (Rainbow split is a special case of rainbow labeling in hypergraphs~\cite{mirzakhani2015} when interpreting labels as cluster assignments). For discount cut, $i_{\max} = \argmax_i  |e_i|$, and $\alpha$ is a parameter between $0$ and $1$.}
	\label{tab:multiwaysplits}
	\centering
	\begin{tabular}{l ll}
		\toprule 
		All-or-nothing & $\vz_e(P) = \begin{cases} 0 & \text{if $|P| = 1$} \\ 1 & \text{otherwise}\end{cases}$  & \cite{BensonGleichLeskovec2016,hadley1995,ihler1993modeling,lawler1973} \\
		Sum of External Degrees & $\vz_e(P) = \begin{cases} 0 & \text{if $|P| = 1$}\\ |P| & \text{ otherwise} \end{cases} $ &  \cite{alpert1995survey,chekuri2011,Karypis:1999:MKW:309847.309954} \\
		$K-1$ Penalty & $\vz_e(P) = |P|-1$ & \cite{yaros2013imbalanced} \\
		Discount Cut & $\vz_e(P) = \sum_{i = 1, i\neq i_{\max}}^{k} |e_i|^\alpha$,  &  \cite{yaros2013imbalanced}\\
		Rainbow Split & $\vz_e(P) = \begin{cases}
		1 & \text{ if $|e| = |P|$}\\
		0 & \text{ otherwise}.
		\end{cases}$  & \cite{mirzakhani2015}\\
		\bottomrule
	\end{tabular}
\end{table}

\subsubsection*{Splitting function parameters}
It is worthwhile to consider the maximum possible number of parameters needed to completely characterize each different type of splitting function. Let $k$ represent the number of terminals in a multiway cut problem and $r$ be the number of nodes in a hyperedge. Assume that $k \geq r$. The number of parameters needed to characterize a general multiway splitting function on $r$ nodes equals the $r$th Bell number~\cite{Bell1934} --- the number of ways to partition $r$ objects into nonempty clusters. Bell numbers grow extremely quickly; for example, 52 parameters are needed to characterize a general multiway splitting function when $r= 5$, and 203 parameters are needed when $r = 6$. Meanwhile, a signature-based multiway splitting function on an $r$-node hyperedge is associated with up to $p_r$ distinct weights, where $p_r$ is the number of ways to write the integer $r$ as a sum of positive integers. Although signature-based functions are significantly more restrictive than general multiway splitting functions, $p_r$ still grows exponentially in the square root of $r$~\cite{andrews1998theory}. Eight-node hyperedges require 22 parameters in the worst case, and nine-node hyperedges require 30. In contrast, move-based and cluster-based multiway splitting functions are significantly more general than the all-or-nothing function, but each can be characterized by $r-1$ penalty parameters. However, although these are both characterized by the same number of penalties, these penalties are not always applied in the same way to different splits of a hyperedge.
\renewcommand{\arraystretch}{1}
\begin{table}[t]
		\caption{Penalty parameters for signature-based splitting functions.}
	\label{tab:splits}
	\centering
		\begin{tabular}{lllll}
                  \toprule
			& 		\textbf{Signature} & 	\textbf{Signature-based} & 	\textbf{Cluster-based} & 	\textbf{Move-based} \\
			\midrule
			\textbf{3 nodes}& 	$(1,1,1)$ &  $g_{1,1,1}$ & $h_3$ & $m_2$ \\
			&	$(2,1,0)$ &  $g_{2,1}$ & $h_2$ & $m_1$ \\
			\midrule 
			\textbf{4 nodes} & $(1,1,1,1)$ &  $g_{1,1,1,1}$ & $h_4$ & $m_3$ \\
			&$(2,1,1,0)$ &  $g_{2,1,1}$ & $h_3$ & $m_2$ \\
			&$(2,2,0,0)$ &  $g_{2,2}$ & $h_2$ & $m_2$ \\
			&$(3,1,0,0)$ &  $g_{3,1}$ & $h_2$ & $m_1$\\
			\bottomrule
		\end{tabular}
	\end{table}
For signature-based splitting functions, let $g_{s}$ denote the penalty associated with a partition with signature $s$. For cluster-based, let $h_t$ be the penalty a partition $P$ with $t$ nonempty clusters. For move-based, let $m_i$ denote the penalty for placing $i$ nodes outside of the largest cluster. Table~\ref{tab:splits} shows the correspondence between these penalties and different signatures on 3- and 4-node hyperedges. For 3-node hyperedges, there is no difference between different types of signature-based splitting functions. Each has a different penalty parameter associated with each signature. However, for 4-node hyperedges, signature-based functions are more general than move-base and cluster-based functions, and these latter two penalize different signatures differently. 

\subsection{The Generalized Hypergraph Multiway Cut Problem}
Let $\mathcal{H} = (V,E)$ be a hypergraph with a multiway splitting function $\vz_e$ for each $e \in E$, and let  $\{t_1, t_2, \hdots, t_k\} \subset V$ be designated terminal nodes, with $k > 2$. Let $\mathcal{P}_V$ denote the set of partitions of $V$, and for a partition $P \in \mathcal{P}_V$, let $V_i$ denote the $i$th cluster of $P$. For any $P = \{V_1, V_2, \hdots , V_k\} \in \mathcal{P}_V$, define the restriction of $P$ to a hyperedge $e$ to be:
\begin{equation}
\label{eq:restriction}
P\big|_e = (V_1 \cap e, V_2 \cap e, \hdots , V_k \cap e ) \in \mathcal{P}_e\,.
\end{equation}

\begin{definition}
Generalized hypergraph multiway cut (\ghmc{}) is the following optimization problem:
	\begin{equation}
	\label{eq:hypermc}
	\begin{array}{ll}
	\minimize_{P = \{V_1, V_2, \hdots , V_k\} \in \mathcal{P}_V} & \sum_{ e \in E} \vz_e\big(P \big|_e\big) \\
	{\normalfont\text{subject to} } & t_i \in V_i \,.
	\end{array}
	\end{equation}
\end{definition}
Two existing hypergraph generalizations of the graph multiway cut problem are captured as special cases of \ghmc{}. Standard hypergraph multiway cut (standard \hmc{}) is the problem of removing a minimum weight set of hyperedges to separate $k$ terminal nodes in a hypergraph~\cite{chekuri2011,Okumoto2012}. This is a special case of \ghmc{} where the all-or-nothing multiway splitting penalty is used. The hypergraph multiway partition problem (\hmp{}) differs in that the cost at a cut hyperedge is proportional to the number of clusters spanned by the hyperedge~\cite{chekuri2011,ene-hmp-2014}. This can be viewed as an instance of \ghmc{} when the sum of external degrees splitting function is used for all hyperedges.
\hmc{} is approximation equivalent to the node-weighted multiway cut in graphs~\cite{Okumoto2012}, for which a $(2-\frac{2}{k})$-approximation is known~\cite{garg1994multiway,GARG200449}. For \hmp{}, the best known approximation factor is $\frac43$~\cite{ene-hmp-2014}. Both \hmc{} and \hmp{} are special cases of the submodular multiway partition objective (\smp{})~\cite{Zhao2005}. \hmp{} is specifically an instance of \emph{symmetric} \smp{}, which has a known $(\frac{3}{2} - \frac{1}{k})$-approximation~\cite{chekuri2011}.

When we wish to solve \ghmc{} with the same type of splitting function applied to all hyperedges, we will refer to the problem as \hmc{}, preceded by the type of splitting function. For example, we will refer to the ``\ghmc{} problem with move-based splitting functions on all hyperedges'' simply as \emph{move-based} \hmc{}. The standard hypergraph multiway cut problem can also be referred to as \emph{all-or-nothing} \hmc{}.

\subsection{Graph Reducibility of Move-Based Multiway Cuts}
Just as we focused on cardinality-based \hc{} in the two-cluster setting, we focus on move-based \hmc{} for multi-cluster problems. Although we require several new techniques to address the multiway problem, our high-level approach and results for \hmc{} closely mirror our \hc{} results. For two-way cuts, we generalized Lawler's technique for reducing all-or-nothing \hc{} to directed graph $s$-$t$ cut, by defining a notion of hypergraph $s$-$t$ gadgets. These gadgets made it possible to model all submodular cardinality-based splitting functions. For the multiway case,
all-or-nothing \hmc{} can be reduced to \nmc{} by applying a type of star expansion to hyperedges~\cite{Okumoto2012}. Inspired by this technique, we define a \nmc{}-gadget that can model a wide range of move-based multiway splitting functions. We show that all move-based functions satisfying a certain submodularity property can be reduced to \nmc{} using a combination of these gadgets.

\subsubsection*{Node-Weighted-MC Gadgets}
First, we define \nmc{}.
\begin{definition}
	\label{def:nmc}
	Let $G = (V,E)$ be an undirected, node-weighted graph where $\omega(v) \geq 0$ is the weight of $v \in V$. If $\{t_1, t_2, \hdots, t_k\}$ is a set of terminal nodes in $V$, the \nmc{} problem seeks a minimum weight set of nodes to remove in order to separate all terminal nodes from each other.
\end{definition}
Given an instance of \nmc{}, we will use $G\backslash R$ to denote the graph obtained by removing a set of nodes $R$, and let $\omega(R) = \sum_{v\in R} \omega(v)$. We can convert a hypergraph $\mathcal{H} = (V,E)$ into an instance of \nmc{} by replacing each hyperedge with a \nmc-gadget (NMC-gadget).
\begin{definition}
	\label{def:nmcgadget}
	An NMC-gadget for a hyperedge $e \in E$ is a node-weighted, undirected graph $G' = (V',E')$, where $V' = e \cup \hat{V}$ for a set of auxiliary nodes $\hat{V}$, and
        $\omega\colon V \rightarrow \{\mathbb{R}_+ \cup \infty\}$ is a node-weighting function satisfying $\omega(v) = \infty$ for all $v \in e$. The gadget comes with an NMC-gadget splitting function
        $\hat{\vz}_e(P)\colon \mathcal{P}_e \rightarrow \mathbb{R}_+$ defined by
	\begin{equation}
	\label{eq:nmcsplitting}
	\begin{array}{lll}
	\hat{\vz}_e(P)= & \minimum_{R \subseteq \hat{V}} & \omega(R)  \\
	& {\normalfont\text{subject to} }& \text{$\{e_i,e_j\}$ are disconnected in $G'\backslash R$ for all $i \neq j$},
	\end{array}
	\end{equation}
	where $\{e_1,e_2, \hdots, e_k\}$ are the clusters of $P$.
\end{definition}
We will say that a multiway splitting function $\vz_e$ can be modeled by an NMC-gadget if that gadget has a splitting function $\hat{\vz}_e = \vz_e$.
An instance of \ghmc{} defined on a hypergraph $\mathcal{H} = (V,E)$ is \emph{NMC-reducible} if for each $e \in E$, $\vz_e$ can be modeled by an NMC-gadget.

\subsubsection*{NMC-reduction for all-or-nothing \hmc}
Okumoto et al.\ showed that all-or-nothing \hmc{} can be reduced in an approximation-preserving way to \nmc{}~\cite{Okumoto2012}. The details of this reduction can be easily described using the terminology and framework we have developed here. For the reduction, each hyperedge $e \in E$ is replaced with an NMC-gadget involving a single auxiliary vertex $\hat{V} = \{v_e\}$, which is connected to each $v \in e$ via an undirected edge. The node-weight of $v_e$ is assigned to be the weight of the hyperedge (and by our definition, $\omega(v) = \infty$ for $v \in e$). In the resulting instance of \nmc{}, deleting an auxiliary vertex $v_e$ corresponds to removing an edge $e$ from the original hypergraph. The cost of removing a minimum weight set of auxiliary nodes in the node-weighted graph to separate terminal nodes is exactly equivalent to removing a minimum weight set of hyperedges in the hypergraph to separate terminal nodes. 

\subsubsection*{NMC-reduction for move-based \hmc}
In general, there is no existing notion of a \emph{multiway} submodular function. Despite this, we can identify a simple property that can be satisfied by move-based multiway splitting functions, and that is related to the definition of submodularity for cardinality-based two-way splitting functions. Let $e \in E$ be a hyperedge with a move-based splitting function $\vz_e$. Since $\vz_e$ depends only on the size of the largest cluster of a partition, we can associate $\vz_e$ with a simplified function $\vm_e\colon 2^e\backslash e \rightarrow \mathbb{R}_+$ defined by
\begin{equation}
\label{eq:mb-2way}
\vm_e(S) = \vz_e(P_S)\,,
\end{equation}
where $P_S$ is any partition of $e$ whose largest cluster size is  $|e\backslash S|$. In other words, $S$ is a minimal size set of nodes that must be moved in order to avoid splitting the hyperedge. Note that $\vm_e$ is asymmetric in the sense that $\vm_e(S)$ can be different from $\vm_e(e\backslash S)$ for two nonempty subsets $S$ and $e \backslash S$, unless $|S| = |e \backslash S|$. 

A natural extension of our results in Section~\ref{sec:positive} is to explore what it means for $\vm_e$ to be submodular. In particular, does submodularity have any bearing on the complexity of move-based \hmc{}? One difficulty in answering this question is that $\vm_e$ is not defined on $S = e$, since it is impossible to cluster \emph{all} nodes in $e$ away from the largest cluster of $e$. We can overcome this issue by selecting a value $m \geq 0$ and extending the definition of $\vm_e$ so that $\vm_e(e) = m$. We can then check what it means for $\vm_e$ to be submodular, even if a move-based splitting function will never in fact assign a penalty of $m$. In principle, we can set $m$ to be any nonnegative value, though different choices of $m$ will lead to different definitions of submodularity for move-based splitting functions. We identify one setting of this parameter that is unintuitive at first glance, yet comes with a several desirable properties.
\begin{lemma}
	\label{lem:multisub}
	Let $e = \{v_1, v_2, \hdots, v_{r}\}$ be an $r$-node hyperedge associated with a move-based splitting function defined by
	\begin{equation}
	\label{eq:mrsubmod}
	\vm_e(S) = \begin{cases}
	m_{i} & \text{ if $S \subset e$ and $|S| = i < r$}\\
	m_r & \text{if $S = e$.}
	\end{cases}
	\end{equation}
	If $m_r = m_{r-2}$ and $\vm_e$ is submodular, then it satisfies the following inequalities:
		\begin{align}
	\label{ps1}
	2m_1 &\geq m_2 \\
	\label{ps2}
	2m_j &\geq m_{j-1} + m_{j+1} \,\, \text{ for $j \in \{1, 2, 3, \hdots , r-3\}$}  \\
	\label{eq:mon}
	m_{j+1} & \geq m_{j} \,\, \text{ for $j \in \{1, 2, 3, \hdots , r-2\}$} \,.
	\end{align}	
\end{lemma}
\begin{proof}
	As in the proof of Lemma~\ref{lem:asub}, we must simply identify certain sets $A,B \subseteq 2^e$, and check the definition of submodularity:
	\begin{equation}
	\label{eq:submodul}
	\vm_e(A) + \vm_e(B) \geq \vm_e(A \cap B) + \vm_e(A\cup B).
	\end{equation}
Inequality~\eqref{ps1} follows by setting $A = \{v_1\}$, and $B = \{v_2\}$, while inequality~\eqref{ps2} follows when $A = \{v_1, v_2, \hdots, v_j\}$ and $B = \{v_2, v_3,\hdots , v_{j+1}\}$. 
To prove inequality~\eqref{ps2}, let $A = \{v_1, v_2, \hdots, v_{r-1}\}$ and $B = \{v_2, v_3, \hdots, v_r\}$. Submodularity implies that $2 m_{r-1} \geq m_{r-2} + m_r$, which reduces to $m_{r-1} \geq m_{r-2}$ when we substitute $m_r = m_{r-2}$. Furthermore, this implies the full set of inequalities $m_{j+1} \geq m_j$ for $j \leq r-2$. This can be seen by essentially reversing arguments used in Lemma~\ref{lem:cbsub}.
Let $\vm = [m_1 \,\, m_2 \,\, \cdots \,\, m_{r-1}]^T$ represent a set of penalties satisfying inequalities~\eqref{ps1},~\eqref{ps2}, and $m_{r-1} \geq m_{r-2}$. We encode these inequalities into a matrix equation to show that there exists a vector $\vc = [c_1 \,\, c_2\,\, \cdots \,\, c_{r-1}]^T$ satisfying
\begin{equation}
\label{mtoc}
\begin{bmatrix}
2 & -1  & \cdots & 0 & 0\\
-1 & 2 & \cdots & 0 & 0\\
\vdots  & \vdots  & \ddots & \vdots  & \vdots\\
0 & 0 &\cdots & 2 &-1 \\
0 &  0 &\cdots & -1 &1 \\
\end{bmatrix}
\begin{bmatrix}
m_1 \\ m_2  \\ \vdots \\m_{r-2} \\ m_{r-1}
\end{bmatrix}
= \begin{bmatrix}
c_1 \\ c_2  \\\vdots\\ c_{q-1} \\ c_q
\end{bmatrix}
\geq
\begin{bmatrix}
0 \\ 0  \\\vdots\\ 0 \\ 0
\end{bmatrix}\,.
\end{equation}
The tridiagonal matrix above is the same one we encountered in Lemma~\ref{lem:cbsub}. By taking the inverse of this matrix, we can conclude that the vector $\vm$ is given by $\mA \vc = \vm$, where just as in Lemma~\ref{lem:cbsub}, $\mA$ is a matrix whose $ij$ entry is $A_{ij} = \min \{i,j\}$. By the structure of $\mA$ and the nonnegativity of $\vc$, we see that the entries in $\vm$ are non-decreasing, and thus inequality~\eqref{eq:mon} is satisfied.
\end{proof}
Although setting $m_r = m_{r-2}$ was at first unintuitive, this produced a natural connection between submodularity for cardinality-based splitting functions and submodularity for move-based multiway splitting functions. Furthermore, constraint~\eqref{eq:mon} encodes a type of monotonocity that is natural for move-based splitting functions. Namely, clustering a larger number of nodes away from the largest cluster in $e$ should be associated with a higher splitting penalty. The following theorem provides further motivation for setting $m_r = m_{r-2}$, as it shows that under this definition of submodularity, move-based \hmc{} is NMC-reducible.
\begin{theorem}
	\label{thm:mb-reduced}
	If a move-based splitting function is defined by~\eqref{eq:mrsubmod} with $m_r = m_{r-2}$ and $\vm_e$ submodular, then it can be modeled by an NMC-gadget. 
\end{theorem}
\begin{proof}
To prove the result, we define the NMC-basis-gadget (Figure~\ref{fig:nmc}), a simple gadget parameterized by a positive integer $b$. Just as we combined CB-gadgets to model all submodular cardinality-based two-way splitting functions, we will show how to combine NMC-basis-gadgets to model submodular move-based functions.
\begin{figure}[h]
\begin{minipage}{0.45\textwidth}
		\includegraphics[width=\linewidth]{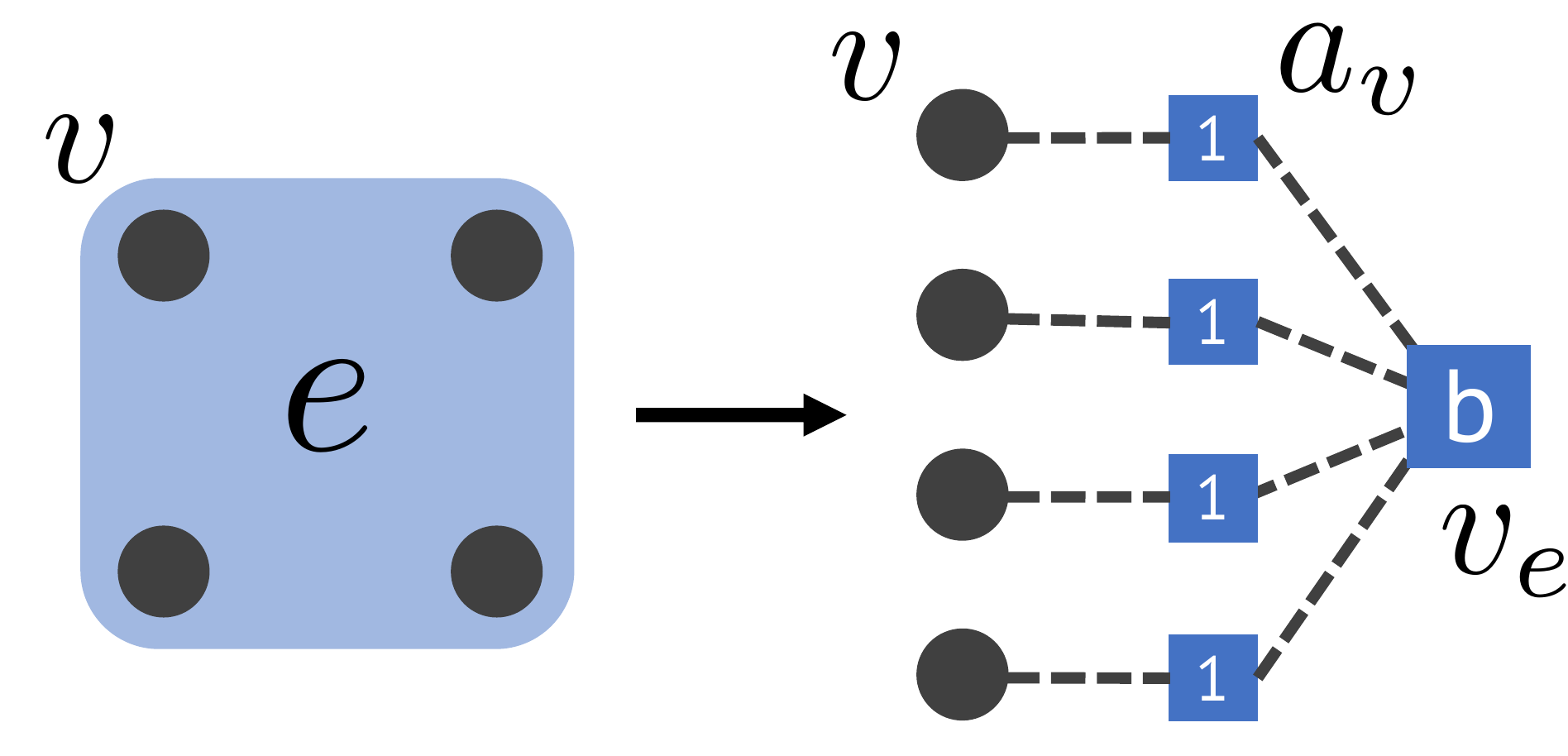}
	\end{minipage}\hfill
	\begin{minipage}{0.55\textwidth}		
		
		\begin{itemize}
\item For each $v \in e$, introduce an auxiliary node $a_v$ and an edge $(v,a_v)$. Set the node-weight of $a_v$ to be 1, and the node weight of $v$ to be $\infty$.
\item Introduce one more auxiliary node $v_e$, and attach $v_e$ to $a_v$ with an edge for all $v \in e$. Set the node-weight of $v_e$ to be $b \in \mathbb{N}$.
\end{itemize}
	\end{minipage}
	\vspace{-.25cm}
\caption{} \label{fig:nmc}
\end{figure}

Given such a gadget and a set of nodes $S \subset e$, there are only two ways to delete auxiliary vertices in order to partition $V'$ in such a way that the cluster with the most nodes from $e$ contains all of $e\backslash S$. 
The first option is to delete node $v_e$, which has a penalty of $b$ and separates all nodes in $e$ from each other. The second option is to delete $a_v$ for each $v \in S$, which has a total penalty of $|S|$. Since no other options are possible, the splitting function for the NMC-basis-gadget is
\begin{equation}
\label{eq:nmcbgadget}
\hat{\vm}_b(S) = \min \{|S|, b\}\,.
\end{equation}

To model an $r$-node hyperedge, we introduce a total of $r-1$ NMC-basis-gadgets, one for each value of $b \in \{1, 2, \hdots r-1\}$. Let $c_i \geq 0$ be a weight we use to scale the NMC-basis-gadget with parameter $b = i$. Combining these leads to a larger NMC-gadget with $(r^2 - 1)$ auxiliary nodes. Let $\hm_i$ denote the penalty that the splitting function of the combined gadget assigns to a partition of $e$ with $i$ nodes outside the largest cluster. The splitting function penalties for the combined gadget can be described by a linear system that is nearly identical to the system we used to model cardinality-based two-way splitting functions in Section~\ref{sec:reduction}:
\begin{equation}
\label{matrixeqmulti}
\begin{bmatrix}
1 & 1 & 1 & \cdots & 1 \\
1 & 2 & 2 &  \cdots & 2\\
1 & 2 & 3 &  \cdots & 3\\
\vdots & \vdots & \vdots  & \ddots & \vdots\\
1 & 2 & 3  &\cdots & r-1 \\
\end{bmatrix}
\begin{bmatrix}
c_1 \\ c_2 \\ c_3 \\\vdots \\ c_{r-1}
\end{bmatrix}
= 
\begin{bmatrix}
\hm_1 \\ \hm_2 \\ \hm_3  \\\vdots \\ \hm_{r-1}
\end{bmatrix}.
\end{equation}
By inverting~\eqref{matrixeqmulti} and enforcing $c_i \geq 0$ for $i = 1, 2, \hdots, r-1$, we find that this approach will enable us to model any move-based splitting function satisfying properties~\eqref{ps1},~\eqref{ps2}, and~\eqref{eq:mon}.
\end{proof}
We leave it as an open question to explore whether NMC-reducibility can be shown for other choices of $m_r \neq m_{r-2}$. For example, if we set $m_r = 0$, then submodularity for $\vm_e$ is equivalent to submodularity for an asymmetric cardinality-based splitting function (see Lemma~\ref{lem:asub}). The submodular region for $m_r = 0$ in fact \emph{contains} the submodular region for $m_r = m_{r-2}$, defined by inequalities~\eqref{ps1},~\eqref{ps2}, and~\eqref{eq:mon}. However, it is unclear whether all submodular functions under the choice $m_r = 0$ can be modeled by NMC-gadgets. Furthermore, with this approach we lose the monotonocity constraint~\eqref{eq:mon}. This constraint is very natural for move-based splitting functions, and we do not know how to model any move-based functions whose splitting penalties \emph{decrease} when we separate \emph{more} nodes from the largest cluster of $e$. Another related question is whether any of these definitions of submodularity is \emph{necessary} for NMC-reduction. Although we were able to show that all hypergraph $s$-$t$ cut gadgets have submodular splitting functions (Corollary~\ref{cor:submodgadget}), there is no clear analogous result for NMC-gadgets, even in the case of move-based functions.

\subsection{Hardness of Approximation for Rainbow Splits}
As is the case for cardinality-based \hc{}, there are special cases of move-based \hmc{} that are inherently harder to solve or approximate than problems in the submodular region. In this section we prove that this is true even in the case of 3-uniform hypergraphs. In particular, we prove a hardness result for the problem under the rainbow split splitting function (see Table~\ref{tab:multiwaysplits}), which assigns a penalty of 1 if every node in a hyperedge is assigned to a different cluster, but otherwise assigns no penalty.
\begin{definition}
\turh{} is the special case of \ghmc{} where the rainbow split splitting function is applied to all hyperedges.
\end{definition}

\turh{} and the hardness result we will show is closely related to the \textsc{Hypergraph Rainbow-Avoiding Labeling} problem~\cite{mirzakhani2015}. An instance of the labeling problem is given by a hypergraph and a list of allowable color labels for each node. The goal is to assign colors to nodes in a way that minimizes the number of \emph{rainbow-labeled} hyperedges, i.e., hyperedges whose nodes are all assigned a different color. Mirzakhani and Vondr\'{a}k showed that for 3-uniform hypergraphs, it is NP-hard to detect whether there is a color assignment with no rainbow-labeled hyperedges~\cite{mirzakhani2015}.
\turh{} is a restriction of this problem, in which terminal nodes have only one allowed color, and all other nodes can be assigned to any color. Since \textsc{Hypergraph Rainbow-Avoiding Labeling} is more general than our problem, the existing hardness result does not directly apply. However, the proof technique
can be directly adapted to show the same hardness result for \turh{}. For completeness we include a full detailed proof.
\begin{theorem}
	\label{thm:rainbow}
	\turh{} (on 3-uniform hypergraphs) is NP-hard to approximate to within any multiplicative factor.
\end{theorem}
\begin{proof}
	Similar to our hardness result for 5-uniform cardinality-based \hc{}, we begin with an instance of \textsc{Mon-NAE-3SAT}: we are given a set of literals $\{x_1, x_2, \hdots , x_n\}$ (but not their negations), and clauses in conjunctive normal form (e.g., $(x_i \lor x_j \lor x_k)$). The goal is to assign literals to \emph{true} and \emph{false} so that each clause contains at least one true and at least one false variable. 
	
	For the reduction we will introduce three terminal nodes, a set of hyperedges and nodes associated with each literal $x_i$, and hyperedges that encode dependencies among literals due to clauses in the \textsc{Mon-NAE-3SAT} instance. To highlight a relationship with the hypergraph labeling problem~\cite{mirzakhani2015} (as well as a connection to Sperner's Lemma), we associate terminal nodes and their respective clusters with a color in $\{ \text{red}, \text{blue}, \text{green}\}$. The goal is to assign all other nodes to a color (cluster), though unlike the \textsc{Hypergraph Rainbow-Avoiding Labeling} problem, there are no prior restrictions on cluster assignment for non-terminal nodes.
	
	\paragraph{The Sperner Gadget}	Let $t_{r}$, $t_b$, and $t_g$ be the red, blue, and green terminal nodes in our instance of \turh{}. For each literal $x_i$, we introduce three new nodes $\{i_1, i_2, i_3\}$, and six hyperedges: $(i_1, t_g, t_r)$, $(i_1, t_r, i_2)$, $(i_2, t_r, t_b)$, $(i_2, t_b, i_3)$, $(i_3, t_b, t_g)$, $(i_3, t_g, i_1)$. We refer to this as the $x_i$-Sperner-Gadget. Figure~\ref{fig:spernergadget} provides a visualization of this gadget, which highlights its relationship to Sperner's Lemma. 
	\begin{figure}[h]
		\caption{The $x_i$-Sperner-gadget.}	\label{fig:spernergadget}
		\begin{minipage}{0.35\textwidth}
			\includegraphics[width=\linewidth]{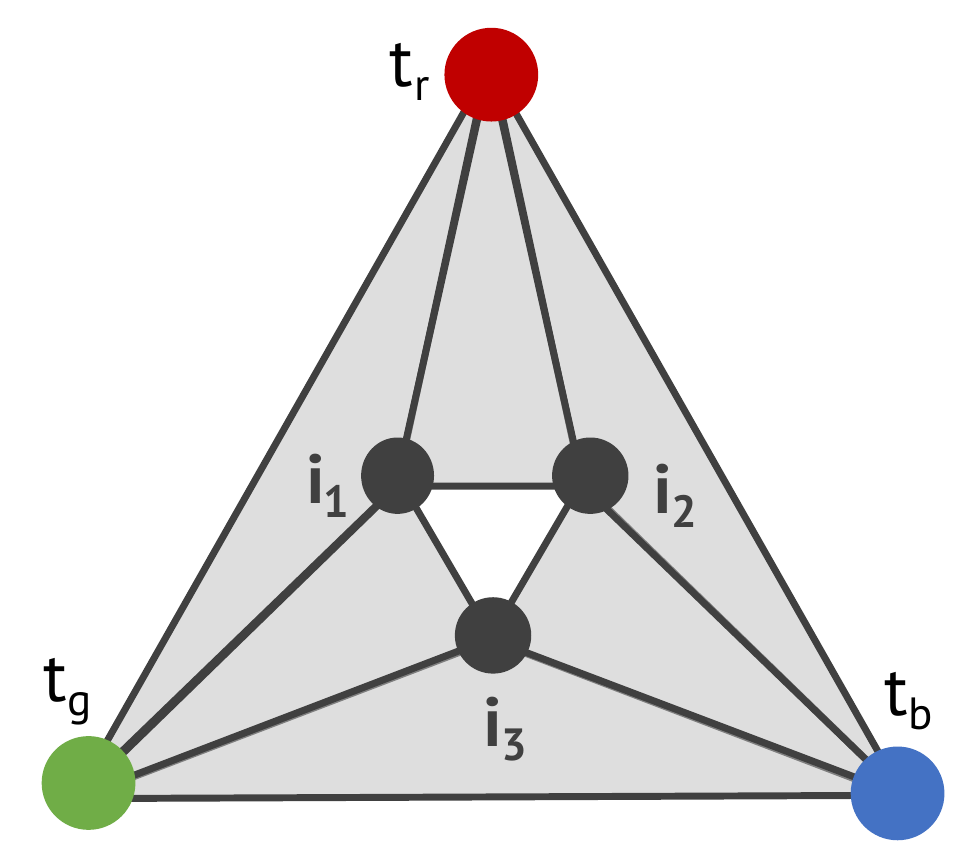}
		\end{minipage}\hfill
		\begin{minipage}{0.65\textwidth}
			For each literal $x_i$ we define a gadget with
			
			\begin{itemize}
				\item three new nodes: $\{i_1, i_2, i_3\}$
				\item six new hyperedges: $(i_1, t_g, t_r)$, $(i_1, t_r, i_2)$, $(i_2, t_r, t_b)$, $(i_2, t_b, i_3)$, $(i_3, t_b, t_g)$, $(i_3, t_g, i_1)$.\\
			\end{itemize}
		
		Zero-penalty node colorings correspond to \emph{true-false} assignments for the $x_i$ literal:
			\begin{itemize}
				\item $x_i$ = \emph{true} $\longleftrightarrow$ $(i_1, i_2, i_3) = (\text{red}, \text{blue}, \text{green})$
				\item $x_i$ = \emph{false} $\longleftrightarrow$  $(i_1, i_2, i_3) =(\text{green},\text{red}, \text{blue})$.
			\end{itemize}
		\end{minipage}
	\end{figure}

	The nodes making up the gadget define a triangulation of the simplex. All of the cells in the triangulation correspond to hyperedges in the instance of \turh{} we construct, \emph{except} the cell delimited by nodes $\{i_1, i_2, i_3\}$. By Sperner's Lemma (or by simply checking nine possible options), one of the cells in the triangulation must be rainbow-colored. Therefore, in order to avoid a non-zero splitting penalty, $\{i_1, i_2, i_3\}$ must be that rainbow-colored cell. There are exactly two ways to do this, and we associate each way with a \emph{true-false} assignment for the literal $x_i$:
	\begin{enumerate}
		\item Assign $i_1 \leftarrow \text{red}$,  $i_2 \leftarrow \text{blue}$, and $i_3 \leftarrow \text{green}$. 
		Call this the \emph{true} color assignment for the $x_i$-gadget. 
		\item Assign $i_1 \leftarrow \text{green}$,  $i_2 \leftarrow \text{red}$, and $i_3 \leftarrow \text{blue}$. 
		Call this the \emph{false} color assignment for the $x_i$-gadget. 
	\end{enumerate}
	Our Sperner gadget construction is the main distinction between our proof and the hardness result of Mirzakhani and Vondr\'{a}k~\cite{mirzakhani2015} for hypergraph labeling. These authors construct a similar gadget related to Sperner triangles but place hard restrictions on which colors can be assigned to each of $i_1$, $i_2$, and $i_3$. Our gadget uses more hyperedges but models the \textsc{Mon-NAE-3SAT} problem in essentially the same way, without placing hard restrictions on to which clusters each non-terminal node can belong.

	\paragraph{Encoding Clauses as Hyperedges} 
	Next, we encode dependencies among literals: for each clause $(x_i \lor x_j \lor x_k)$, we also add the hyperedge $(i_1, j_2, k_3)$. Consider how different \emph{true} or \emph{false} color assignments for the Sperner-gadgets for $x_i$, $x_j$, and $x_k$ will affect the colors of nodes in hyperedge $(i_1, j_2, k_3)$. We list the color of each node for the respective \emph{true} or \emph{false} color assignment of the Sperner-gadget it belongs to:
	\[
	\centering 
	\begin{tabular}{lll}
		node & \emph{true} assignment & \emph{false} assignment\\
		\midrule
		$i_1$ & red & green \\
		$j_2$ & blue& red \\
		$k_3$ & green & blue\\
		\bottomrule
	\end{tabular}
	\]
	
	Noting these possibilities, we can list all possible triplets of \emph{true-false} (T or F) assignments for the $x_i$, $x_j$, and $x_k$ Sperner gadgets, along with the resulting colors of nodes in the hyperedge $(i_1, j_2, k_3)$ (R = red, B = blue, G = green):
	\[
		\centering 
	\begin{tabular}{lllllllll}
		$(x_i, x_j, x_k)$ & TTT & FTT & TFT & FFT & TTF & FTF & TFF & FFF\\
		\midrule
		$(i_1, j_2, k_3)$ & RBG & GBG & RRG & GRG & RBB & GBB & RRG & GRB\,.
	\end{tabular}
	\]
	\paragraph{Problem Equivalence}
	Observe from the above table that, conditioned on there being no rainbow-colored hyperedges in the Sperner gadgets for literals $\{x_i, x_j, x_k\}$, the hyperedge $(i_1, j_2, k_3)$ will be rainbow-colored if and only if all three of these Sperner gadgets have the same \emph{true-false} color assignment. This exactly models the goal of not-all-equal 3SAT: there is a mistake only when the three literals of a clause are all \emph{true} or all \emph{false}. We conclude that there is a solution the \textsc{Mon-NAE-3SAT} instance if and only if there is a zero-penalty solution for the instance of \turh{} we have constructed. Since the former problem is NP-hard, it is also NP-hard to detect zero-penalty solutions for \turh{}, and therefore the optimization version of our problem is hard to approximate to within any multiplicative factor.
\end{proof}

To conclude, there are instances of~\ghmc{} that can be approximated via reduction to \nmc{}, as well as other instances that are NP-hard even to approximate. As was the case for our results on \hstgen{}, our results on \hmc{} lead to several open questions, especially regarding the computational complexity and approximability of the problem under different types of splitting functions. For example, are there other instances of move-based \hmc{} that are inapproximable? More generally, is it possible to obtain new approximation guarantees for other splitting functions (e.g., cluster-based function) by reducing \hmc{} to other other variants of graph multiway cut (e.g., directed multiway cut)? We include a more formal list of open questions such as these in the next section.

\section{Discussion and Open Questions}
\label{sec:discussion}
The first hypergraph generalization of the minimum $s$-$t$ cut problem was introduced nearly fifty years ago by Lawler~\cite{lawler1973}, who gave a polynomial time solution based on a reduction to a graph $s$-$t$ cut problem. 
Since this initial work, however, there has been hardly any consideration of how the complexity or applicability of the problem changes under more general notions of splitting penalties at hyperedges.
Moreover, the existing literature on cut definitions for \emph{other} hypergraph cut problems is both fragmented and sparse. In our work, we have shown that considering broader notions of hypergraph cuts leads to a wealth of new algorithms, complexity results, data modeling techniques, and open questions. In particular, we motivated the cardinality-based splitting function, which is a natural penalty function to use in hypergraph cut applications, and is implicitly related to a number of hypergraph cut problems proposed previously.
Our techniques provide a unified framework for reducing hypergraphs to graphs when solving minimum $s$-$t$ cut problems, and we are the first to identify parameter regimes of the hypergraph $s$-$t$ cut problem that are NP-hard to optimize. Many of these results can also be extended to the hypergraph multiway cut objective; we have identified parameter regimes for the multiway case which come with approximation guarantees, and an example for which obtaining approximations is intractable. All of these contributions are accompanied by specific and well-defined open questions for future work. We end with a summary of open questions for \hstgen{} and \ghmc{}.

\subsubsection*{Open Questions on \hstgen{}}
Sections~\ref{sec:positive} and~\ref{sec:tractability} lead to several open questions related to the tractability of the \hstgen{} problem for different types of splitting functions.
\begin{enumerate}
	\item (Conjecture~\ref{con:submod}) Is it possible to model every submodular splitting function using some hypergraph $s$-$t$ cut gadget? We demonstrated in Section~\ref{sec:submodular} that a large number of submodular splitting functions on 4-node hyperedges can be modeled using a gadget involving several different edge parameters, though we do not know of a strategy that will enable us to model \emph{all} 4-node submodular splitting functions, let alone general $k$-node submodular splitting functions.
	\item (Definition~\ref{def:nesst}, Section~\ref{sec:openstcut}) Is \nesst{} (4-uniform cardinality-based \hc{} with $w_1 = 1$, $w_2 \rightarrow \infty$) solvable in polynomial time, or is it NP-hard? The related \dhstgen{} problem admits a trivial solution, but the complexity of \nesst{} is still unknown. 
	\item (Figure~\ref{fig:tractability}, Section~\ref{sec:openstcut}) More generally, what is the computational complexity for cardinality-based \hc{} in regions of the tractability diagrams that are neither in the submodular regime nor the identified NP-hard regime?
%
%
	\item What other special subclasses of submodular functions can be modeled by hyperedge $s$-$t$ cut gadgets? Most importantly, which subclasses do not ever require introducing an exponential number of auxiliary nodes?
	\item Are there classes of non-submodular splitting functions (other than the trivial~\dhstgen{} problem) for which the \hstgen{} problem is polynomial-time solvable even though it is not graph reducible?
\end{enumerate}

\subsubsection*{Open Questions on \ghmc{}}
\ghmc{} is more challenging than \hstgen{}, since the problem is NP-hard even for the graph version.
Our positive results showed how a large class of move-based \hmc{} problems can be reduced to \nmc{}
(and thus are approximable),
but we also found an instance of the problem that is NP-hard to approximate. We consider five specific open questions around which to extend our current results.
\begin{enumerate}
	\setcounter{enumi}{5}
	\item Theorem~\ref{thm:mb-reduced} shows that when $m_r = m_{r-2}$ (see Lemma~\ref{lem:multisub}), \emph{submodular} move-based splitting functions are NMC-modelable. Are NMC-reducibility results possible for submodular regions corresponding to other choices of $m_r$?
	\item Related to the above, is there any notion of submodularity (i.e., any choice of $m_r$) for which we can prove that submodularity is \emph{necessary} for modeling move-based splitting functions with NMC-gadgets?
	\item What is the computational complexity for move-based 3-uniform \hmc{} with penalty weights $m_1 = 1$ and $m_2 = 0$? The submodular region for 3-uniform move-based \hmc{} is directly related to the submodular region for 5-uniform cardinality-based \hc{}, under the relation $m_1 \leftrightarrow w_1$ and $m_2 \leftrightarrow w_2$. We proved that the latter problem is hard to approximate when $w_1 = 1$ and $w_2 = 0$ (Theorem~\ref{thm:5card}), but complexity of the analogous parameter regime for move-based \hmc{} is unknown. 
	\item Cluster-based \hmc{} with penalties $h_i = i$ for each $i \in \{2,3, \hdots, k\}$ (where $k$ is the number of nodes in a hyperedge) is equivalent to the hypergraph multiway partition problem, which can be approximated to within a factor $\frac{4}{3}$~\cite{ene-hmp-2014}. Are there other regimes of cluster-based \hmc{} for which we can obtain approximation guarantees?
	\item We considered only reductions from \ghmc{} to \nmc{}, but there are other generalizations of graph multiway cut with approximation guarantees~\cite{chekuri2011,GARG200449}. Is is possible to define a notion of a directed multiway cut gadget, and obtain approximation guarantees for variants of \ghmc{} via reduction to directed multiway cut? Similarly, can we obtain approximations for different multiway splitting functions via reduction to an instance of submodular multiway partition?\newline
\end{enumerate}

While all of our work is directly related to the hypergraph minimum $s$-$t$ cut problem, our results draw from a wide variety of techniques in combinatorial optimization, graph theory, machine learning, theoretical computer science, and many other disciplines. We hope that these results, along with these clear directions for future work, will continue to stimulate progress on new theoretical results and applications for general hypergraph cut problems at the intersection of these disciplines.

\bibliographystyle{siamplain}
\bibliography{main-hypercuts}
\end{document}